\documentclass[12pt]{article}

\usepackage[margin=1.25in]{geometry}
\usepackage{setspace}
\usepackage[T1]{fontenc}
\usepackage{newtxtext}
\usepackage{mathpazo}
\usepackage{microtype}
\usepackage{enumitem}

\usepackage{amsmath,amssymb,amsthm,mathtools,bm,mathrsfs}
\usepackage{dsfont}
\usepackage{centernot}

\newcommand{\tightdisplay}[1]{%
{\setlength{\abovedisplayskip}{5pt}%
 \setlength{\belowdisplayskip}{5pt}%
 \setlength{\abovedisplayshortskip}{3pt}%
 \setlength{\belowdisplayshortskip}{3pt}%
\[
#1
\]}%
\par\noindent\ignorespaces
}

\newcommand{\tightalign}[1]{%
{\setlength{\abovedisplayskip}{5pt}%
 \setlength{\belowdisplayskip}{5pt}%
 \setlength{\abovedisplayshortskip}{3pt}%
 \setlength{\belowdisplayshortskip}{3pt}%
 \setlength{\jot}{2pt}%
\begin{align*}
#1
\end{align*}}%
\par\noindent\ignorespaces
}

\usepackage{graphicx}
\usepackage{booktabs,threeparttable}
\usepackage{caption}
\usepackage{subcaption}
\usepackage{algorithm}
\usepackage{algpseudocode}
\usepackage{float}

\usepackage{tikz}
\usetikzlibrary{trees, positioning}
\usepackage{forest}

\usepackage{xcolor}

\definecolor{treatgreen}{RGB}{200, 230, 201}
\definecolor{notreatred}{RGB}{255, 205, 210}
\definecolor{mycitecolor}{RGB}{0,102,204}
\definecolor{myurlcolor}{RGB}{0,102,102}
\definecolor{mylinkcolor}{RGB}{0,0,0}

\usepackage[round,authoryear]{natbib}

\usepackage[
    colorlinks = true,
    linkcolor  = mycitecolor,
    citecolor  = mycitecolor,
    urlcolor   = myurlcolor
]{hyperref}

\newcommand{\mycomment}[1]{}
\newcommand{\indep}{\perp\!\!\!\perp}

\newcommand{\R}{\mathbb{R}}

\newcommand{\F}{\mathcal{F}}

\newcommand{\D}{\mathcal{D}}
\newcommand{\G}{\mathcal{G}}
\newcommand{\Exp}{\mathbb{E}}
\newcommand{\Prob}{\mathbb{P}}
\newcommand{\Var}{\mathbb{V}}

\DeclareMathOperator{\sign}{sign}

\forestset{
  ewmtree/.style = {
    for tree = {
      draw,
      rounded corners = 4pt,
      align           = center,
      font            = \small,
      edge            = {-latex, thick},
      l sep           = 1.5cm,
      s sep           = 0.5cm,
      inner sep       = 5pt,
      minimum width   = 2.6cm,
      minimum height  = 0.7cm,
    }
  }
}

\newtheoremstyle{assumpstyle}
  {\topsep}
  {\topsep}
  {\normalfont}
  {}
  {\bfseries}
  {}
  {\newline}
  {}

\theoremstyle{plain}
\newtheorem{assumption}{Assumption}
\newtheorem{theorem}{Theorem}
\newtheorem{lemma}{Lemma}
\newtheorem{definition}{Definition}
\newtheorem{proposition}{Proposition}
\newtheorem{remark}{Remark}
\newtheorem{corollary}{Corollary}
\newtheorem{example}{Example}

\makeatletter
\renewenvironment{proof}[1][\proofname]{%
  \par\pushQED{\qed}%
  \normalfont
  \topsep6\p@\@plus6\p@\relax
  \trivlist
  \item[\hskip\labelsep{\normalfont\scshape #1\@addpunct{.}}]\ignorespaces
}{%
  \popQED\endtrivlist\@endpefalse
}
\makeatother

\title{\Large Nonparametric Bayesian Policy Learning 
\thanks{I am grateful to Chris Walker for his unwavering guidance and support throughout this project and for many insightful discussions. I also benefited from helpful comments from Matt Masten, Arnaud Maurel, Adam Rosen, and participants at the Duke Econometrics Seminar. I am especially grateful to my cat, Coconut, for his unwavering emotional support. All remaining errors are my own.}}
\author{%
  Haonan Ye\thanks{Department of Economics, Duke University. Email: \texttt{haonan.ye@duke.edu}}
}
\date{\today}

\begin{document}

\maketitle
\thispagestyle{empty}

\begin{abstract}
\noindent

I propose \textit{Nonparametric Bayesian Policy Learning} (NBPL) as a framework for uncertainty-aware treatment choice. 
The key observation is that welfare is fully determined by a reduced-form distribution, so uncertainty about optimal policies entirely reflects uncertainty about this distribution.
NBPL places a Dirichlet process prior on the reduced-form distribution and uses the resulting posterior for both traditional policy choice and inference on optimal welfare and optimal treatment assignments.
NBPL is computationally tractable: the default Bayesian-bootstrap implementation requires only exponential reweighting of the observations.
I establish two theoretical properties. First, posterior welfare regret converges at the minimax-optimal rate, providing a novel policy-relevant analogue of posterior contraction rates. Second, posterior model selection across
policy classes is consistent.
I relate NBPL to existing policy learning approaches and illustrate using two empirical applications: the JTPA experiment and the bednet subsidy experiment. In both applications, decision-tree rules tend to yield higher optimal welfare than linear rules.

\end{abstract}


\clearpage
\setcounter{page}{1}    


\section{Introduction}

\subsection{Overview}

In many economic settings, policymakers must allocate scarce interventions across heterogeneous populations.
This motivates an important literature on statistical treatment rules, or policy learning, pioneered by \citet{manski2004statistical}, which studies how to design treatment rules using observable individual characteristics to maximize social welfare;\footnote{Throughout, I use the terms ``treatment rule'', ``treatment assignments'', and ``policy'' interchangeably.} 
see, for example, \citet{hirano2009asymptotics}, \cite{stoye2009minimax}, \cite{chamberlain2011bayesian}, \cite{bhattacharya2012inferring}, \citet{kitagawa2018should}, and \citet{athey2021policy}, among many others.
A growing empirical literature shows that such individualized rules can deliver substantial welfare gains across a wide range of settings, including labor market and education programs \citep{athey2025machine, goller2025active}, anti-poverty interventions in developing countries \citep{haushofer2025targeting, higbee2025policy}, and health and precision medicine \citep{kosorok2019precision, inoue2023machine, athey2025targeted}.

This paper proposes the \textit{Nonparametric Bayesian Policy Learning} (NBPL) framework. The key idea is that welfare is fully determined by a reduced-form data distribution, so uncertainty over optimal policies entirely reflects uncertainty about this distribution. 
This blends seamlessly with a Bayesian approach: the researcher places a nonparametric prior on the reduced-form distribution and uses the resulting posterior for both traditional policy choice analysis and inference on optimal welfare and optimal treatment assignments. 
For traditional policy choice, my framework leads to a class of prior-regularized treatment rules, which I call 
\textit{Posterior Expected Welfare Maximization} (PEWM).
Simultaneously, because the posterior induces distributions over both optimal welfare and optimal treatment assignments, one can apply Bayesian statistical decision theory to obtain optimal point estimates or uncertainty quantification.
The core theoretical contribution is a minimax-optimal convergence rate for posterior welfare regret (Theorem~\ref{Theorem1}). It provides a novel policy-relevant analogue of posterior contraction rates, a notion that has played an important role in the theoretical statistics literature since the seminal work of \cite{ghosal2000convergence} and \cite{shen2001rates}; see also \citet{ghosal2017fundamentals} and the references therein.
As a further distinctive feature, NBPL provides a consistent posterior model-selection procedure for comparing competing, possibly non-nested policy classes (Theorem~\ref{Theorem2}).

NBPL has several appealing features.
First, it incorporates parameter uncertainty directly into the decision problem through the posterior; ignoring such uncertainty can lead to suboptimal decisions \citep[e.g.,][]{dehejia2005program, christensen2026optimal, moon2026optimal}.
Second, to the best of my knowledge, NBPL provides the first unified inferential framework for optimal welfare, optimal treatment assignments, and policy-class selection. 
Policy-class selection is central to policy learning because it determines the feasible set over which welfare is optimized, thereby shaping both optimal treatment assignments and the resulting policy recommendations
\citep[e.g.,][]{mbakop2021model}.
NBPL provides a consistent posterior diagnostic for comparing policy classes (Theorem~\ref{Theorem2}). As a byproduct, the same procedure can be used to rank fixed policies \citep[e.g.,][]{kasy2016partial,yang2022offline}.
Third, NBPL is flexible: the prior for the reduced-form data distribution is a Dirichlet process \citep{ferguson1973bayesian}, which offers a level of flexibility similar to that of frequentist nonparametric approaches to policy learning \citep[e.g., the important Empirical Welfare Maximization (EWM) framework of][]{kitagawa2018should}. Fourth, posterior computation is highly tractable. The default Bayesian-bootstrap implementation \citep{rubin1981bayesian} requires only exponential reweighting of the observations, can reuse off-the-shelf policy learning algorithms, and is naturally parallelizable.

I establish three theoretical properties of NBPL. First, posterior welfare regret converges to zero at the minimax-optimal rate (Theorem~\ref{Theorem1}), matching the EWM regret rate of \citet{kitagawa2018should}. Thus, incorporating parameter uncertainty does not slow the rate of learning. To the best of my knowledge, this provides the first policy-relevant analogue of posterior contraction rates, a central notion in the theoretical statistics literature over the past two decades; see \cite{ghosal2017fundamentals} and the references therein.
Theorems~\ref{Theorem:C1}--\ref{Theorem:C3} in Supplemental Appendix~\ref{supplement:extensions} extend Theorem~\ref{Theorem1} to settings with multi-valued treatments, demeaned outcomes, and capacity constraints, respectively.
Second, I prove pointwise consistency of the posterior model-selection procedure (Theorem~\ref{Theorem2}). This result justifies using posterior probabilities to compare competing, possibly non-nested policy classes.
Third, I relate NBPL to existing policy learning approaches 
\citep{manski2004statistical, chamberlain2011bayesian, kitagawa2018should}.
I show that EWM arises as the limiting case of PEWM as the prior mass vanishes (Proposition~\ref{Proposition1}). I then establish uniform minimax optimality for PEWM under the welfare-regret criterion of \citet{kitagawa2018should}, thereby recovering the uniform minimax optimality of EWM as a special case (Theorem~\ref{Theorem3}).
I further prove an asymptotic equivalence between NBPL and these approaches under discrete covariates and unconstrained policy classes (Proposition~\ref{Proposition2}).

I illustrate NBPL in two empirical applications: the Job Training Partnership Act (JTPA) experiment studied by \citet{kitagawa2018should} and the anti-malaria bednet subsidy experiment analyzed by
\citet{bhattacharya2012inferring}. 
Both applications yield two main findings. First, EWM welfare estimates tend to lie below the posterior medians
of optimal welfare under NBPL. Second, the
posterior model-selection procedure provides strong evidence that decision-tree rules attain higher optimal welfare than linear rules. 
A Monte Carlo simulation calibrated to the data from \citet{bhattacharya2012inferring} in Appendix~\ref{Appendix:simulation} shows that NBPL and EWM perform comparably under the welfare-regret criterion in
Section~\ref{theory:regret}, and that the model-selection procedure consistently favors the decision-tree class.

\subsection{Literature}

This paper connects to three literatures in economics. The first is the literature on statistical treatment rules, or policy learning. 
A leading approach is Empirical Welfare Maximization
\citep[EWM;][]{kitagawa2018should}, which selects treatment rules by plugging in the empirical distribution in place of the true population distribution. In contrast, NBPL explicitly incorporates parameter uncertainty by placing a Dirichlet process prior on the reduced-form data distribution. An important property of EWM is that it attains the uniformly minimax-optimal convergence rate for welfare regret. I show that posterior welfare regret under NBPL contracts at the same rate. I further establish uniform minimax optimality of PEWM under the welfare-regret criterion of \citet{kitagawa2018should}, thereby recovering the uniform minimax optimality of EWM as a special case.
Within this literature, NBPL is also closely related in spirit to the limited-information Bayesian framework of \citet{chamberlain2011bayesian}. Chamberlain places a structural prior on latent potential-outcome distributions in a setting restricted to discrete covariates and outcomes.
In contrast, NBPL takes a reduced-form route: it places the prior directly on the reduced-form distribution, accommodates continuous variables and constrained policy classes, and provides posterior regret guarantees.
NBPL is also related to the broader literature on uncertainty-aware decision-making \cite[e.g.,][]{dehejia2005program, chernozhukov2025policy, christensen2026optimal, moon2026optimal}.
The policy-learning literature is extensive; other contributions include
\citet{hirano2009asymptotics}, \citet{stoye2009minimax}, \cite{bhattacharya2012inferring}, \cite{kallus2018confounding, kallus2018policy}, \citet{athey2021policy}, \citet{kitagawa2021equality}, \citet{mbakop2021model}, \citet{zhou2023offline}, \citet{adjaho2025externally}, \citet{kitagawa2025leave}, \citet{viviano2025policy}, \citet{ida2026choosing}, and \citet{montiel2026decision}.

Another related literature develops Bayesian inference through reduced-form parameters.
The motivation is that many economic models admit a known mapping from an identifiable reduced-form parameter to a structural parameter of interest. 
One can therefore place a prior on the reduced-form parameter and use the induced posterior for inference on the structural parameter.
Given this insight, this literature then divides into two strands according to whether the structural parameter is point-identified. 
NBPL spans both cases: optimal welfare is point-identified, whereas optimal treatment assignments may be partially identified.
For point-identified structural parameters, related reduced-form Bayesian approaches include \citet{chamberlain2003nonparametric} and \citet{walker2026semiparametric}. NBPL is also connected to the loss-likelihood bootstrap of \citet{lyddon2019general}, with the loss given by negative welfare. 
For partially identified structural parameters, related approaches include \citet{kline2016bayesian} and \citet{florens2021revisiting}. I draw on this
perspective for inference on optimal treatment assignments.
\citet{florens2021revisiting} are particularly close in spirit because they also place a Dirichlet process prior on the population distribution, although they do not study treatment choice. 
Additional related contributions include \citet{norets2014semiparametric}, \citet{liao2019bayesian}, \citet{ditraglia2025bayesian}, \citet{andrews2026bayesian}, \citet{li2026empirical}, and \citet{tamer2026nonparametric}.

A final related literature studies inference for optimal welfare and optimal treatment assignments. Inference on optimal welfare is important because it quantifies the maximum gains from policy learning, but it is challenging because optimal welfare is an irregular parameter; see, for example, the impossibility results in \citet{hirano2012impossibility}. \citet{luedtke2016statistical}, \citet{ponomarev2025lower}, and \citet{whitehouse2025inference} develop confidence intervals, lower confidence bands, and smoothing-based inference procedures for this parameter, respectively.
Inference on optimal treatment assignments is also crucial because it quantifies the evidence in favor of treating units selected by an optimal rule, but it has received comparatively less attention; two frequentist examples are \citet{rai2018statistical} and \citet{armstrong2023inference}.
Relatedly, \citet{kitagawa2023stochastic} develop a quasi-Bayesian inference framework for linear treatment rules, but place the prior directly on the space of treatment rules rather than on the data-generating distribution.
The NBPL framework provides a Bayesian complement to this literature.

\subsection{Outline}

The paper is organized as follows. Section~\ref{sec:NBPL} introduces the NBPL framework and its implementation, Section~\ref{sec:empirical} presents two empirical applications, Section~\ref{sec:theory} develops the theoretical results, and Section~\ref{sec:conclusion} concludes. Appendix~\ref{Appendix:simulation} reports a Monte Carlo simulation calibrated to the data of \citet{bhattacharya2012inferring}, and Appendix~\ref{Appendix:proofs} contains the main proofs. 
Supplemental Appendices~\ref{supplement:extensions}--\ref{supplement:proofs of lemmas} contain extensions and proofs of auxiliary lemmas.


\section{Bayesian Inference Framework}
\label{sec:NBPL}

\subsection{Setup: Welfare Maximization}
\label{NBPL:setup}

A decision-maker (DM) seeks to choose a treatment rule in order to maximize expected welfare in a target population. Individual $i$ in the target population is characterized by the random vector $(Y_i(1), Y_i(0), X_i) \sim Q_0^{\star}$, where $Y_i(1), Y_i(0) \in \mathcal{Y} \subseteq \R$ denote the potential outcomes under treatment and control, respectively, and $X_i \in \mathcal{X} \subseteq \R^{d_x}$ is a vector of baseline covariates. 
The DM's problem is to choose a treatment rule $G \subseteq \mathcal{X}$ that maximizes utilitarian welfare,\footnote{Hereafter, I normalize welfare relative to the status quo of assigning no treatment, so that  $W(Q_0^{\star}; G) = \Exp_{Q_0^{\star}}[(Y(1)-Y(0)) \mathbf{1}\{X \in G \}]$. This normalization leaves the optimal treatment assignment unchanged, but welfare values should henceforth be interpreted as gains relative to the status quo.}
\tightdisplay{
    W(Q_0^{\star}; G) \coloneqq \Exp_{Q_0^{\star}}[Y(1) \mathbf{1}\{X \in G \} + Y(0) \mathbf{1}\{X \notin G \}].
}
The DM chooses a treatment rule from a feasible \emph{policy class} $\mathcal{G}\subseteq 2^{\mathcal X}$, so treatment assignment may depend only on observable individual characteristics $X$.

Suppose the DM observes a random sample $\{ \widetilde{D}_i\}_{i=1}^n$ drawn from an experimental population. 
Individual $i$ in this population is characterized by the random vector $(Y_i(1), Y_i(0), X_i) \sim P_0^{\star}$. 
The observed data are $\widetilde{D}_i \coloneqq (Y_i, T_i, X_i)$, where $T_i \in \{0,1\}$ denotes the realized treatment assignment and $Y_i = Y_i(1)T_i + Y_i(0)(1-T_i)$ denotes the realized outcome. 
Let $\mathbf P_0^{\star}$ denote the true distribution of $(Y(1), Y(0), T, X)$ and define the propensity score $e(x) \coloneqq \mathbf P_0^{\star}(T =1 \mid X = x)$.

\begin{assumption}
\label{Assumption:DGP}

\begin{enumerate}[label=(\alph*)]
    \item \emph{(External Validity)} The target and  experimental populations share the same distribution of potential outcomes and covariates: $Q_0^{\star} = P_0^{\star}$.
    
    \item \emph{(Unconfoundedness)} $(Y(1), Y(0)) \indep T \mid X$.

    \item \emph{(Strict Overlap)} There exists $\kappa \in (0, 1/2)$ such that the propensity score satisfies $e(x) \in [\kappa, 1-\kappa]$ for $\mathbf P_0^{\star}$-almost every $x \in \mathcal{X}$. The propensity score is known.
    
    \item \emph{(Outcome Moments)} The realized outcome has finite $(2+\delta)$-th moment under $\mathbf P_0^{\star}$:
    $\Exp_{\mathbf P_0^{\star}}|Y|^{2+\delta} < \infty$ for some $\delta > 0$.
    
\end{enumerate}

\end{assumption}

Assumption~\ref{Assumption:DGP} imposes standard conditions on
$\mathbf P_0^{\star}$.
Assumption~\ref{Assumption:DGP}(a) ensures that the experimental data are informative about the target population distribution; in particular, it holds when the sample is drawn directly from the target population.
Assumptions~\ref{Assumption:DGP}(b) and (c) are standard in causal inference.
The assumption of a known propensity score is a standard benchmark in policy learning \citep[e.g.,][]{kitagawa2018should, mbakop2021model, kitagawa2023stochastic} and is not restrictive in many applications: it holds by design in randomized controlled trials, including the two empirical applications considered in this paper \citep{bhattacharya2012inferring, kitagawa2018should}, and can sometimes be reasonable in quasi-experimental settings \citep[e.g.,][]{angrist1990lifetime, angrist1999using, katz2001moving, abdulkadirouglu2011accountability, abdulkadirouglu2022breaking}.
Assumption~\ref{Assumption:DGP}(d) is a technical condition used to control empirical-process terms in the proofs.\footnote{It relaxes the bounded-outcome condition in Assumption~2.1(BO) of \citet{kitagawa2018should} and Assumption~3 of \citet{kitagawa2023stochastic}.}

Under Assumption~\ref{Assumption:DGP}, welfare in the target population under any treatment rule $G$ is point-identified as \citep[see e.g.,][]{kitagawa2018should}:
\tightdisplay{
    W(P_0; G) = \mathbb{E}_{P_0} \left[ \left( \frac{YT}{e(X)} - \frac{Y(1-T)}{1-e(X)} \right) \mathbf{1}\{X \in G\}  \right], \tag{2.1} \label{main:eq1}
}
where $P_0$ is the \textit{reduced-form distribution} of
$D \coloneqq (Z_1,Z_0,X)$ induced by $\mathbf P_0^\star$, with $Z_1\coloneqq YT/e(X)$ and $Z_0\coloneqq Y(1-T)/[1-e(X)]$; see Remark~\ref{Remark1} for further discussion.
The key observation is that, for any given $G$, population welfare depends on $\mathbf P_0^\star$ only through $P_0$.
Hereafter, let $\mathcal D_n \coloneqq \{D_i\}_{i=1}^n$ denote the observed reduced-form sample.

Define the set of optimal treatment rules under $P_0$ by
\tightdisplay{
    \mathcal{G}^{\star}(P_0) \coloneqq \arg\max_{G \in \mathcal{G}} \, W(P_0; G).
    \tag{2.2} \label{main:eq2}
}
An optimal treatment rule $G^\star(P_0)$ is any element of
$\mathcal G^\star(P_0)$. Following \citet{kline2016bayesian}, I treat the set $\mathcal G^\star(P_0)$, rather than a particular selection $G^\star(P_0)$, as the target of inference.
The corresponding optimal welfare is
\tightdisplay{
    W_{\mathcal{G}}^{\star}(P_0) \coloneqq \sup_{G \in \mathcal{G}} \, W(P_0; G).
    \tag{2.3} \label{main:eq3}
}
Supplemental Appendix~\ref{supplement:extensions} discusses four extensions: alternative welfare criteria, multi-valued discrete treatments, demeaned outcomes, and capacity constraints.

\subsection{Nonparametric Bayesian Policy Learning}
\label{NBPL:framework}

\subsubsection{Inference Framework}
\label{NBPL:inference framework}

The key challenge in the welfare maximization problem~\eqref{main:eq2} is that the true reduced-form distribution $P_0$ is unknown to the DM. For any given policy class $\mathcal{G}$, both the optimal welfare and the set of optimal treatment rules, $(W_{\mathcal{G}}^{\star}(P_0), \mathcal G^{\star}(P_0))$, are \textit{deterministic} functions of $P_0$ only. Consequently, uncertainty about these welfare-relevant objects is driven entirely by uncertainty about $P_0$.

To account for this uncertainty, I assume the DM is Bayesian and places a nonparametric prior on the reduced-form distribution: $P \sim \Pi$. Specifically, $\Pi$ is taken to be a Dirichlet process \citep[DP;][]{ferguson1973bayesian}, a standard nonparametric prior on probability distributions.

\begin{definition}[Dirichlet Process] 
\label{def:DP}

A random probability measure $P$ on the measurable space $(\mathcal{D}, \mathscr{D})$\footnote{In the context of NBPL, $(\mathcal{D}, \mathscr{D})$ denotes the measurable space on which the reduced-form data are defined. In particular, $\mathcal{D} = \mathcal{Z}_1 \times \mathcal{Z}_0 \times \mathcal{X}$, where $\mathcal{Z}_1$ and $\mathcal{Z}_0$ are the supports of the reduced-form variables $Z_1$ and $Z_0$, respectively, and $\mathcal{X}$ is the support of the covariates $X$. The $\sigma$-algebra $\mathscr D$ is the corresponding product $\sigma$-algebra. See Remark~\ref{Remark1}.}
is said to follow a Dirichlet process with base measure $\alpha$, denoted $\mathrm{DP}(\alpha)$, if for every finite measurable partition $(A_1, \ldots, A_k)$ of $\mathcal{D}$, 
the random probability vector satisfies $(P(A_1), \ldots, P(A_k)) \sim \mathrm{Dir}(\alpha(A_1), \ldots, \alpha(A_k))$, where $\mathrm{Dir}(\cdot)$ denotes the Dirichlet distribution on the $k$-dimensional unit simplex $\Delta_k = \{x \in \R^k: x_i \geq 0, \sum_{i=1}^k x_i = 1\}$.

\end{definition}

I impose the following mild integrability condition on the base measure $\alpha$, which is required only to control the prior component of the posterior in the proofs.

\begin{assumption}
\label{Assumption:base measure}

The base measure $\alpha$ on $\mathcal D\coloneqq \mathcal Z_1\times\mathcal Z_0\times\mathcal X$ has finite total mass $|\alpha|\coloneqq \alpha(\mathcal D)<\infty$ and satisfies $\int_{\mathcal D}|z_j|\,d\alpha(z_1,z_0,x)<\infty$ for $j=0,1$.
\end{assumption}

The Dirichlet process has two appealing features. 
First, it is flexible: it imposes no parametric restrictions on the underlying distribution and has large support under the weak topology.\footnote{The weak topology is the topology of weak convergence of probability measures. Let $\mathcal P$ denote the space of Borel probability measures on $\mathcal D$, equipped with the weak topology, and let $\mathscr P$ be the corresponding Borel $\sigma$-algebra. The \textit{weak support} of $\mathrm{DP}(\alpha)$ is $\mathrm{supp}_w\big(\mathrm{DP}(\alpha)\big) \coloneqq \{ P \in \mathcal{P} : \mathrm{supp}(P) \subseteq \mathrm{supp}(\alpha) \}$. Thus, if $\alpha$ has full support on $\mathcal{D}$, then $\mathrm{DP}(\alpha)$ assigns positive prior mass to every weak neighborhood of every $P \in \mathcal{P}$.} 
In particular, if the base measure $\alpha$ has full support, then $\mathrm{DP}(\alpha)$ assigns positive prior mass to every weak neighborhood of any probability distribution on $\mathcal D$.
Second, the Dirichlet process is conjugate: if
$P\sim\mathrm{DP}(\alpha)$, then the posterior remains a Dirichlet process, $P \mid \mathcal{D}_n \sim \mathrm{DP}(\alpha + n \mathbb{P}_n)$, where $\mathbb{P}_n \coloneqq n^{-1} \sum_{i=1}^n \delta_{D_i}$ denotes the empirical measure.\footnote{$\delta_x$ denotes the Dirac measure at $x$, which places probability mass one at the point $x$.}
This conjugacy makes posterior computation highly tractable, mitigating the concern that inference with nonparametric priors may be computationally challenging.

The posterior for $P$ induces marginal posterior distributions for optimal welfare and optimal treatment assignments.
The posterior distribution of optimal welfare is a probability measure on $\mathbb R$. Under suitable measurability conditions, it is the pushforward of the posterior for $P$ under the map $P\mapsto W_{\mathcal G}^\star(P)$.\footnote{Given measurable spaces $(\mathcal P,\mathscr P)$ and $(\mathcal T,\mathscr T)$, a measurable map
$T:(\mathcal P,\mathscr P)\to(\mathcal T,\mathscr T)$, and a
probability measure $\mu$ on $(\mathcal P,\mathscr P)$, the
pushforward measure $\mu\circ T^{-1}$ on $(\mathcal T,\mathscr T)$ is defined by $(\mu\circ T^{-1})(A)=\mu(T^{-1}(A))$, $A\in\mathscr T$. Assume that the map
$W_{\mathcal G}^{\star}: (\mathcal P,\mathscr P)
\to(\mathbb R,\mathcal B(\mathbb R))$ is measurable. Then, for every $A\in\mathcal B(\mathbb R)$, $\Pi\!\left( W_{\mathcal G}^{\star}\in A \mid \D_n \right) = \Pi\!\left( P: W_{\mathcal G}^{\star}(P)\in A \mid \D_n \right)$.
}
The posterior distribution of optimal treatment assignments is more subtle.
Because $\mathcal G^{\star}(P)$ may contain multiple optimal treatment rules, the posterior is a distribution over sets of rules. Under suitable measurability conditions, this set-valued posterior can be characterized through its posterior capacity functional \citep[see, e.g.,][]{florens2021revisiting}.\footnote{Formally, let $\Pi(P \in \cdot\mid\D_n)$ denote the posterior law of $P$
on the measurable space $(\mathcal P, \mathscr P)$ of reduced-form distributions.
Suppose that the policy class $\mathcal G$ is a locally compact second countable Hausdorff space.
Let $\mathcal C$ denote the collection of closed subsets of $\mathcal G$, and let $\mathcal K$ denote the collection of compact subsets of $\mathcal G$. 
Suppose that the map $\mathcal G^{\star}: \mathcal{P} \to \mathcal C$ is a \emph{random closed set}, in the sense that, for every compact $K \in \mathcal K$, $\{P:\mathcal G^{\star}(P)\cap K\neq \emptyset \} \in \mathscr{P}$. 
Then the posterior law of $\mathcal G^\star(P)$ is uniquely determined by its \emph{posterior capacity functional} $C_{\mathcal G^\star\mid\D_n}: \mathcal K \to[0,1]$, defined by $C_{\mathcal G^\star\mid\D_n}(K) \coloneqq \Pi\!\left( \mathcal G^{\star}(P) \cap K \neq \emptyset \mid\D_n \right)$ for $K \in \mathcal{K}$.}
Taken together, these marginal posteriors provide a full theory of inference for optimal welfare and optimal treatment assignments.

\begin{remark}[Reduced-form Data]
\label{Remark1}

This paper adopts a limited-information Bayesian approach \citep{sims2006example}. Specifically, I assume that the DM places a nonparametric prior on the distribution of the reduced-form vector $D \coloneqq (Z_1, Z_0, X)$, where $Z_1 \coloneqq YT/e(X)$ and $Z_0 \coloneqq Y(1-T)/[1-e(X)]$.
This yields a well-defined Bayesian inference problem for the transformed data $D$.
Such an approach is motivated by difficulties that can arise when using structural priors under Bayesian ignorability; see \citet{li2023bayesian} for a review. 
The first is the difficulty of incorporating known propensity scores, which do not enter the outcome likelihood. Existing solutions include propensity-dependent priors \citep[e.g.,][]{robins2012robins, ray2020semiparametric, breunig2025double} and outcome models augmented with the propensity score
\citep[e.g.,][]{hahn2020bayesian, walker2026parametrization}.
The second is the difficulty of obtaining minimax regret guarantees such as Theorem~\ref{Theorem1}, given nontrivial semiparametric bias concerns \citep[e.g.,][]{ritov2014bayesian, ray2020semiparametric, breunig2025double}.
In contrast to NBPL, \citet{chamberlain2011bayesian} places a structural prior on the distribution of $(Y(1),Y(0),T,X)$ and discusses extensions with propensity-dependent priors. However, his prior specification is restricted to finite-dimensional models and the analysis does not establish welfare-regret guarantees.

\end{remark}

\subsubsection{Posterior Reporting}
\label{NBPL:posterior reporting}

I use the marginal posterior distributions defined in Section~\ref{NBPL:inference framework} to construct posterior summaries of optimal welfare and optimal treatment assignments and to conduct model selection across policy classes.

Bayesian inference on optimal welfare is standard because its marginal posterior is a distribution on $\R$. One may therefore report the posterior median as a point estimate and a $(1-\alpha)$ equal-tailed credible interval for uncertainty quantification, with endpoints given by the $\alpha/2$ and $(1-\alpha/2)$ posterior quantiles. 
This provides a Bayesian complement to existing frequentist approaches to inference on optimal welfare \citep[e.g.,][]{luedtke2016statistical, chen2025inference,
ponomarev2025lower, whitehouse2025inference}.

Bayesian inference on optimal treatment assignments is more subtle because its marginal posterior is a distribution over \textit{sets} of rules. 
I conduct such inference using the framework of \citet{kline2016bayesian}. In particular, I place a flexible Dirichlet process prior on the identifiable reduced-form parameter $P$ and map its posterior into a posterior for the identified set $\mathcal G^\star(P)$ of optimal treatment rules.\footnote{\citet{kline2016bayesian} places a prior on a finite-dimensional identifiable reduced-form parameter and maps its posterior to obtain a posterior for the identified set of a structural parameter. The same idea extends naturally to NBPL by treating $P$ as an infinite-dimensional identifiable reduced-form parameter. The posterior for $P$ then provides a full theory of inference for functionals of $P$, including $\mathcal G^\star(P)$.}
The marginal posterior for optimal treatment assignments can then be used to form posterior probability statements about the identified set, such as whether a candidate rule belongs to $\mathcal G^\star(P)$ or whether all optimal rules
satisfy a given property (e.g., a capacity constraint or fairness criterion).
Quantities induced by the set of optimal rules, such as the treatment share, may likewise be set-identified.

I illustrate NBPL inference on optimal treatment assignments using linear rules.

\begin{example}[Inference on optimal assignments]
\label{Example1}

Consider the class of linear rules
\tightdisplay{
    \mathcal G_{\mathrm{lin}}
    \coloneqq
    \left\{
    G_{\beta} = \{x\in\mathcal X: \beta_0 + x^\top \beta_1 > 0 \}:
    \beta = (\beta_0, \beta_1)\in\mathcal B\subseteq \mathbb R^{d_x + 1}
    \right\}.
}
For each $P$, let $\mathcal B^\star(P) \coloneqq \arg\max_{\beta\in\mathcal B} W(P;G_\beta)$ denote the set of coefficients indexing the optimal linear rules under $P$. 
The posterior for $P$ induces a posterior distribution over $\mathcal B^\star(P)$. Thus, posterior inference on optimal linear treatment rules reduces to inference on the corresponding set of optimal coefficients, $\mathcal B^\star(P)$.
For instance, a policymaker may report the following posterior summaries (under suitable measurability conditions):
\begin{enumerate}[label=(\alph*)]

    \item For a candidate linear rule $G_{\widetilde\beta}$, the posterior probability $\Pi(P: \widetilde\beta \in \mathcal B^\star(P)\mid \mathcal D_n)$ measures the posterior support for its optimality within $\mathcal G_{\mathrm{lin}}$. It therefore provides a Bayesian diagnostic for assessing a prespecified rule.

    \item Policy implementation is often subject to budget constraints \citep[e.g.,][]{bhattacharya2012inferring, sun2025empirical}. 
    A simple feasibility diagnostic asks whether every unconstrained optimal linear rule treats at most a fraction $c\in(0,1)$ of the population. The corresponding posterior probability is $\Pi\big(P: \sup_{\beta\in \mathcal B^\star(P)} P(X\in G_\beta) \leq c \mid \mathcal D_n \big)$.
    
\end{enumerate}

\end{example}

Beyond inference on $(\mathcal G^{\star}(P),W_{\G}^{\star}(P))$ for a fixed policy class $\G$, NBPL also provides a simple and consistent posterior diagnostic for practitioners comparing competing, possibly non-nested policy classes.
Consider two policy classes $\G_1$ and $\G_2$. 
For example, a policymaker may wish to compare linear rules \citep[e.g.,][]{kitagawa2018should}
with decision-tree rules \citep[e.g.,][]{athey2021policy, ida2026choosing}.
Uncertainty about such comparisons again entirely reflects uncertainty about $P_0$. In particular, the posterior for $P$ induces a joint posterior distribution for $(W_{\G_1}^{\star}(P),W_{\G_2}^{\star}(P))$. The posterior probability
$\Pi(P: W_{\G_1}^{\star}(P) > W_{\G_2}^{\star}(P) \mid \D_n)$
then provides a natural measure of evidence that the maximal attainable welfare under $\G_1$ exceeds that under $\G_2$.
Theorem~\ref{Theorem2} in Section~\ref{theory:model selection} formalizes the pointwise consistency of this posterior model selection procedure.

\subsection{Implementation}
\label{NBPL:implementation}

In this section, I develop a computationally tractable Bayesian-bootstrap algorithm \citep{rubin1981bayesian} for sampling from the posterior distribution of $(\mathcal G^\star(P),W_{\mathcal G}^\star(P))$.
Algorithm~\ref{Algorithm:NBPL} is generic and can be implemented using off-the-shelf policy-learning algorithms for the chosen policy class $\mathcal G$: one need only replace the uniform weights $n^{-1}$ with normalized independent $\mathrm{Exp}(1)$ weights.\footnote{\label{footnote:Bayesian bootstrap}Formally, a draw from the Bayesian-bootstrap posterior $\mathrm{DP}(n\mathbb P_n)$ admits the representation 
$P \stackrel{d}{=} \sum_{i=1}^n \widetilde{\omega}_i\,\delta_{D_i}$, where $\widetilde{\omega}_i = \omega_i/\sum_{j=1}^n\omega_j$ and $\omega_i\stackrel{\mathrm{i.i.d.}}{\sim}\mathrm{Exp}(1)$ \citep{rubin1981bayesian}. 
Thus, Bayesian-bootstrap draws can be viewed as randomly reweighted or smoothed versions of the empirical distribution.
}
The $S$ posterior draws can be computed in parallel.

\begin{algorithm}[H]
\caption{Nonparametric Bayesian Policy Learning via the Bayesian Bootstrap}
\label{Algorithm:NBPL}
\begin{algorithmic}[1]
\State \textbf{Input:} Data $\{(Y_i,T_i,X_i)\}_{i=1}^n$, policy class $\mathcal{G}$, number of posterior draws $S$.
\For{$s=1,\ldots,S$}
    \State \textbf{Step 1.} Draw $\{\omega_i^{[s]}\}_{i=1}^n \stackrel{\mathrm{i.i.d}}{\sim} \mathrm{Exp}(1)$ and let $\widetilde{\omega}_i^{[s]} \coloneqq \omega_i^{[s]} / \sum_{j=1}^n \omega_j^{[s]}$.
    
    \State \textbf{Step 2.} Compute the set of welfare-maximizing rules
    \tightdisplay{
    \mathcal G^{\star,[s]}
    \coloneqq
    \arg\max_{G\in\mathcal{G}} \,\,
    \sum_{i=1}^n
    \widetilde{\omega}_i^{[s]}
    \left(
    \frac{Y_iT_i}{e(X_i)}
    -
    \frac{Y_i(1-T_i)}{1-e(X_i)}
    \right)
    \mathbf{1}\{X_i\in G \}.
    }
    
    \State \textbf{Step 3.} Compute the optimal welfare
    \tightdisplay{
    W_{\G}^{\star,[s]}
    \coloneqq 
    \sup_{G \in \mathcal{G}} \,\,
    \sum_{i=1}^n
    \widetilde{\omega}_i^{[s]}
    \left(
    \frac{Y_iT_i}{e(X_i)}
    -
    \frac{Y_i(1-T_i)}{1-e(X_i)}
    \right)
    \mathbf{1}\{X_i\in G \}.
    }
\EndFor
\State \textbf{Output:} Posterior draws $\{\mathcal G^{\star,[s]}, W_{\G}^{\star,[s]}\}_{s=1}^S$.
\end{algorithmic}
\end{algorithm}

Given posterior draws $\{\mathcal G^{\star,[s]}, W_{\G}^{\star,[s]}\}_{s=1}^S$, posterior summaries can be computed from their empirical distribution. Take optimal welfare as an example. For notational simplicity, let $W^{[s]}\coloneqq W_{\G}^{\star,[s]}$, and denote the corresponding order statistics by $W^{(1)}\leq\cdots\leq W^{(S)}$. For $p\in(0,1)$, define the empirical posterior $p$-quantile as $\widehat q_p\coloneqq W^{(\lceil Sp\rceil)}$. One may report the posterior median $\widehat q_{1/2}$ as a point estimate and the $(1-\alpha)$ equal-tailed credible interval $\left[ \widehat q_{\alpha/2}, \widehat q_{1-\alpha/2} \right]$ for uncertainty quantification. 
To perform model selection between two policy classes $\G_1$ and $\G_2$, one simply reports the empirical fraction of the posterior draws for which $W_{\G_1}^{\star,[s]} > W_{\G_2}^{\star,[s]}$.

\begin{remark}[General Base Measure]
\label{Remark2}

The Bayesian-bootstrap posterior $\mathrm{DP}(n \Prob_n)$ arises as the weak limit of $\mathrm{DP}(\alpha + n\Prob_n)$ as the prior mass vanishes, $|\alpha| \to 0$. 
Depending on the application, researchers may wish to employ a DP prior with a nonzero base measure $\alpha$ to incorporate external evidence, expert judgment, or predictions from auxiliary models.
Such prior specifications connect NBPL to evidence aggregation and meta-analysis, which combine information from previous studies or related populations with the current data to inform decision-making \citep[e.g.,][]{ishihara2024evidence, athey2026surrogate}. 
Recent work also uses generative AI to construct base measures from synthetic data \citep[e.g.,][]{o2025ai, li2026empirical}. 
Algorithm~\ref{Algorithm:NBPL} can be extended to sample from
the resulting posterior $\mathrm{DP}(\alpha+n\mathbb P_n)$ using standard stick-breaking methods \citep[e.g.,][]{sethuraman1994constructive,ishwaran2001gibbs}.

\end{remark}


\section{Empirical Illustration}
\label{sec:empirical}

I illustrate the Nonparametric Bayesian Policy Learning (NBPL) framework in two empirical applications \citep{bhattacharya2012inferring, kitagawa2018should} and compare its performance with that of empirical welfare maximization \citep[EWM;][]{kitagawa2018should}.\footnote{
Following \citet[Remark~2.3]{kitagawa2018should}, I implement both EWM and NBPL using demeaned outcomes. This normalization is useful in applications because the resulting rules are invariant to positive affine transformations of the outcome, $Y\mapsto aY+b$ with $a>0$. Theorem~\ref{Theorem:C2} in Supplemental Appendix~\ref{supplement:extensions} shows that demeaning does not affact the posterior regret contraction rate in Theorem~\ref{Theorem1}.} 
In both applications, I implement NBPL using the Bayesian bootstrap and consider two policy classes: linear rules and decision-tree rules of depth at most two, denoted by
$\mathcal{G}_{\mathrm{lin}}$ and $\mathcal{G}_{\mathrm{tree},2}$, respectively.
I also illustrate posterior model selection across these
classes, a novel feature of NBPL.
Appendix~\ref{Appendix:simulation} presents a Monte Carlo simulation calibrated to the data from \cite{bhattacharya2012inferring}.

\subsection{Targeting Job Trainings}
\label{empirical:JTPA}

The first application uses data from the Job Training Partnership Act (JTPA) experiment, in which applicants were randomly assigned eligibility for job training and related services for an 18-month period. 
Let $T=1$ denote program eligibility and $T=0$ otherwise, and let $Y$ denote earnings measured 30 months after assignment. Following \citet{kitagawa2018should}, I consider two outcome specifications: earnings without treatment costs and earnings
net of a treatment cost of \$774 per participant.
The baseline covariate vector is $X = (\texttt{PreEarn},\texttt{Educ})^{\top}$, where \texttt{PreEarn} and \texttt{Educ} denote pre-treatment earnings and years of education, respectively. 
Random assignment implies unconditional unconfoundedness, 
$(Y(1),Y(0),X)\indep T$. 
The sample used by \citet{kitagawa2018should} contains $n=9{,}223$ observations, and the propensity score is known and constant: $e(X)=2/3$.

Figure~\ref{fig:JTPA} reports the posterior distributions of optimal welfare under NBPL for the JTPA application without treatment costs. The analysis extends \citet{kitagawa2018should} by considering decision-tree rules and conducting posterior model selection across policy classes.
The left panel plots the posterior empirical cumulative distribution functions (CDFs) of $W^\star_{\G_{\mathrm{lin}}}(P)$ and $W^\star_{\G_{\mathrm{tree},2}}(P)$ based on 1,000 Bayesian-bootstrap draws. The CDF for $\G_{\mathrm{tree},2}$ lies below that for $\G_{\mathrm{lin}}$ throughout, indicating that optimal welfare under the decision-tree class first order stochastically dominates that under the linear class.
The EWM welfare estimates lie at relatively low posterior quantiles: approximately the 31st percentile for $\G_{\mathrm{lin}}$ and the 23rd percentile for $\G_{\mathrm{tree},2}$. This suggests that plug-in estimates may understate optimal welfare relative to the posterior distribution.
The right panel plots the posterior distribution of the optimal-welfare difference $W^\star_{\G_{\mathrm{tree},2}}(P) - W^\star_{\G_{\mathrm{lin}}}(P)$. 
The posterior probability that the decision-tree class attains higher optimal welfare than the linear class is 
94.8\%, providing strong evidence in favor of the decision-tree class.

\begin{figure}[htbp]
\centering
\includegraphics[width=1\textwidth]{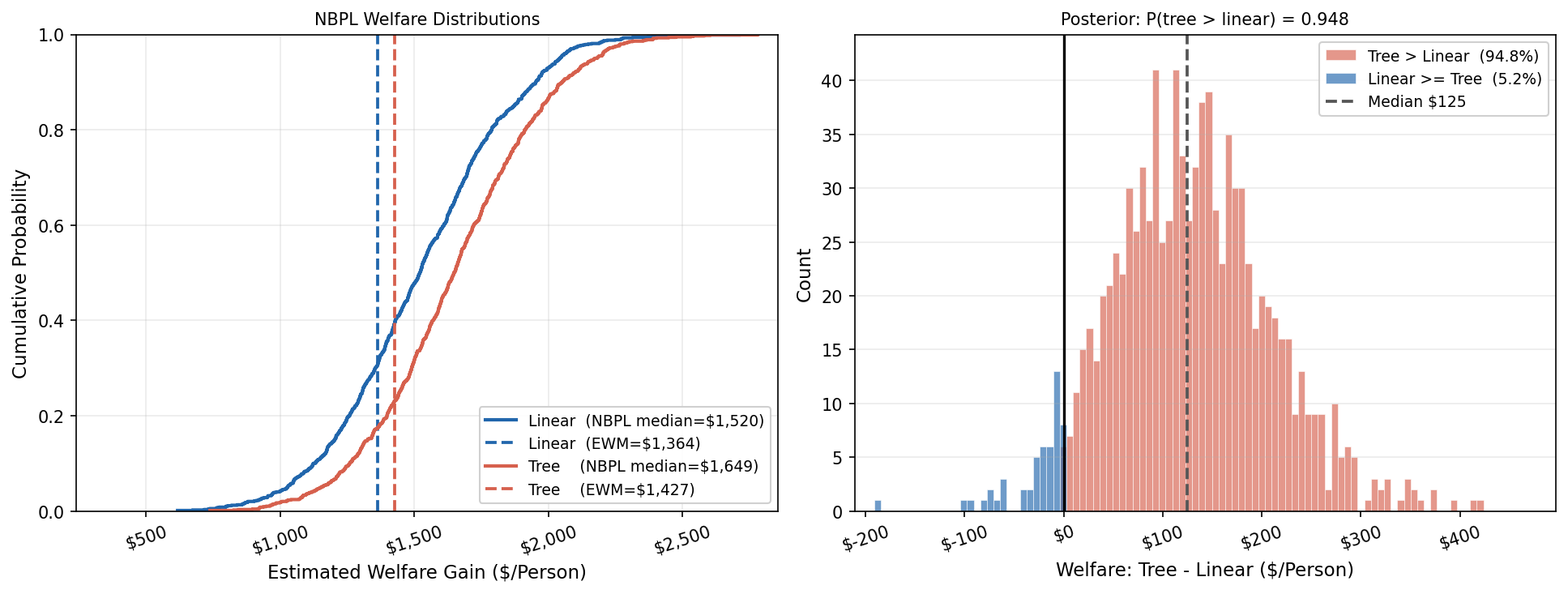}
\caption[Posterior Welfare Comparison, JTPA]{Posterior welfare comparison for the JTPA application without treatment costs. 
The left panel plots the posterior CDFs of optimal welfare under $\G_{\mathrm{lin}}$ and $\G_{\mathrm{tree},2}$; dashed vertical lines mark the corresponding EWM welfare estimates.
The right panel plots the posterior distribution of the optimal-welfare difference $W^\star_{\G_{\mathrm{tree},2}}(P)-W^\star_{\G_{\mathrm{lin}}}(P)$.
The posterior probability that the decision-tree class attains higher optimal welfare than the linear class is 
94.8\%.}
\label{fig:JTPA}
\end{figure}

\begin{figure}[htbp]
\centering
\includegraphics[width=1\textwidth]{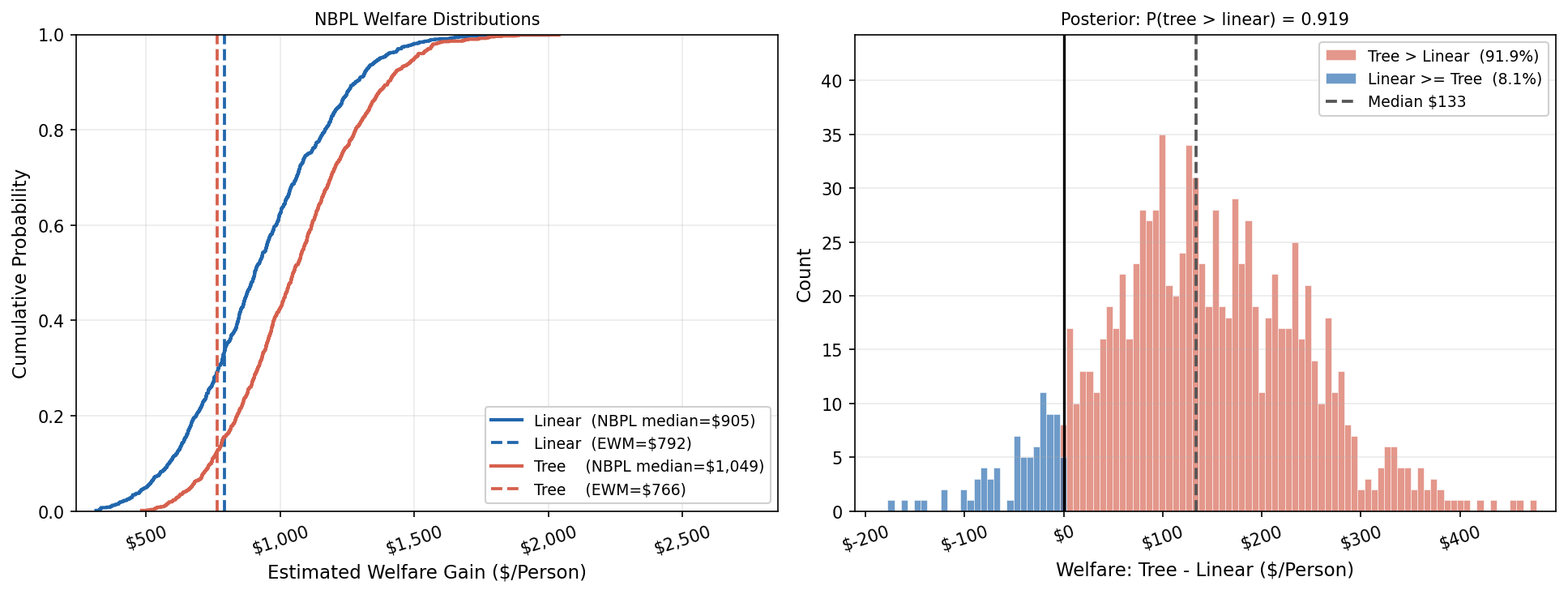}
\caption[Posterior Welfare Comparison with Treatment Costs, JTPA]{Posterior welfare comparison for the JTPA application with a \$774 treatment cost. 
The left panel plots the posterior CDFs of optimal welfare under $\G_{\mathrm{lin}}$ and $\G_{\mathrm{tree},2}$; dashed vertical lines mark the corresponding EWM welfare estimates.
The right panel plots the posterior distribution of the welfare difference $W^\star_{\G_{\mathrm{tree},2}}(P)-W^\star_{\G_{\mathrm{lin}}}(P)$. 
The posterior probability that the decision-tree class attains higher optimal welfare than the linear class is 91.9\%.}
\label{fig:JTPA cost}
\end{figure}

Figure~\ref{fig:JTPA cost} reports the posterior distributions of optimal welfare with a \$774 per-person treatment cost. As in Figure~\ref{fig:JTPA}, the CDF for $\mathcal G_{\mathrm{tree},2}$ lies below that for $\mathcal G_{\mathrm{lin}}$, indicating first-order stochastic dominance. The EWM welfare estimates lie at the 34th and 13th posterior percentiles under $\G_{\mathrm{lin}}$ and $\G_{\mathrm{tree},2}$, respectively. The posterior probability that $\mathcal G_{\mathrm{tree},2}$ attains higher optimal welfare than $\G_{\mathrm{lin}}$ is 91.9\%.

\subsection{Targeting Bednet Subsidies}
\label{empirical:dupas}

The second application uses data from the randomized subsidy experiment for insecticide-treated bednets (ITNs) analyzed by
\citet{bhattacharya2012inferring}. In the experiment, rural households in western Kenya were randomly offered vouchers that specified different subsidy levels for ITNs and could be redeemed at local retailers \citep{dupas2009what}.
Let $T=1$ if a household was offered a high ITN subsidy and $T=0$ otherwise.
 The outcome $Y \in \{0,1\}$ indicates bednet coverage and equals one if the household redeemed the voucher and the net was observed in use.
The baseline covariate vector is $X=(\texttt{YoungChild},\texttt{BankAccount},\texttt{Wealth})^{\top}$, where $\texttt{YoungChild}$ indicates the presence of a child under age ten, $\texttt{BankAccount}$ indicates bank account ownership, and $\texttt{Wealth}$ denotes log household wealth per capita. 
Random assignment implies unconditional unconfoundedness, $(Y(1),Y(0),X) \indep T$. 
The sample analyzed by \citet{bhattacharya2012inferring} contains $n=1{,}098$ observations. The propensity score is constant by design, and I approximate its value by the empirical treatment share, $e(X)=0.16$.\footnote{The experiment uses complete random assignment at the household level \citep{dupas2009what}, implying a constant propensity score. Because the assignment probability is not reported explicitly, I approximate its value using the empirical treatment share in the analysis sample.
}

Figure~\ref{fig:dupas} reports the posterior distributions of optimal welfare under NBPL for the bednet subsidy application. The left panel plots the empirical posterior CDFs of $W_{\G_{\mathrm{lin}}}^\star(P)$ and $W_{\G_{\mathrm{tree},2}}^\star(P)$ based on 1,000 Bayesian-bootstrap draws.
The two CDFs appear similar, although the distribution of $W_{\G_{\mathrm{tree},2}}^\star(P)$ is shifted slightly to the right, suggesting modest welfare gains from using decision-tree rather than linear rules.
The smaller difference relative to the JTPA application may reflect the simpler covariate structure in the data analyzed by \citet{bhattacharya2012inferring}: two of the three covariates are binary, potentially limiting differences in the treatment assignments generated by the two classes.
The EWM estimates are close to the posterior medians for both classes.
The right panels plot the posterior distribution of the optimal-welfare difference $W^\star_{\G_{\mathrm{tree},2}}(P) - W^\star_{\G_{\mathrm{lin}}}(P)$. The posterior probability that the decision-tree class attains higher optimal welfare than the linear class is  91.6\%.
Thus, the posterior provides substantial evidence in favor of the decision-tree class, although the magnitude of the welfare gain is small.

\begin{figure}[htbp]
\centering
\includegraphics[width=1\textwidth]{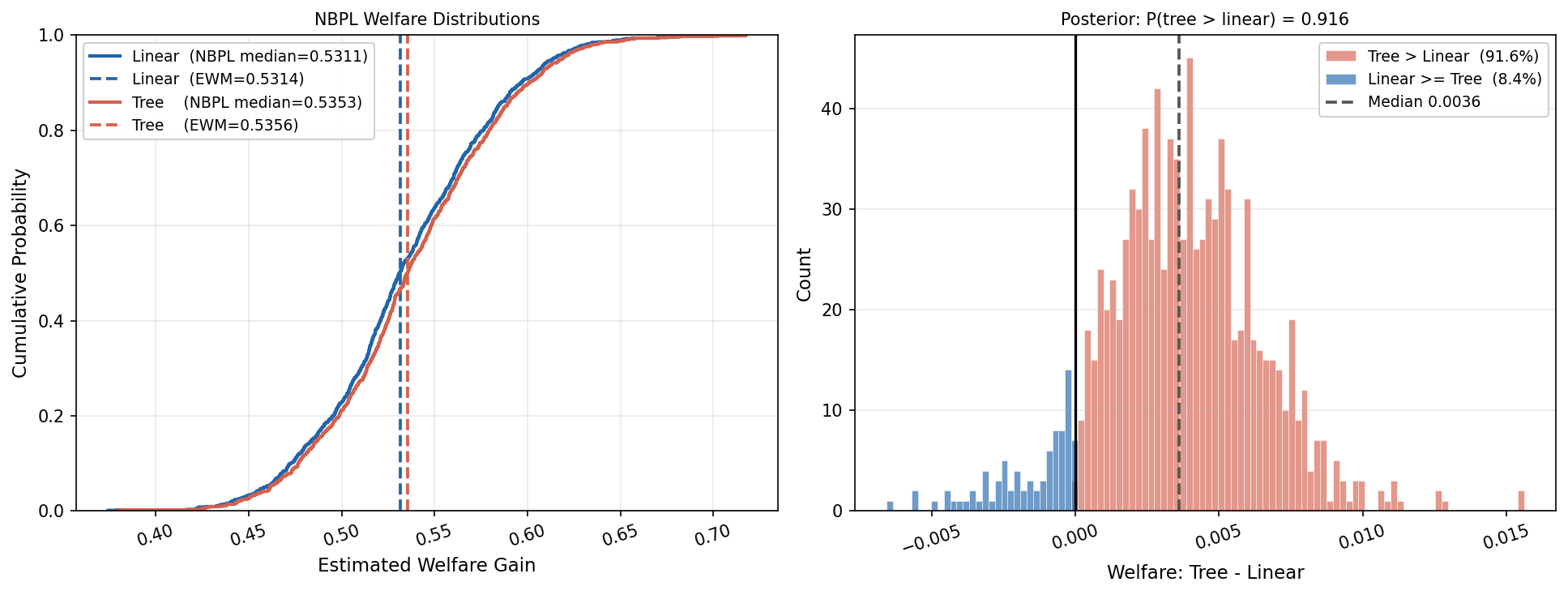}
\caption[Posterior Welfare Comparison, Bednet Application]{Posterior welfare comparison for the bednet subsidy application. The left panel plots the posterior CDFs of optimal welfare under $\G_{\mathrm{lin}}$ and $\G_{\mathrm{tree},2}$; dashed vertical lines mark the corresponding EWM welfare estimates.
Because the outcome is a binary indicator for bednet usage, welfare gains are measured in bednet-usage-rate units (e.g., a value of 0.5 represents a 50-percentage-point increase in bednet usage).
The right panel plots the posterior distribution of the optimal-welfare difference $W^\star_{\G_{\mathrm{tree},2}}(P)-W^\star_{\G_{\mathrm{lin}}}(P)$.
The posterior probability that the decision-tree class attains higher optimal welfare than the linear class is
91.6\%.}
\label{fig:dupas}
\end{figure}

\begin{remark}[Calibrated Simulation]
\label{Remark3}

Appendix~\ref{Appendix:simulation} presents a Monte Carlo study calibrated to the data analyzed by \citet{bhattacharya2012inferring}. Three main findings emerge.
First, NBPL and EWM perform similarly under the welfare-regret
criterion studied in Section~\ref{theory:regret}. 
This is not surprising. As noted in footnote~\ref{footnote:Bayesian bootstrap}, NBPL evaluates welfare under random reweightings of the empirical distribution $\Prob_n$, with the weight on each observation centered at its value under $\Prob_n$ and fluctuating at order $O_P(n^{-1})$.
Second, decision-tree rules attain slightly higher optimal welfare than linear rules. 
Third, the posterior model-selection procedure provides stable evidence in favor of the decision-tree class.
    
\end{remark}


\section{Theoretical Properties of NBPL}
\label{sec:theory}

This section presents two theoretical results for NBPL. First, I establish a minimax-optimal convergence rate for posterior welfare regret (Theorem~\ref{Theorem1}). 
Second, I propose a posterior model-selection procedure for comparing policy classes and prove its pointwise consistency
(Theorem~\ref{Theorem2}). 
I conclude by relating NBPL to existing policy learning approaches \citep{manski2004statistical, chamberlain2011bayesian, kitagawa2018should}.

\subsection{Assumption on the Policy Class}
\label{theory:VC}

The theoretical results require controlling the complexity of the policy class $\mathcal{G}$. Following \citet{kitagawa2018should}, I assume that $\mathcal{G}$ has a finite Vapnik--Chervonenkis (VC) dimension.\footnote{
The Vapnik--Chervonenkis (VC) dimension of a class of sets $\mathcal G\subseteq 2^{\mathcal X}$ is a combinatorial measure of its complexity. The class $\mathcal G$ is said to \emph{shatter} a finite set $\{x_1,\ldots,x_\ell\} \subseteq \mathcal X$ if 
$\{G\cap\{x_1,\ldots,x_\ell\}:G\in\mathcal G \}$ 
contains all $2^\ell$ subsets of $\{x_1,\ldots,x_\ell\}$. The VC dimension is the largest $\ell$ for which $\mathcal G$ shatters some set of $\ell$ points in $\mathcal{X}$ \citep[see, e.g.,][]{vaart2023weak}.
}

\begin{assumption}
\label{Assumption:VC}

The policy class $\G$ has a finite VC-dimension
$v < \infty$.

\end{assumption}

This assumption accommodates many tractable and interpretable policy classes, including linear rules and decision-tree rules.\footnote{It does not cover all policy classes; monotone rules are one example \citep{mbakop2021model}.} 
The linear class $\mathcal G_{\mathrm{lin}}$ in Example~\ref{Example1} has a VC-dimension
$\mathrm{VC}(\mathcal G_{\mathrm{lin}})\leq d_x + 1$
\citep[Problem 2.6.14]{vaart2023weak}.
Decision-tree classes are also covered. 
Let $\mathcal G_{\mathrm{tree},L}$ denote the class of treatment rules generated by decision trees of depth at most $L$.
A decision tree recursively partitions the covariate space into smaller regions. Each split divides one existing region into two according to whether $x_j\leq t$ or $x_j>t$ for some covariate $j \in \{1,\ldots,d_x \}$ and cutoff $t\in\mathbb R$. A region that is not split further is called a terminal leaf and is assigned treatment or control. The tree has depth at most $L$ if each terminal leaf results from at most $L$ successive splits. The treatment rule $G$ is the union of the terminal leaves assigned to treatment. 
The class has a VC-dimension of order $\mathrm{VC}(\mathcal G_{\mathrm{tree},L})=O(2^L\log d_x)$, up to polylogarithmic factors \citep{athey2021policy}.

\subsection{Posterior Regret Guarantee}
\label{theory:regret}

Following \citet{manski2004statistical}, I evaluate a treatment rule by its \textit{welfare regret}, defined as the welfare loss relative to the best feasible rule in the policy class $\G$.

\begin{definition}
\label{def:regret}
    The welfare regret of the rule $G^{\star}(P)$ is $R(P_0;P) \coloneqq W_{\mathcal{G}}^{\star}(P_0) - W(P_0; G^{\star}(P))$, where $G^\star(P)\in\arg\max_{G\in\mathcal G}W(P;G)$ is the rule selected under $P$.
\end{definition}

The regret $R(P_0;P)$ measures the welfare loss incurred under $P_0$ by acting as if $P$ were the true reduced-form distribution. It therefore defines a policy-relevant discrepancy between $P$ and $P_0$. Posterior-regret analysis examines whether the posterior concentrates on distributions $P$ that incur small welfare losses. This parallels posterior contraction rates, a central notion in the theoretical statistics literature \citep[e.g.,][]{ghosal2017fundamentals}, which studies the rate at which the posterior concentrates around the true data-generating parameter.
The key distinction is that NBPL measures contraction using welfare regret rather than conventional statistical distances such as $L_p$ or Hellinger distance.

Theorem~\ref{Theorem1} establishes that posterior regret contracts to zero at the minimax-optimal rate.\footnote{Although $R(P_0;P)$ is infeasible because $P_0$ is unknown, it has a well-defined posterior distribution induced by the posterior distribution of $P$. When $\mathcal G^\star(P)$ contains multiple maximizers, the choice of $G^\star(P)\in\mathcal G^\star(P)$ does not affect the posterior regret contraction rate.} 
To the best of my knowledge, this result provides the first policy-relevant analogue of posterior contraction rates. 
Theorems~\ref{Theorem:C1}--\ref{Theorem:C3} in Supplemental Appendix~\ref{supplement:extensions} extend Theorem~\ref{Theorem1} to settings with multi-valued discrete treatments, demeaned outcomes, and capacity constraints, respectively.

\begin{theorem}
\label{Theorem1}

Suppose Assumptions~\ref{Assumption:DGP}--\ref{Assumption:VC} hold. Let $\mathcal{D}_n \coloneqq \{D_i\}_{i=1}^n \stackrel{\mathrm{i.i.d.}}{\sim} P_0$ denote the observed sample, and let $\Pi(P \in \cdot \mid \D_n)$ denote the Dirichlet process posterior $P \mid \D_n \sim \mathrm{DP}(\alpha + \sum_{i=1}^n \delta_{D_i})$.
Then, for every sequence $C_n\to\infty$ (arbitrarily slowly), as $n \to \infty$,
\tightdisplay{
    \Pi \left.\left(P: R(P_0; P) \leq C_n \, \sqrt{\frac{v}{n}} \,\,\right|\,\, \mathcal{D}_n \right) 
    \stackrel{P_0}{\to} 1.
    \tag{4.1} \label{main:eq6}
}
\end{theorem}

\begin{proof}
    See Appendix~\ref{Appendix:proof Theorem1}.
\end{proof}

Two features of Theorem~\ref{Theorem1} are worth emphasizing. 
First, posterior welfare regret under NBPL contracts at the same minimax-optimal rate achieved by empirical welfare maximization (EWM) in \cite{kitagawa2018should}, up to an arbitrarily slowly diverging factor $C_n$.\footnote{Such slowly diverging factors are standard in posterior contraction theory \citep[e.g.,][]{ghosal2000convergence, shen2001rates, vaart2008rates} and are typically regarded as inessential for rate comparisons.}
Thus, accounting for parameter uncertainty does not slow the rate of learning relative to nonparametric frequentist approaches.
Second, the posterior welfare-regret criterion differs from the regret criterion typically used in frequentist policy learning \citep[e.g.,][]{kitagawa2018should,athey2021policy}.
The frequentist criterion evaluates the expected regret of a selected rule over repeated samples drawn from $P_0$.
By contrast, Theorem~\ref{Theorem1} shows that, under repeated sampling from $P_0$, the posterior concentrates on distributions $P$ for which the regret under $P_0$ is at most $C_n\sqrt{v/n}$.
In Section~\ref{theory:comparison with literature}, Theorem~\ref{Theorem3} complements this posterior contraction result with a finite-sample bound under the regret criterion of \citet{kitagawa2018should}. The result establishes uniform minimax optimality for a class of prior-regularized treatment rules that embeds EWM as a special case.

To derive the rate in~\eqref{main:eq6}, I control the two terms in the upper bound
\tightdisplay{
    R(P_0;P) \le 2 \{ \rho_{\G}(P_0, \Prob_n) +\rho_{\G}(\Prob_n, P) \},
}
where $\rho_{\G}(Q_1,Q_2)\coloneqq \sup_{G\in\G}\lvert W(Q_1;G)-W(Q_2;G)\rvert$ defines a welfare-relevant semi-metric on the space of probability distributions.
The first term, $\rho_{\G}(P_0,\Prob_n)$, captures sampling error. Since $\G$ has a finite VC-dimension, the closeness of $\Prob_n$ to $P_0$ translates into $n^{-1/2}$ control of the maximal welfare differences uniformly over $\G$.
The second term, $\rho_{\G}(\Prob_n, P)$, captures posterior uncertainty and is specific to NBPL. Under
the posterior $\mathrm{DP}(\alpha+\sum_{i=1}^n\delta_{D_i})$, a draw admits the decomposition
\tightdisplay{
    P \stackrel{d}{=}
    V_n Q+(1-V_n)\sum_{i=1}^n \widetilde \omega_i\delta_{D_i},
    \qquad
    \widetilde \omega_i\coloneqq
    \frac{\omega_i}{\sum_{j=1}^n \omega_j}.
}
where $V_n\sim\mathrm{Beta}(|\alpha|,n)$, $Q\sim\mathrm{DP}(\alpha)$, and  $\omega_i\overset{\mathrm{i.i.d.}}{\sim}\mathrm{Exp}(1)$ are mutually independent. 
Since $V_n=O_P(n^{-1})$, the prior component $V_nQ$ is asymptotically negligible. The remaining component is a stochastic reweighting of the empirical distribution: it assigns observation $i$ a weight $\widetilde\omega_i$ that is centered at $n^{-1}$ and fluctuates at order $O_P(n^{-1})$.
Thus, $P$ remains asymptotically close to $\Prob_n$. The required uniform $n^{-1/2}$ control for $\rho_{\mathcal G}(\Prob_n,P)$ then follows from multiplier empirical-process bounds for VC classes.

\subsection{Posterior Model Selection}
\label{theory:model selection}

As discussed in Section~\ref{NBPL:posterior reporting}, the posterior for $P$ induces a joint posterior for
$(W_{\G_1}^{\star}(P),W_{\G_2}^{\star}(P))$.
This joint posterior provides a basis for comparing competing policy classes.
In particular, the posterior probability $\Pi(W_{\G_1}^{\star}>W_{\G_2}^{\star}\mid\D_n)$ provides a natural diagnostic of whether the maximal attainable welfare under $\G_1$ exceeds that under $\G_2$.
To the best of my knowledge, NBPL is the first policy-learning framework to provide an explicit model-selection criterion for comparing non-nested policy classes.

Theorem~\ref{Theorem2} establishes pointwise consistency of this model-selection procedure. If one policy class attains strictly higher welfare than the other under $P_0$, then the posterior probability of selecting the welfare-dominant class converges to one. 
If the two classes attain the same optimal welfare under $P_0$, the posterior distribution of the optimal-welfare difference concentrates at zero and assigns vanishing probability to any fixed positive gap.

\begin{theorem}
\label{Theorem2}

Let $\mathcal G_1, \mathcal G_2$ be two policy classes. 
Suppose Assumptions~\ref{Assumption:DGP}--\ref{Assumption:base measure} hold and that each class $\mathcal G_i$ satisfies Assumption~\ref{Assumption:VC}. 
Let $\mathcal{D}_n \coloneqq \{D_i\}_{i=1}^n \stackrel{\mathrm{i.i.d.}}{\sim} P_0$ denote the observed sample, and let $\Pi(P \in \cdot \mid \D_n)$ denote the Dirichlet process posterior $P \mid \D_n \sim \mathrm{DP}(\alpha + \sum_{i=1}^n \delta_{D_i})$.

\begin{enumerate}[label=(\roman*)]
    \item If $W_{\G_1}^{\star}(P_0) \neq W_{\G_2}^{\star}(P_0)$, then, as $n\to\infty$,
    \tightdisplay{
        \Pi\left(
        P:\,
        \sign\,\bigl(W_{\G_1}^{\star}(P)-W_{\G_2}^{\star}(P)\bigl)
        \,=
        \sign\,\bigl(W_{\G_1}^{\star}(P_0)-W_{\G_2}^{\star}(P_0)\bigr)
        \,\middle|\, \D_n
        \right)
        \stackrel{P_0}{\to} 1.
        \tag{4.3}\label{main:eq8}
    }

    \item If $W_{\G_1}^{\star}(P_0) = W_{\G_2}^{\star}(P_0)$, then for any $\varepsilon>0$, as $n\to\infty$,
    \tightdisplay{
        \Pi\left(
        P:\,
        \bigl|W_{\G_1}^{\star}(P)-W_{\G_2}^{\star}(P)\bigr| < \varepsilon
        \,\middle|\, \mathcal D_n
        \right)
        \stackrel{P_0}{\to} 1.
        \tag{4.4} \label{main:eq9}
    }
\end{enumerate}

\end{theorem}

\begin{proof}
    See Appendix~\ref{Appendix:proof Theorem2}.
\end{proof}

The proof of Theorem~\ref{Theorem2} uses the following posterior concentration property: for every policy class $\mathcal G$ satisfying Assumption~\ref{Assumption:VC} and every $\varepsilon > 0$,
\tightdisplay{
    \Pi\left( P:\rho_{\mathcal G}(P_0,P)>\varepsilon \mid \mathcal D_n \right) \stackrel{P_0}{\to} 0.
}
Suppose first that $W_{\mathcal G_1}^\star(P_0)\neq W_{\mathcal G_2}^\star(P_0)$. Then an incorrect posterior ranking can occur only if, for at least one class $\mathcal G_i$, the posterior optimal welfare $W_{\mathcal G_i}^\star(P)$ deviates sufficiently from $W_{\mathcal G_i}^\star(P_0)$. Since $|W_{\G_i}^{\star}(P_0) - W_{\G_i}^{\star}(P)| \le \rho_{\G_i}(P_0,P)$, such a mistake requires at least one of the discrepancy terms $\rho_{\mathcal G_i}(P_0,P)$ to be non-negligible. The concentration result above then implies that the posterior probability of an incorrect ranking converges to zero. The equality case follows from the same argument.

As a corollary of Theorem~\ref{Theorem2}, the posterior model-selection procedure is also pointwise consistent for comparing any two fixed treatment rules. Specifically, if $W(P_0;G_1)\neq W(P_0;G_2)$, the posterior probability of correctly ranking the two rules converges to one in $P_0$-probability.
This covers policy-ranking problems involving prespecified rules \citep[e.g.,][]{kasy2016partial,yang2022offline}, such as comparisons between rules proposed by different analysts or between a candidate targeting rule and the no-treatment status quo.
The result follows by applying Theorem~\ref{Theorem2} to the singleton classes $\G_1=\{G_1\}$ and $\G_2=\{G_2\}$.

\begin{corollary}
\label{Corollary1}

Let $G_1, G_2$ be two treatment rules.
Suppose Assumptions~\ref{Assumption:DGP}--\ref{Assumption:base measure} hold.
Let $\mathcal{D}_n \coloneqq \{D_i\}_{i=1}^n \stackrel{\mathrm{i.i.d.}}{\sim} P_0$ denote the observed sample, and let $\Pi(P \in \cdot \mid \D_n)$ denote the Dirichlet process posterior $P \mid \D_n \sim \mathrm{DP}(\alpha + \sum_{i=1}^n \delta_{D_i})$.

\begin{enumerate}[label=(\roman*)]
    \item If $W(P_0; G_1) \neq W(P_0; G_2)$, then, as $n\to\infty$,
    \tightdisplay{
        \Pi\left(
        P:\,
        \sign\,\bigl(W(P; G_1) - W(P; G_2)\bigl)
        \,=
        \sign\,\bigl(W(P_0; G_1) - W(P_0; G_2)\bigr)
        \,\middle|\, \D_n
        \right)
        \stackrel{P_0}{\to} 1.
    }

    \item If $W(P_0; G_1) = W(P_0; G_2)$, then for any $\varepsilon>0$, as $n\to\infty$,
    \tightdisplay{
        \Pi\left(
        P:\,
        \bigl| W(P; G_1) - W(P; G_2) \bigr| < \varepsilon
        \,\middle|\, \D_n
        \right)
        \stackrel{P_0}{\to} 1.
    }
\end{enumerate}

\end{corollary}

\subsection{NBPL in Relation to Existing Approaches}
\label{theory:comparison with literature}

This section relates NBPL to three policy learning approaches: empirical welfare maximization \citep[EWM;][]{kitagawa2018should}, the limited-information Bayes (LIB) rule of \citet{chamberlain2011bayesian}, and the conditional empirical success (CES) rule of \citet{manski2004statistical}.

\subsubsection{Comparison with EWM}

I begin with EWM. EWM selects a treatment rule by plugging in the empirical distribution $\Prob_n$ in place of the unknown population distribution $P_0$:
\tightdisplay{
    \widehat G_{\mathrm{EWM}}
    \in
    \arg\max_{G\in\mathcal G} W(\mathbb P_n;G).
}
In contrast,  NBPL accommodates traditional policy-choice analysis through a class of \textit{Posterior Expected Welfare Maximization} (PEWM) rules defined by\footnote{Equivalently, PEWM minimizes posterior expected welfare loss, where $L(P;G)\coloneqq W_{\mathcal G}^{\star}(P)-W(P;G)$.}
\tightdisplay{
    \widehat{G}_{\mathrm{PEWM}}(\alpha) :\in \arg\max_{G \in \mathcal{G}} \int W(P;G) \, d\Pi(P \mid \mathcal{D}_n),
}
where the notation makes the dependence on the base measure $\alpha$ explicit.
When $|\alpha|>0$, let
$\widetilde\alpha\coloneqq\alpha/|\alpha|$. Proposition~\ref{Proposition1} shows that PEWM is equivalent to maximizing welfare by plugging in the prior-regularized empirical distribution $P_{n,\alpha}$ in place of $P_0$:
\tightdisplay{
    \widehat{G}_{\mathrm{PEWM}}(\alpha) \in
    \arg\max_{G \in \mathcal{G}} W(P_{n,\alpha};G), \quad \text{where } P_{n,\alpha} \coloneqq \frac{|\alpha|}{|\alpha|+n}\widetilde\alpha
    +
    \frac{n}{|\alpha|+n}\mathbb P_n.
}
Thus, PEWM defines a class of prior-regularized treatment rules that embeds EWM when $|\alpha| = 0$.

\begin{proposition}
\label{Proposition1}

Suppose $P\mid\mathcal D_n \sim \mathrm{DP}(\alpha+n\mathbb P_n)$. For $|\alpha|>0$, let $\widetilde\alpha\coloneqq\alpha/|\alpha|$.
Then any PEWM rule satisfies
\tightdisplay{
    \widehat G_{\mathrm{PEWM}}(\alpha)
    \in
    \arg\max_{G\in\mathcal G}
    \left\{
        \frac{|\alpha|}{|\alpha|+n}W(\widetilde\alpha;G)
        +
        \frac{n}{|\alpha|+n}W(\mathbb P_n;G)
    \right\}. 
    \tag{4.5} \label{main:eq10}
}
As the prior mass vanishes, PEWM reduces to EWM: $\widehat G_{\mathrm{PEWM}}(0) = \widehat G_{\mathrm{EWM}}$.

\end{proposition}

\begin{proof}
    See Appendix~\ref{Appendix:proof Proposition1}.
\end{proof}

Equation~\eqref{main:eq10} shows that PEWM rules solve a weighted welfare maximization problem that combines prior information with the data and coincides with EWM when $|\alpha| = 0$. Consequently, it is natural to wonder how $\widehat G_{\mathrm{PEWM}}(\alpha)$ behaves under the welfare-regret  criterion of \cite{kitagawa2018should}. To this end, define its welfare regret by
\tightdisplay{
    \mathcal{R}_n(\alpha)
    \coloneqq
    \Exp_{P_0}\!\left[
        W_{\mathcal G}^{\star}(P_0)
        -
        W\!\left(P_0;\widehat G_{\mathrm{PEWM}}(\alpha)\right)
    \right].
}
Let $\mathcal P(M)$ denote the bounded-outcome subclass of
Assumption~\ref{Assumption:DGP}, consisting of distributions
$\mathbf P_0^\star$ satisfying Assumption~\ref{Assumption:DGP}(a)--(c) and $|Y|\leq M$ almost surely. This class corresponds to the distributions satisfying Assumption~2.1 of \citet{kitagawa2018should}; the bounded-outcome restriction is imposed only to obtain uniformity.
Theorem~\ref{Theorem3} establishes minimax-optimal welfare regret uniformly over $\mathcal P(M)$ for the class of PEWM rules. Because this class includes EWM as the special case $|\alpha|=0$, the theorem also recovers the uniform minimax optimality of EWM \citep{kitagawa2018should}.

\begin{theorem}
\label{Theorem3}

Suppose Assumptions~\ref{Assumption:base measure}--\ref{Assumption:VC} hold. 
Fix $M<\infty$, and let $\mathcal P(M)$ be as defined above. 
Then there exists a constant $C > 0$ such that
\tightdisplay{
    \sup_{\mathbf{P}_0^{\star} \in \mathcal{P}(M)} \mathcal{R}_n(\alpha) \leq C \sqrt{\frac{v}{n}}.
    \tag{4.6} \label{main:eq11}
}
    
\end{theorem}

\begin{proof}

See Appendix~\ref{Appendix:proof Theorem3}.
\end{proof}

\subsubsection{Comparison with Chamberlain (2011) and CES}

I now compare NBPL with the LIB rule of \citet{chamberlain2011bayesian} and, by extension, the CES rule of \citet{manski2004statistical}.
To match their analyses, I restrict attention to finite covariate support, $\mathcal X=\{1,\ldots,K\}$, and an unconstrained policy class, $\mathcal G=2^{\mathcal X}$. I also extend \citet{chamberlain2011bayesian}'s setup by allowing outcomes to be continuous, rather than finitely supported.

Fix any $\tau\in\mathcal X$. The LIB rule assigns treatment to individuals with $X=\tau$ if 
\tightdisplay{
    \frac{\int y d\beta_{1,\tau}(y) + \sum_{i=1}^n Y_i T_i \mathbf{1}\{X_i=\tau\}}
    {|\beta_{1,\tau}| + n(1,\tau)}
    >
    \frac{\int y d\beta_{0,\tau}(y) + \sum_{i=1}^n Y_i (1-T_i) \mathbf{1}\{X_i=\tau\}}
    {|\beta_{0,\tau}| + n(0,\tau)},
    \tag{4.7} \label{main:eq12}
}
where $n(t,\tau) \coloneqq \sum_{i=1}^n \mathbf{1}\{T_i=t, X_i = \tau\}$, $n(\tau) \coloneqq \sum_{i=1}^n \mathbf{1}\{X_i=\tau\}$, and $\{\beta_{t,\tau}: t\in\{0,1\},\tau\in\mathcal X\}$ are the base measures
of the cellwise independent Dirichlet process priors.
Analogously to the relationship between PEWM and EWM, LIB defines a class of prior-regularized rules that embeds CES when $|\beta_{t,\tau}| = 0$. To see this, recall that the CES rule assigns treatment to $X=\tau$ if
\tightdisplay{
    \frac{1}{n}\sum_{i=1}^n
    \frac{Y_iT_i}{n(1,\tau)/n(\tau)}\mathbf 1\{X_i=\tau\}
    >
    \frac{1}{n}\sum_{i=1}^n
    \frac{Y_i(1-T_i)}{n(0,\tau)/n(\tau)}\mathbf 1\{X_i=\tau\}.
    \tag{4.8} \label{main:eq13}
}

Under the same setup, the PEWM rule assigns treatment to $X=\tau$ if
\tightalign{
    &\frac{|\alpha|}{|\alpha|+n}
    \Exp_{\widetilde{\alpha}}
    \left[ \frac{YT}{e(X)} \mathbf 1\{X=\tau\}\right]
    +
    \frac{1}{|\alpha|+n}
    \sum_{i=1}^n
    \frac{Y_iT_i}{e(X_i)}
    \mathbf 1\{X_i=\tau\}
    \\
    > \,\,
    &\frac{|\alpha|}{|\alpha|+n}
    \Exp_{\widetilde{\alpha}}
    \left[\frac{Y(1-T)}{1-e(X)} \mathbf 1\{X=\tau\}\right]
    +
    \frac{1}{|\alpha|+n}
    \sum_{i=1}^n
    \frac{Y_i(1-T_i)}{1-e(X_i)}
    \mathbf 1\{X_i=\tau\}.
    \tag{4.9} \label{main:eq14}
}
Letting the prior mass vanish yields the EWM rule: it assigns treatment to $X=\tau$ if
\tightdisplay{
    \frac{1}{n}
    \sum_{i=1}^n
    \frac{Y_iT_i}{e(X_i)}
    \mathbf 1\{X_i=\tau\}
    >
    \frac{1}{n}
    \sum_{i=1}^n
    \frac{Y_i(1-T_i)}{1-e(X_i)}
    \mathbf 1\{X_i=\tau\}.
    \tag{4.10} \label{main:eq15}
}
This is a concrete instance of how EWM arises as the limiting case of PEWM as $|\alpha| \to 0$. Moreover, comparing \eqref{main:eq13} with \eqref{main:eq15} shows that, in this setup, CES coincides with EWM when EWM uses the empirical propensity score.

The LIB rule of \citet{chamberlain2011bayesian} differs from PEWM because they are induced by different priors. Chamberlain starts from a prior on the joint distribution of $(Y(1),Y(0),T,X)$,\footnote{This joint distribution can be decomposed into the conditional distribution of potential outcomes given $(T,X)$, the treatment assignment mechanism $T\mid X$, and the marginal distribution of $X$. Chamberlain specifies hierarchical priors on these components.}
and treats individuals with $X=\tau$ when the posterior predictive mean of $Y(1)$ exceeds that of $Y(0)$ conditional on $X=\tau$.
Its limited-information feature is that the predictive distribution of $Y(t)$ is formed after discarding outcomes from observations with $T_i=1-t$, while retaining their treatment and covariate values. 
The key distinction is therefore the prior specification: Chamberlain places a structural prior on latent potential outcomes, whereas NBPL places a reduced-form prior directly on the welfare-relevant distribution.
Moreover, NBPL accommodates continuous covariates in a fully nonparametric way and incorporates the known propensity score directly through the reduced-form prior. Establishing comparable regret guarantees for structural priors, as in \citet{chamberlain2011bayesian}, are likely more challenging due to nontrivial semiparametric bias concerns under Bayesian ignorability.\footnote{See, for example, \citet{ritov2014bayesian}, \citet{ray2020semiparametric}, and \citet{breunig2025double}.}

Despite these finite-sample discrepancies, it turns out that all four rules are asymptotically equivalent to the oracle cellwise rule, which treats a covariate cell if and only if its conditional average treatment effect is positive. Proposition~\ref{Proposition2} formalizes this equivalence.

\begin{proposition}
\label{Proposition2}

Suppose Assumptions~\ref{Assumption:DGP}--\ref{Assumption:base measure} hold. 
Let $\mathcal X=\{1,\ldots,K\}$ and 
$\mathcal G=2^{\mathcal X}$. 
Define $\mathrm{CATE}(\tau) $ $\coloneqq \Exp_{\mathbf{P}_0^\star}[Y(1)-Y(0)\mid X=\tau]$ and assume the no-tie condition: $\mathrm{CATE}(\tau) \neq 0$ for all $\tau \in \mathcal{X}$.
Suppose each base measure $\beta_{t,\tau}$ on $\mathbb R$, for $t\in\{0,1\}$ and $\tau\in\mathcal X$, has finite total mass and finite first absolute moment.
Then, for every $\tau\in\mathcal X$, as $n\to\infty$, the PEWM rule~\eqref{main:eq14}, the EWM rule~\eqref{main:eq15}, the LIB rule~\eqref{main:eq12}, and the CES rule~\eqref{main:eq13} all converge almost surely to $\mathbf{1}\{\mathrm{CATE}(\tau) > 0\}$.

\end{proposition}

\begin{proof}

See Appendix~\ref{Appendix:proof Proposition2}.
\end{proof}


\section{Conclusion}
\label{sec:conclusion}

This paper proposes the \textit{Nonparametric Bayesian Policy Learning} (NBPL) framework. The key observation is that welfare is fully determined by a reduced-form data distribution. Consequently, uncertainty about optimal policies entirely reflects uncertainty about this distribution. 
This naturally motivates a Bayesian approach in which the researcher places a nonparametric Dirichlet process prior on the reduced-form distribution. 
The resulting posterior provides a unified framework for traditional policy choice analysis, inference on optimal welfare and optimal treatment assignments, and model selection across competing policy classes.
NBPL is also computationally tractable: its default Bayesian bootstrap implementation requires only exponential reweighting of the observations.
The paper makes two main theoretical contributions. First, it establishes a minimax-optimal convergence rate for posterior welfare regret. This result provides a novel policy-relevant analogue of posterior contraction rates, a central notion in the theoretical statistics literature over the past two decades. Second, it establishes consistency of the posterior model-selection procedure.
Several extensions are left for future work. These include allowing for unknown propensity scores, accommodating partially identified welfare criteria, and studying whether minimax-optimal regret gurantees can be obtained under fully structural priors.

\pagebreak
\bibliographystyle{ecta}
\bibliography{references}  


\appendix
\renewcommand{\thesection}{\Alph{section}}
\renewcommand{\thesubsection}{\Alph{section}.\arabic{subsection}}


\pagebreak

\section{Monte Carlo Simulation}
\label{Appendix:simulation}

\subsection{Simulation Design}

To construct a data-generating process calibrated to the dataset analyzed by \citet{bhattacharya2012inferring}, I first estimate conditional choice probabilities from the original sample $\mathcal D_{N}=\{(Y_i,T_i,X_i)\}_{i=1}^{N}$. I then generate $S=1{,}000$ independent simulated samples for each sample size $n \in\{100,250,500,$ $1000,2000,5000\}$.
For each $n$ and $s=1,\ldots,S$, let $ \widetilde{\mathcal D}_n^{[s]} = \{(\widetilde Y_i^{[s]},\widetilde T_i^{[s]}, \widetilde X_i^{[s]})\}_{i=1}^n$ denote the simulated sample.

\textbf{Step 1.} Fit random-forest classifiers of the binary outcome $Y$ on $X$ separately within each treatment arm $t\in\{0,1\}$ in the original sample \citep{breiman2001random}. 
Let $\widehat m_{N}(t,x) \coloneqq \widehat{\Pr}_{N}(Y=1\mid T=t,X=x)$ denote the predicted probability of bednet usage, where $\widehat m_{N}(t,x)\in[0,1]$ by construction.

\textbf{Steps 2--4.} For each $n\in\{100,250,500,1000,2000,5000\}$ and $s=1,\ldots,S$:

\begin{itemize}
    \item \textbf{Step 2.} Generate covariates
    $\{\widetilde X_i^{[s]}\}_{i=1}^n$ by resampling with replacement from $\widehat F_X \coloneqq N^{-1} \sum_{i=1}^N\delta_{X_i}$.

    \item \textbf{Step 3.} Generate treatment assignments $\{\widetilde T_i^{[s]}\}_{i=1}^n \overset{\mathrm{i.i.d.}}{\sim} \mathrm{Bernoulli}(e)$,
    independently of $\widetilde X_i^{[s]}$, where $e = 0.16$ is the known propensity score in the original experiment.

    \item \textbf{Step 4.} Generate potential outcomes independently
    conditional on $\widetilde X_i^{[s]}$:
    \tightdisplay{
        \widetilde Y_i^{[s]}(t)
        \mid \widetilde X_i^{[s]}
        \overset{\mathrm{ind}}{\sim}
        \mathrm{Bernoulli}\!\left(
        \widehat m_N(t,\widetilde X_i^{[s]})
        \right),
        \qquad t\in\{0,1\}.
    }
    The observed outcome is then $\widetilde Y_i^{[s]} = \widetilde Y_i^{[s]}(1)\widetilde T_i^{[s]} + \widetilde Y_i^{[s]}(0) (1 - \widetilde T_i^{[s]})$.
    
\end{itemize}

\subsection{Analysis Procedure}

For each simulated dataset $\widetilde{\mathcal D}_n^{[s]}$, I compare a linear class $\mathcal G_{\mathrm{lin}}$ with a depth-$2$ decision-tree class $\mathcal G_{\mathrm{tree},2}$. Hereafter, let $\mathcal G$ denote either class.

\paragraph{Solving EWM and NBPL.}
Both EWM and NBPL are computed as weighted welfare maximization problems. Let $q=(q_1,\ldots,q_n)$ be any weight vector with $q_i\geq0$ and $\sum_{i=1}^n q_i=1$. 
Following \citet[Remark 2.3]{kitagawa2018should}, define
\tightdisplay{
    g_i^{[s]}(q)
    =
    q_i
    \big(\widetilde Y_i^{[s]}-\bar Y_q^{[s]}\big)
    \left(
    \frac{\widetilde T_i^{[s]}}{e}
    -
    \frac{1-\widetilde T_i^{[s]}}{1-e}
    \right), 
    \qquad 
    \bar Y_q^{[s]} \coloneqq \sum_{i=1}^n q_i \widetilde Y_i^{[s]}.
}
For a given $q$, the selected rule is
\tightdisplay{
    G^{[s]}(q) \in \arg\max_{G\in\mathcal G} \,\, \sum_{i=1}^n g_i^{[s]}(q)\mathbf 1\{\widetilde X_i^{[s]}\in G\}.
}
EWM and NBPL differ in the choice of $q$:
\begin{itemize}
    \item EWM uses the uniform weights $q_i^{\mathrm{EWM}} =1/n$. The resulting rule is $\widehat G_{\mathrm{EWM}}^{[s]} \coloneqq G^{[s]}(q^{\mathrm{EWM}})$.

    \item NBPL uses Bayesian-bootstrap weights. For each simulation draw $s$ and posterior draw $b=1,\ldots,B$, with $B=1{,}000$, draw $\xi_i^{[s,b]}\overset{\mathrm{i.i.d.}}{\sim}\mathrm{Exp}(1)$ ($i=1,\ldots,n$) and set $\omega_i^{[s,b]} = \xi_i^{[s,b]} / \sum_{j=1}^n \xi_j^{[s,b]}$.
    Writing $\omega^{[s,b]} \coloneqq (\omega_1^{[s,b]}, \ldots, \omega_n^{[s,b]})$, the $b$th posterior draw of the NBPL optimal rule is $G^{[s,b]} \coloneqq G^{[s]}(\omega^{[s,b]})$
    and the corresponding posterior draw of optimal welfare is
    $W_{\mathcal G}^{\star,[s,b]} = \sum_{i=1}^n g_i^{[s]}(\omega^{[s,b]}) \mathbf 1\{\widetilde X_i^{[s]}\in G^{[s,b]}\}$.
\end{itemize}

\paragraph{Oracle welfare and regret.}
Under the simulated DGP, $\mathbb E[\widetilde Y(t)\mid \widetilde X=x] = \widehat m_{N}(t,x)$ for $t\in\{0,1\}$.
Therefore, the simulated CATE is $\tau(x) = \widehat m_{N}(1,x)-\widehat m_{N}(0,x)$, and the superpopulation welfare under any rule $G$ is given by
\tightdisplay{
    W(G) \coloneqq \frac{1}{N} \sum_{i=1}^{N} \tau(X_i)\mathbf 1\{X_i\in G\}.
}
For each policy class $\mathcal G$, let $G^\star\in\arg\max_{G\in\mathcal G}W(G)$ denote an oracle rule and define the corresponding oracle welfare by $W_{\mathcal G}^\star\coloneqq\sup_{G\in\mathcal G}W(G)$.

I conduct two regret exercises. First, I compare EWM regret with the posterior mean of NBPL regret across simulated sample sizes. The EWM regret in Monte Carlo sample $s$ is $R_{\mathrm{EWM}}^{[s]} \coloneqq W_{\mathcal G}^{\star} -
 W(\widehat G_{\mathrm{EWM}}^{[s]})$.
For NBPL, posterior draw $b$ has regret $R_{\mathrm{NBPL}}^{[s,b]} \coloneqq W_{\mathcal G}^{\star} - W(G^{[s,b]})$, which I summarize using the empirical posterior mean $\bar R_{\mathrm{NBPL}}^{[s]} \coloneqq B^{-1} \sum_{b=1}^B R_{\mathrm{NBPL}}^{[s,b]}$.
The posterior mean regret is a useful summary because its convergence to zero is sufficient for posterior regret to concentrate at zero.\footnote{By Markov's inequality, for any $C>0$, $\Pi ( R(P_0;P)>C \mid \widetilde{\mathcal D}_n^{[s]} ) \leq \mathbb{E}[R(P_0;P) \mid \widetilde{\mathcal D}_n^{[s]}] / C$.} 
I therefore examine whether posterior mean regret shrinks toward zero as the sample size increases.

Second, I summarize NBPL posterior regret using empirical posterior CDFs averaged across Monte Carlo samples. 
For each Monte Carlo sample $s$ and regret threshold $r$, define 
\tightdisplay{
    \widehat F_n^{[s]}(r) \coloneqq \frac{1}{B} \sum_{b=1}^B \mathbf{1} \{ R_{\mathrm{NBPL}}^{[s,b]}\leq r \}.
}
I report the Monte Carlo average $\bar{F}_n(r) \coloneqq S^{-1} \sum_{s=1}^S \widehat F_n^{[s]}(r)$
over a grid of thresholds $r$.

\paragraph{Model selection.}
Under the simulated DGP, the oracle welfare gap is
$W_{\G_{\mathrm{tree},2}}^{\star} -W_{\G_{\mathrm{lin}}}^{\star}=0.0004$.
Thus, the decision-tree class yields a small increase in optimal welfare, corresponding to approximately $0.04$ percentage points in bednet usage. The simulation therefore falls under case~(i) of Theorem~\ref{Theorem2}, although the welfare separation is small. For Monte Carlo sample $s$, define $\pi_n^{[s]} \coloneqq \Pi ( W_{\G_{\mathrm{tree},2}}^{\star}(P) >  W_{\G_{\mathrm{lin}}}^{\star}(P) \mid \widetilde{\mathcal D}_n^{[s]} )$.
Theorem~\ref{Theorem2} predicts that the Monte Carlo distribution of $\{\pi_n^{[s]}\}_{s=1}^S$ should concentrate toward one as $n$ increases. I examine this prediction across simulated sample sizes.


\subsection{Simulation Results}

Figures~\ref{fig:sim ewm vs nbpl regret linear} and
\ref{fig:sim ewm vs nbpl regret tree} summarize EWM regret and posterior mean NBPL regret across simulations, for the linear class and the depth-$2$ decision-tree class, respectively. 
Overall, EWM and NBPL perform similarly under the welfare-regret criterion, with only small differences between the two policy classes. In smaller samples, posterior mean NBPL regret is less dispersed than EWM regret, especially in the upper tail, but tends to be slightly higher near the center of the distribution. This pattern resembles the bias--variance tradeoff for shrinkage estimators. As the sample size increases, both EWM regret and posterior mean NBPL regret shrink toward zero and become more concentrated across simulations, consistent with the regret convergence theory for both methods.

Figure~\ref{fig:sim nbpl regret distribution} plots the Monte Carlo-averaged empirical posterior CDFs of NBPL welfare regret at each simulated sample size for the linear and depth-$2$ decision-tree classes. The regret distributions are similar across policy classes, indicating comparable learning performance relative to their respective oracle rules. This similarity may reflect the simple covariate structure of the application. As the sample size increases, the regret CDFs concentrate increasingly near zero, consistent with the NBPL regret contraction theory. At $n=5{,}000$, posterior regret is concentrated below $0.03$ for both classes, about $6\%$ of the maximum attainable welfare within each class.

Figure~\ref{fig:sim nbpl model selection} summarizes, across Monte Carlo samples, the posterior probability that the depth-$2$ decision-tree class attains higher optimal welfare than the linear class.
Under the simulated DGP, the decision-tree class has higher oracle welfare, although the gap is only $0.0004$. 
Even at $n=100$, the 25th percentile of the posterior probability exceeds $0.5$, and the median exceeds $0.9$. 
As the sample size increases, the distribution becomes more concentrated near one, consistent with case (i) of Theorem~\ref{Theorem2}. 
The gradual concentration over the range of sample sizes considered may reflect the small oracle welfare gap.

\begin{figure}[htbp]
\centering
\includegraphics[width=0.8\textwidth]{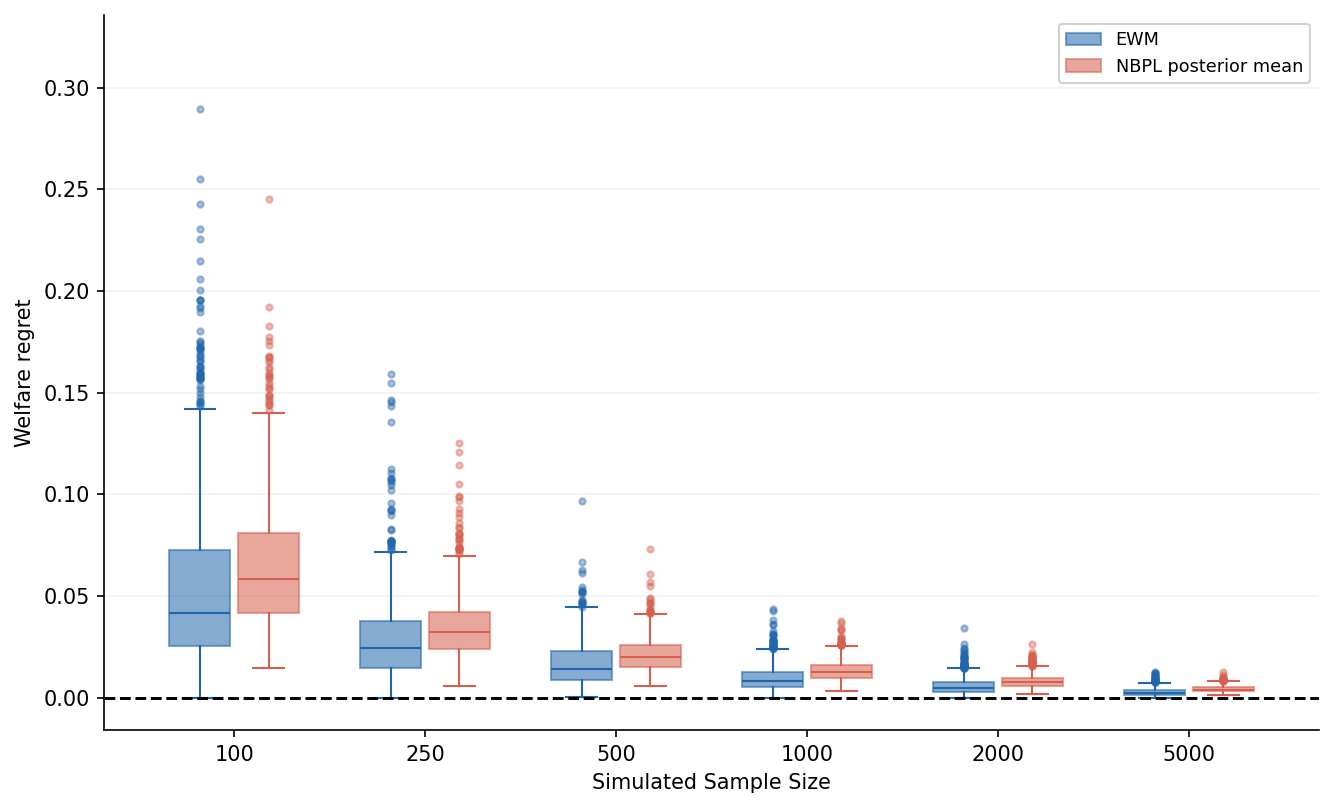}
\caption[NBPL versus EWM Welfare Regret, Linear Class]{NBPL versus EWM welfare regret for the linear class. For each simulated sample size $n$, the blue boxplot summarizes EWM regret $R_{\mathrm{EWM}}^{[s]}$ across the $S=1{,}000$ Monte Carlo samples. The red boxplot summarizes posterior mean NBPL regret $\bar R_{\mathrm{NBPL}}^{[s]}  = B^{-1} \sum_{b=1}^B R_{\mathrm{NBPL}}^{[s,b]}$ across the same samples, with each posterior mean computed from $B=1{,}000$ draws.}
\label{fig:sim ewm vs nbpl regret linear}
\end{figure}

\begin{figure}[htbp]
\centering
\includegraphics[width=0.8\textwidth]{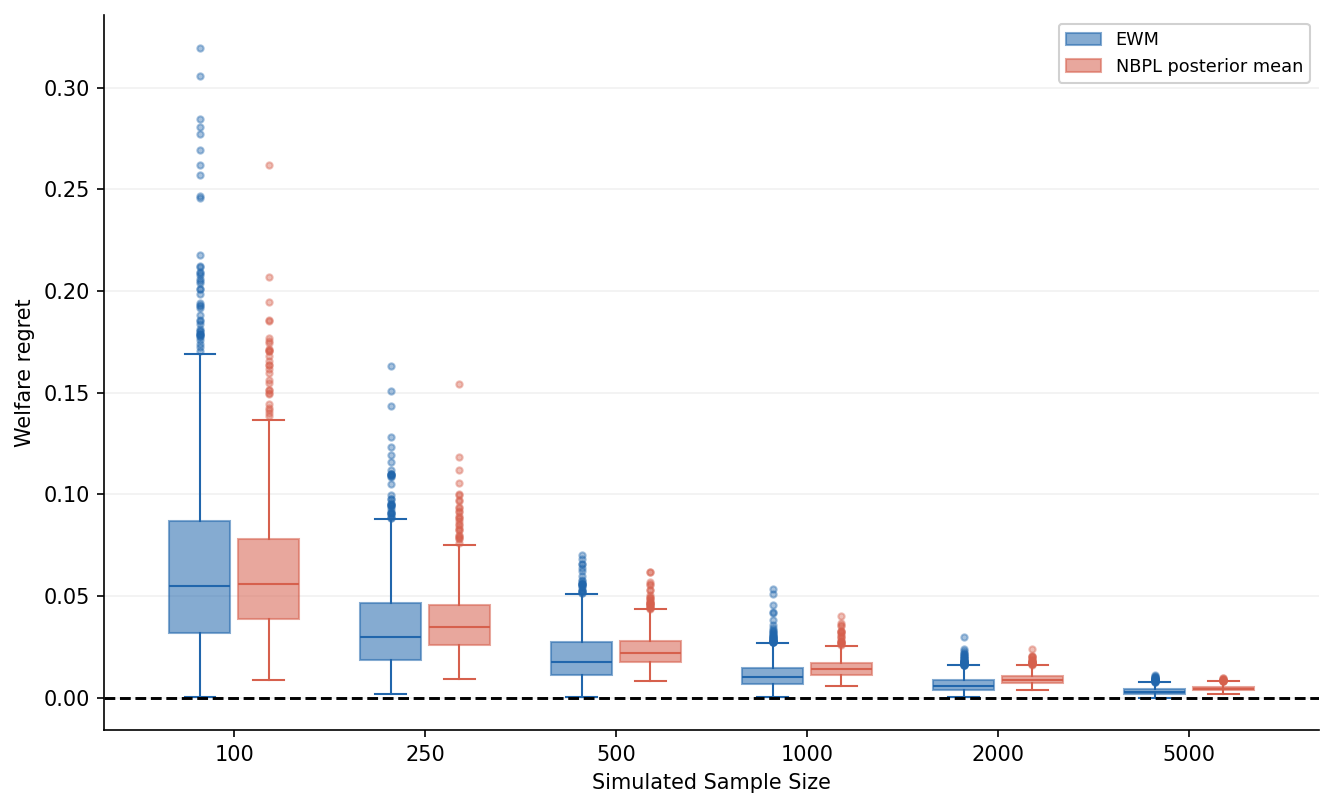}
\caption[NBPL versus EWM Welfare Regret, Tree Class]{NBPL versus EWM welfare regret for the depth-$2$ decision-tree class. For each simulated sample size $n$, the blue boxplot summarizes EWM regret $R_{\mathrm{EWM}}^{[s]}$ across the $S=1{,}000$ Monte Carlo samples. The red boxplot summarizes posterior mean NBPL regret $\bar R_{\mathrm{NBPL}}^{[s]}  = B^{-1} \sum_{b=1}^B R_{\mathrm{NBPL}}^{[s,b]}$ across the same samples, with each posterior mean computed from $B=1{,}000$ draws.}
\label{fig:sim ewm vs nbpl regret tree}
\end{figure}

\begin{figure}[htbp]
\centering
\includegraphics[width=1\textwidth]{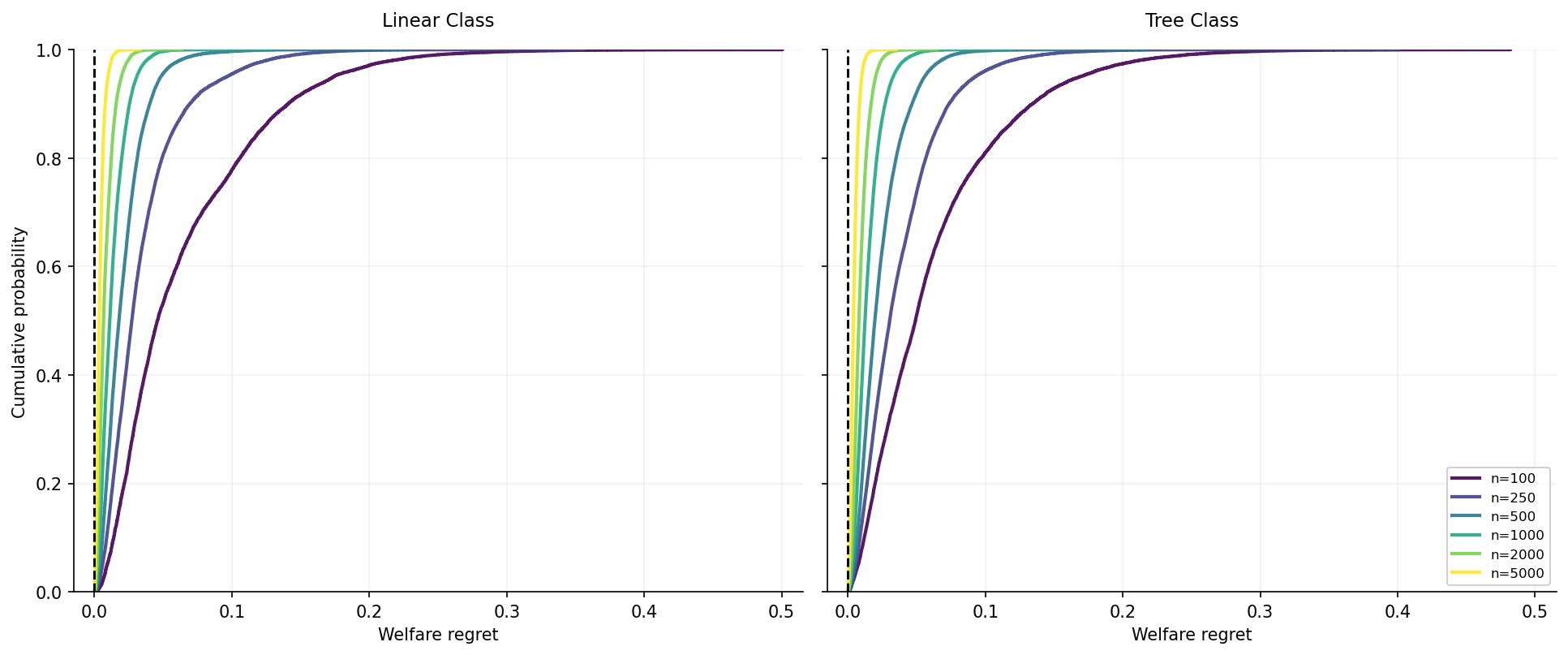}
\caption[NBPL Welfare Regret Distribution]{NBPL welfare regret distribution.
The left and right panels report results for $\G_{\mathrm{lin}}$ and $\G_{\mathrm{tree},2}$, respectively. For each simulated sample size $n$, the curves plot the Monte Carlo-averaged posterior CDF $\bar{F}_n(r) \coloneqq S^{-1} \sum_{s=1}^S \widehat F_n^{[s]}(r)$ over a grid of regret thresholds $r$, where $\widehat F_n^{[s]}(r) \coloneqq B^{-1} \sum_{b=1}^B \mathbf{1} \{ R_{\mathrm{NBPL}}^{[s,b]}\leq r \}$
is the empirical posterior CDF in Monte Carlo sample $s$. 
The averages are based on $S=1{,}000$ Monte Carlo samples.}
\label{fig:sim nbpl regret distribution}
\end{figure}

\begin{figure}[htbp]
\centering
\includegraphics[width=0.7\textwidth]{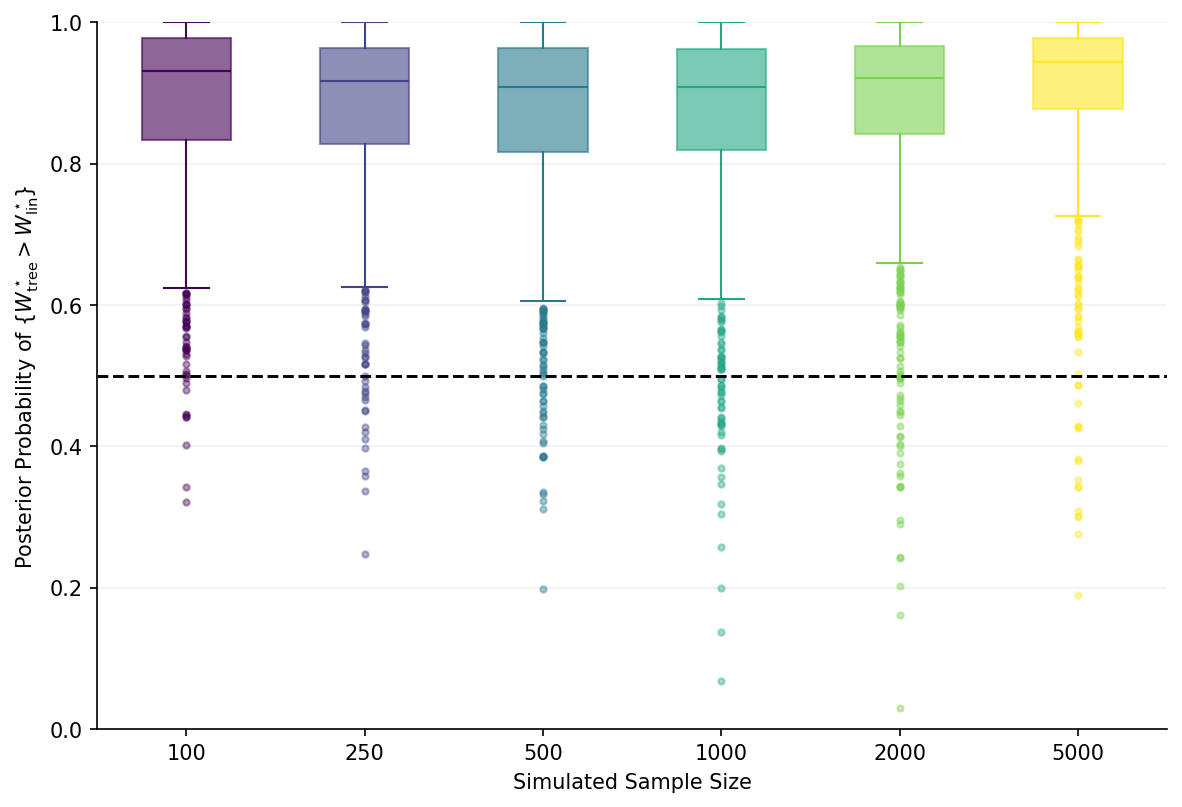}
\caption[Posterior Model Selection in Simulation]{Posterior model selection in the simulation. For each simulated sample size $n$, the boxplot summarizes the posterior probabilities
$\pi_n^{[s]} \coloneqq \Pi( W^\star_{\G_{\mathrm{tree},2}}(P) > W^\star_{\G_{\mathrm{lin}}}(P) \mid \widetilde{\mathcal D}_n^{[s]})$ across the $S=1{,}000$ Monte Carlo samples.
The dashed horizontal line marks probability $0.5$. 
Since the oracle welfare under $\G_{\mathrm{tree},2}$ is
slightly higher than under $\G_{\mathrm{lin}}$, upward movement of the boxplots indicates increasing posterior support for the better policy class.}
\label{fig:sim nbpl model selection}
\end{figure}


\pagebreak
\section{Proofs of Main Results}
\label{Appendix:proofs}

\subsection{Notations}
\label{Appendix:notation}

I use the following empirical-process notation. Let
\tightdisplay{
    \mathcal F
    =
    \{f_G(D):G\in\G\},
    \qquad
    f_G(D)
    =
    \left\{
    \frac{YT}{e(X)}
    -
    \frac{Y(1-T)}{1-e(X)}
    \right\}
    \mathbf 1\{X\in G\}.
}
By \citet[Lemma~A.1]{kitagawa2018supp}, $\mathcal F$ is a VC-subgraph class with VC-dimension at most $v\coloneqq \mathrm{VC}(\G)$. In what follows, I suppress the
dependence of $f$ on $G$.

Let $F$ be an envelope for $\mathcal F$, i.e., $|f(d)|\le F(d)$ for all $f\in\mathcal F$ and $d\in\mathcal D$. 
In the present setting, one may take $F(D) = |YT/e(X)| + |Y(1-T)/(1-e(X))|$. Under Assumptions~\ref{Assumption:DGP}--\ref{Assumption:base measure}, $\Exp_{P_0}[F^2] < \infty$ and $\int F d\alpha < \infty$.
Let $\{D_i\}_{i\geq1}$ be an i.i.d. sequence with common law $P_0$. Write $P_0^\infty\coloneqq\bigotimes_{i=1}^\infty P_0$ for its product law, and write $P_0^{\otimes n}$ for the law of $(D_1,\ldots,D_n)$, viewed as the $n$-dimensional marginal of $P_0^\infty$. When no confusion arises, I use $P_0$ and $\Exp_{P_0}$ for probabilities and expectations under either
$P_0^\infty$ or $P_0^{\otimes n}$.
I use $\stackrel{d}{=}$ to denote equality in law.
For any probability measure $P$ and $f\in\mathcal F$, write $Pf=\Exp_P[f(D)]$ and $\Prob_n f=n^{-1}\sum_{i=1}^n f(D_i)$. Define $\|P\|_{\mathcal F}\coloneqq \sup_{f\in\mathcal F}|Pf|$ and $\|Z\|_{P,r}\coloneqq (\Exp_P|Z|^r)^{1/r}$ for $r>0$. I ignore
measurability subtleties and write $\Exp[\cdot]$ in place of outer expectation. Finally, let $\mathbb G_n=\sqrt n(\Prob_n-P)$ denote the empirical process indexed by $\mathcal F$.

\subsection{Technical Lemmas}
\label{Appendix:lemmas}

This section presents the technical lemmas used in the proofs of the main results. Proofs of Lemmas~\ref{Lemma1}--\ref{Lemma8} are collected in Supplemental Appendix~\ref{supplement:proofs of lemmas}.


\begin{lemma}
\label{Lemma1}

The welfare regret $R(P_0;P)$ in Definition~\ref{def:regret} satisfies
\tightdisplay{
    R(P_0;P) \leq 2 \{ \rho_{\G}(P_0,\Prob_n) + \rho_{\G}(\Prob_n, P) \},
    \tag{B.1} \label{appendix:eq1}
}
where $\rho_{\G}(Q_1,Q_2) \coloneqq \sup_{G \in \G} |W(Q_1;G) - W(Q_2; G)|$ for $Q_1, Q_2$ probability measures.

\end{lemma}


\begin{lemma}
\label{Lemma2}

Suppose Assumptions~\ref{Assumption:DGP}--\ref{Assumption:VC} hold. 
Let $\{D_i\}_{i \geq 1} \stackrel{\mathrm{i.i.d.}}{\sim} P_0$ and $\Prob_n\coloneqq n^{-1}\sum_{i=1}^n\delta_{D_i}$. Let $C_n \to \infty$ (arbitrarily slowly). 
Define $\mathcal{A}_n \coloneqq \left\{ \{D_i\}_{i=1}^n: \rho_{\G}(P_0,\Prob_n) \leq (C_n/4) \sqrt{v/n} \right\}$, where $\rho_{\G}(\cdot,\cdot)$ is defined as in Lemma~\ref{Lemma1}.
Then $P_0^{\infty}(\mathcal{A}_n) \to 1$ as $n \to \infty$.

\end{lemma}

\begin{lemma}[DP Representation]
\label{Lemma3}

Let $P \mid \{D_i\}_{i=1}^n \sim \mathrm{DP}(\alpha + \sum_{i=1}^n \delta_{D_i})$, and write $|\alpha|\coloneqq\alpha(\mathcal D)$.
Let $h:\mathcal D\to\mathbb R$ be any integrable function. \emph{Conditional on} $\{D_i\}_{i=1}^n$, the posterior law of $Ph$ admits the representation
\tightdisplay{
    Ph \stackrel{d}{=} V_n \cdot Qh + (1-V_n) \frac{\sum_{i=1}^n W_i h(D_i)}{\sum_{i=1}^n W_i},
    \tag{B.2} \label{appendix:eq2}
}
where $V_n \sim \mathrm{Beta}(|\alpha|,n)$, $Q \sim \mathrm{DP}(\alpha)$, and $\{W_i\}_{i=1}^n \stackrel{\mathrm{i.i.d.}}{\sim} \mathrm{Exp}(1)$ are mutually independent. Let $\mathbb Q_n$ denote the joint law of $(V_n,Q,W_1,\ldots,W_n)$.

\end{lemma}

\begin{lemma}
\label{Lemma4}

Let $\{D_i\}_{i \geq 1} \stackrel{\mathrm{i.i.d.}}{\sim} P_0$. Conditional on $\{D_i\}_{i=1}^n$, let $\mathbb Q_n$ denote the joint law of $(V_n,Q,W_1,\ldots,W_n)$ specified in Lemma~\ref{Lemma3}.
Let $\F$ be a class of measurable functions with envelope $F$ satisfying $\Exp_{P_0}[F^2] < \infty$ and $\int F d\alpha < \infty$. Then, for $P_0^{\infty}$-almost every realization of $\{D_i\}_{i \geq 1}$,
\tightdisplay{
    \sqrt{n} V_n \,
    \sup_{f \in \F} \left| Qf - \frac{\sum_{i=1}^n W_i f(D_i)}{\sum_{i=1}^n W_i} \right| = o_{\mathbb{Q}_n}(1).
    \tag{B.3} \label{appendix:eq3}
}
    
\end{lemma}

\begin{lemma}[Conditional Multiplier Inequality]
\label{Lemma5}

Let $\mathcal F$ be a class of measurable functions. Fix $n\in\mathbb N$, and let $Z_1,\ldots,Z_n$ be arbitrary stochastic processes indexed by $\mathcal F$. Let $\xi = (\xi_1,\ldots,\xi_n)$ be an exchangeable random vector independent of $(Z_1,\ldots,Z_n)$. Then, conditional on $(Z_1,\ldots,Z_n)$,
\tightdisplay{
    \Exp_{\xi} \left\lVert \frac{1}{\sqrt{n}} \sum_{i=1}^n \xi_i Z_i \right\rVert_{\F} 
    \leq 
    2 \lVert \xi_1 \rVert_{2,1} \cdot \max_{1\leq k \leq n} \Exp_{R} \left\lVert \frac{1}{\sqrt{k}} \sum_{i=1}^k Z_{R_i} \right\rVert_{\F},
    \tag{B.4} \label{appendix:eq4}
}
where $R=(R_1,\ldots,R_n)$ is uniformly distributed over all permutations of $\{1,\ldots,n\}$ and is independent of $(Z_1,\ldots,Z_n)$. Here, $\Exp_\xi$ and $\Exp_R$ denote expectations over $\xi$ and the random permutation $R$, respectively, and $\lVert \xi_1 \rVert_{2,1} := \int_{0}^{\infty} \sqrt{\Prob(|\xi_1|>x)}dx$.
    
\end{lemma}

\begin{lemma}
\label{Lemma6}

Let $\F$ be a class of measurable functions.
Let $\{D_i\}_{i\geq1}\overset{\mathrm{i.i.d.}}{\sim}P_0$, define $\Prob_n\coloneqq n^{-1}\sum_{i=1}^n\delta_{D_i}$, and set $Z_i\coloneqq \delta_{D_i}-\Prob_n$.
Conditional on $(D_1,\ldots,D_n)$, let $R=(R_1,\ldots,R_n)$ be uniformly distributed over all permutations of $\{1,\ldots,n\}$, and let $\widehat D_1,\ldots,\widehat D_n \stackrel{\mathrm{i.i.d.}}{\sim} \mathbb{P}_n$.
For $1\leq k\leq n$, define the multinomial bootstrap process $\widehat{\mathbb{G}}_{n,k} \coloneqq k^{-1/2} \sum_{i=1}^k (\delta_{\widehat{D}_i} - \Prob_n)$. 
Then, conditional on $(D_1,\ldots,D_n)$, for every $1\leq k\leq n$,
\tightdisplay{
    \Exp_{R} \left\lVert \frac{1}{\sqrt{k}} \sum_{i=1}^k Z_{R_i} \right\rVert_{\F} 
    \leq
    \Exp_{\widehat{D}} \left\lVert \widehat{\mathbb{G}}_{n,k} \right\rVert_{\F},
    \tag{B.5} \label{appendix:eq5}
}
where $\Exp_R$ and $\Exp_{\widehat D}$ denote expectations over the random permutation $R$ and the bootstrap sample $(\widehat D_1,\ldots,\widehat D_n)$, respectively.

\end{lemma}

\begin{lemma}
\label{Lemma7}

Let $\F$ be a class of measurable functions.
Let $\{D_i\}_{i \geq 1} \stackrel{\mathrm{i.i.d.}}{\sim} P_0$ , and define $\Prob_n\coloneqq n^{-1}\sum_{i=1}^n\delta_{D_i}$. Conditional on
$(D_1,\ldots,D_n)$, let $\widehat D_1,\ldots,\widehat D_n\stackrel{\mathrm{i.i.d.}}{\sim}\Prob_n$
and, for $1\leq k\leq n$, define $\widehat{\mathbb{G}}_{n,k} \coloneqq k^{-1/2} \sum_{i=1}^k (\delta_{\widehat{D}_i} - \Prob_n)$. Let $\{\varepsilon_i\}_{i \geq 1}$ be an i.i.d sequence of Rademacher variables independent of $\{D_i\}_{i \geq 1}$, and define $U_j \coloneqq \Exp_{\varepsilon} \lVert j^{-1/2} \sum_{i=1}^j \varepsilon_i \delta_{D_i} \rVert_{\F}$. Then, for $P_0^\infty$-almost every realization of $\{D_i\}_{i\geq1}$,
\tightdisplay{
    \limsup_{n \to \infty} \max_{1 \leq k \leq n} \Exp_{\widehat{D}} \left\lVert \widehat{\mathbb{G}}_{n,k} \right\rVert_{\F}
    \leq 
    16 \sqrt{2} \cdot \Exp\left( \sup_{j \geq 1} U_j \right). 
    \tag{B.6} \label{appendix:eq6}
}
Here, $\Exp_{\widehat D}$ and $\Exp_\varepsilon$ denote expectations over the bootstrap sample and the Rademacher variables, respectively, conditional on the original sample, while the expectation on the right-hand side is taken under $P_0^\infty$.

\end{lemma}

\begin{lemma}
\label{Lemma8}

Let $\{D_i\}_{i\geq1}\overset{\mathrm{i.i.d.}}{\sim}P_0$. Let $\mathcal F$ be a measurable VC-subgraph class with VC-dimension $v$ and envelope $F$ satisfying $\|F\|_{P_0,2}<\infty$. 
Let $\{\varepsilon_i\}_{i\geq1}$ be an i.i.d. sequence of Rademacher variables independent of $\{D_i\}_{i\geq1}$, and define $U_j \coloneqq \Exp_\varepsilon \| j^{-1/2} \sum_{i=1}^j \varepsilon_i\delta_{D_i} \|_{\mathcal F}$.
Then there exists a constant $C<\infty$, depending on $P_0$ only through $\|F\|_{P_0,2}$, such that
\tightdisplay{
    \Exp\!\left(\sup_{j\geq 1} U_j\right)
    \leq
    C\sqrt v.
    \tag{B.7}\label{appendix:eq7}
}

\end{lemma}

\subsection{Proof of Theorem~\ref{Theorem1}}
\label{Appendix:proof Theorem1}

\begin{proof}[Proof of Theorem~\ref{Theorem1}]

The proof proceeds in four steps.

\medskip \noindent
\textbf{Step 1: Condition on the event of small $\rho_{\G}(P_0,\Prob_n)$.}

Let $\mathcal{E}_n \coloneqq \{ P: R(P,P_0) > C_n \sqrt{v/n} \}$ and $\mathcal{A}_n \coloneqq \left\{ \{D_i\}_{i=1}^n: \rho_{\G}(P_0,\Prob_n) \leq (C_n/4) \sqrt{v/n} \right\}$.
By Lemma~\ref{Lemma2}, $P_0^\infty(\mathcal A_n)\to 1$. Therefore, it suffices to prove \eqref{main:eq6} along data realizations $\{D_i\}_{i \geq 1}$ such that $\D_n \coloneqq \{D_i\}_{i=1}^n \in \mathcal A_n$.
To justify this reduction, decompose
\tightalign{
    \Pi \left( \mathcal{E}_n \mid \{D_i\}_{i=1}^n \right) 
    & = 
    \Pi \left( \mathcal{E}_n \mid \{D_i\}_{i=1}^n \right) \mathbf{1}\{\mathcal{A}_n\} 
    +
    \Pi \left( \mathcal{E}_n \mid \{D_i\}_{i=1}^n \right) \mathbf{1}\{\mathcal{A}_n^c\} \\
    & \leq 
    \Pi \left( \mathcal{E}_n \mid \{D_i\}_{i=1}^n \right) \mathbf{1}\{\mathcal{A}_n\} 
    + 
    \mathbf{1}\{\mathcal{A}_n^c\}.
}
Taking expectations under $P_0^{\infty}$ and using
$P_0^\infty(\mathcal A_n^c)=o(1)$ gives
\tightalign{
    \mathbb E_{P_0}
    \!\left[
    \Pi(\mathcal E_n \mid \{D_i\}_{i=1}^n)
    \right]
    \le
    \mathbb E_{P_0}
    \!\left[
    \Pi(\mathcal E_n \mid \{D_i\}_{i=1}^n)\mathbf{1}\{\mathcal A_n\}
    \right]
    + o(1).
}
Thus, it remains to show that the first term on the right-hand side is $o(1)$.

By Lemma~\ref{Lemma1}, $R(P_0;P) \leq 2 \{ \rho_{\G}(P_0,\Prob_n)+\rho_{\G}(\Prob_n,P) \}$. On $\mathcal A_n$, one has $\rho_{\G}(P_0,\Prob_n)\leq (C_n/4)\sqrt{v/n}$.
Hence, on $\mathcal A_n$,
\tightdisplay{
    \left\{
    P:R(P_0;P)>C_n\sqrt{\frac{v}{n}}
    \right\}
    \subseteq
    \left\{
    P:\rho_{\G}(\Prob_n,P)>\frac{C_n}{4}\sqrt{\frac{v}{n}}
    \right\}.
}
Therefore, on $\mathcal A_n$,
\tightdisplay{
    \Pi(\mathcal E_n\mid \{D_i\}_{i=1}^n)
    \leq
    \Pi\!\left(P:
    \rho_{\G}(\Prob_n,P)>\frac{C_n}{4}\sqrt{\frac{v}{n}}
    \,\middle|\, \{D_i\}_{i=1}^n
    \right).
}
Thus, it remains to show that the right-hand side is $o_{P_0}(1)$.

\medskip \noindent
\textbf{Step 2: Control $\rho_{\G}(\Prob_n,P)$ by a multiplier empirical process.}

I first isolate the multiplier empirical process used to control $\rho_{\G}(\Prob_n,P)$. Conditional on 
$\{D_i\}_{i=1}^n$, Lemma~\ref{Lemma3} gives, for every $f\in\F$,
\tightalign{
    \sqrt{n} (P - \Prob_n)f \,\, = \,\, & \sqrt{n} V_n \left( Qf - \frac{\sum_{i=1}^n W_i f(D_i)}{\sum_{i=1}^n W_i} \right) \\
    & + \frac{1}{\bar{W}_n} \left(\frac{1}{\sqrt{n}} \sum_{i=1}^n (W_i-1)(f(D_i) - \Prob_n f) \right),
    \tag{B.8} \label{appendix:eq8}
}
where $\bar{W}_n \coloneqq n^{-1}\sum_{i=1}^n W_i$. 
Conditional on $\{D_i\}_{i=1}^n$, let $\mathbb Q_n$ denote the joint law of $(V_n,Q,W_1,\ldots,W_n)$ defined in Lemma~\ref{Lemma3}. Under $\mathbb Q_n$, these variables are mutually independent, with $V_n\sim\mathrm{Beta}(|\alpha|,n)$, $Q\sim\mathrm{DP}(\alpha)$, and $W_1,\ldots,W_n\overset{\mathrm{i.i.d.}}{\sim}\mathrm{Exp}(1)$.
Define the \emph{multiplier empirical process} $\nu_n(f)
\coloneqq n^{-1/2} \sum_{i=1}^n (W_i-1)(f(D_i)-\Prob_n f)$.

Let $\xi_i\coloneqq W_i-1$ denote the centered multipliers. Under $\mathbb Q_n$, the $\xi_i$'s are i.i.d., independent of $(V_n,Q)$, and satisfy $\Exp_{\mathbb Q_n}[\xi_i]=0$ and $\Var_{\mathbb Q_n}(\xi_i)=1$. Setting $Z_i\coloneqq\delta_{D_i}-\Prob_n$, the multiplier empirical process can be written as $\nu_n \coloneqq n^{-1/2} \sum_{i=1}^n \xi_i Z_i$.


Taking the supremum over $f\in\mathcal F$ and applying the triangle inequality yields
\tightdisplay{
    \sqrt{n} \, \rho_{\G}(\Prob_n, P) \leq 
    \underbrace{
    \sqrt{n} V_n \left\lVert Q - \frac{\sum_{i=1}^n W_i \delta_{D_i}}{\sum_{i=1}^n W_i} \right\rVert_{\mathcal{F}}
    }_{\eqqcolon T_{1n} } 
    + 
    \underbrace{\frac{1}{\bar{W}_n} \left\lVert \frac{1}{\sqrt{n}} \sum_{i=1}^n \xi_i Z_i \right\rVert_{\mathcal{F}}
    }_{\eqqcolon T_{2n}}.
    \tag{B.9}  \label{appendix:eq9}
}

\medskip \noindent
\textbf{Step 3: Conditional probability bounds.}

By the union bound and~\eqref{appendix:eq9},
\[
    \Pi\!\left(P:
    \sqrt{n} \, \rho_{\G}(\Prob_n, P) > \frac{C_n}{4}\sqrt v
    \,\middle|\,
    \D_n
    \right)
    \le
    \Pi\!\left(P:
    T_{1n}>\frac{C_n}{8}\sqrt v
    \,\middle|\,
    \D_n
    \right)
    +
    \Pi\!\left(P:
    T_{2n}>\frac{C_n}{8}\sqrt v
    \,\middle|\,
    \D_n
    \right).
\]

\paragraph{Control of $T_{1n}$.}
By Lemma~\ref{Lemma4}, for $P_0^\infty$-almost every data realization, $T_{1n}=o_{\mathbb Q_n}(1)$.
Since $C_n\sqrt v/8$ is bounded away from zero
for all sufficiently large $n$,
\[
    \Pi\!\left(P:
    T_{1n}>\frac{C_n}{8}\sqrt v
    \,\middle|\,\D_n
    \right)
    =
    \mathbb Q_n\!\left(
    T_{1n}>\frac{C_n}{8}\sqrt v
    \right)
    =
    o(1), \qquad P_0^{\infty}\text{-almost surely}.
\]

\paragraph{Control of $T_{2n}$.}
Recall that $T_{2n} = \bar{W}_n^{-1} \|\nu_n\|_{\mathcal{F}}$, where $\nu_n=n^{-1/2} \sum_{i=1}^n \xi_i Z_i$.
Conditional on $\D_n$, the signed measures $Z_1,\ldots,Z_n$ are fixed, so the randomness in $T_{2n}$ arises solely from the multipliers $(\xi_1,\ldots,\xi_n)$. 
Let $\Prob_\xi$ denote the common law of the $\xi_i$. I write $\Prob_\xi^{\otimes n}$ for the joint law of $(\xi_1,\ldots,\xi_n)$ and $\Exp_\xi$ for expectation under this joint law.

Since $\bar W_n=n^{-1}\sum_{i=1}^n(\xi_i+1)\to 1$ in
$\Prob_\xi^{\infty}$-probability,
$\Prob_\xi^{\otimes n}(\bar W_n^{-1}>2)=o(1)$. Moreover,
\[
    \left\{ T_{2n} > \frac{C_n}{8}\sqrt v \right\}
    \subseteq
    \{\bar W_n^{-1}>2\}
    \cup
    \left\{
    \|\nu_n\|_{\mathcal F} > \frac{C_n}{16}\sqrt v
    \right\}.
\]
Therefore, for every realization of $\{D_i\}_{i \geq 1}$,
\tightalign{
    \Pi\!\left(P:
    T_{2n}>\frac{C_n}{8}\sqrt v
    \,\middle|\,\D_n
    \right)
    =
    \Prob_\xi^{\otimes n}\!\left(
    T_{2n}>\frac{C_n}{8}\sqrt v
    \right)
    \leq
    o(1)
    +
    \Prob_\xi^{\otimes n}\!\left(
    \|\nu_n\|_{\mathcal F}>\frac{C_n}{16}\sqrt v
    \right).
}
By Markov's inequality,
\tightdisplay{
    \Prob_\xi^{\otimes n}\!\left(
    \|\nu_n\|_{\mathcal F}>\frac{C_n}{16}\sqrt v
    \right)
    \leq
    \frac{16}{C_n\sqrt v}\,
    \Exp_\xi\|\nu_n\|_{\mathcal F}.
}
Combining this bound with the control of $T_{1n}$ gives, for
$P_0^\infty$-almost every data realization,
\tightdisplay{
    \Pi\!\left(P:
    \rho_{\G}(\Prob_n,P)>\frac{C_n}{4}\sqrt{\frac v n}
    \,\middle|\,\D_n
    \right)
    \leq
    o(1)
    +
    \frac{16}{C_n\sqrt v}\,
    \Exp_\xi\|\nu_n\|_{\mathcal F}.
}
Thus, it remains to show that there exists a constant $K<\infty$ such that
\tightdisplay{
    \limsup_{n\to\infty}
    \Exp_\xi\|\nu_n\|_{\mathcal F}
    \leq K\sqrt v,
    \qquad
    P_0^\infty\text{-a.s.}
    \tag{B.10}\label{appendix:eq10}
}

\medskip \noindent
\textbf{Step 4: Uniform bound for the multiplier empirical process.}

I now establish~\eqref{appendix:eq10}. All expectations below are conditional on $\mathcal{D}_n$. Applying Lemma~\ref{Lemma5} with $Z_i=\delta_{D_i}-\Prob_n$ gives
\tightdisplay{
    \Exp_\xi
    \left\|
    \frac{1}{\sqrt n}\sum_{i=1}^n \xi_i Z_i
    \right\|_{\mathcal F}
    \leq
    2\|\xi_1\|_{2,1}
    \max_{1\leq k\leq n}
    \Exp_R
    \left\|
    \frac{1}{\sqrt k}\sum_{i=1}^k Z_{R_i}
    \right\|_{\mathcal F},
}
where $R=(R_1,\ldots,R_n)$ is uniformly distributed over permutations of $\{1,\ldots,n\}$.
By Lemma~\ref{Lemma6}, for every $1\leq k\leq n$,
\tightdisplay{
    \Exp_R
    \left\|
    \frac{1}{\sqrt k}\sum_{i=1}^k Z_{R_i}
    \right\|_{\mathcal F}
    \leq
    \Exp_{\widehat D}
    \left\|
    \widehat{\mathbb G}_{n,k}
    \right\|_{\mathcal F},
    \qquad
    \widehat{\mathbb G}_{n,k}
    \coloneqq
    \frac{1}{\sqrt k}\sum_{i=1}^k(\delta_{\widehat D_i}-\Prob_n).
}
Define $U_j \coloneqq \Exp_\varepsilon \| j^{-1/2} \sum_{i=1}^j\varepsilon_i\delta_{D_i} \|_{\mathcal F}$. Taking the maximum over $1 \leq k \leq n$ and applying Lemma~\ref{Lemma7} gives
\tightdisplay{
    \limsup_{n\to\infty}
    \Exp_\xi
    \left\|
    \frac{1}{\sqrt n}\sum_{i=1}^n \xi_i Z_i
    \right\|_{\mathcal F}
    \leq
    32\sqrt 2\,\|\xi_1\|_{2,1}
    \Exp\!\left(\sup_{j\geq1} U_j\right),
    \qquad
    P_0^\infty\text{-a.s.}
}

Finally, Lemma~\ref{Lemma8} implies that
\tightdisplay{
    \Exp\left(\sup_{j\geq1}U_j\right) \leq C\sqrt v,
}
for some $C<\infty$. Since $\|\xi_1\|_{2,1}<\infty$, there exists a constant $K<\infty$ such that
\tightdisplay{
    \limsup_{n\to\infty}
    \Exp_\xi\|\nu_n\|_{\mathcal F}
    \leq
    K\sqrt v,
    \qquad
    P_0^\infty\text{-a.s.}
}
This proves~\eqref{appendix:eq10}.
Combining Steps~1--4 completes the proof.
\end{proof}


\subsection{Proof of Theorem~\ref{Theorem2}}
\label{Appendix:proof Theorem2}

\begin{proof}[Proof of Theorem~\ref{Theorem2}]

For $i= 1,2$, write
\tightdisplay{
    W_{\G_i}^{\star}(P)\coloneqq \sup_{G\in\G_i} W(P;G),
    \qquad
    \rho_{\G_i}(P_0,P)
    \coloneqq
    \sup_{G\in\G_i}|W(P_0;G)-W(P;G)|.
}
Let $\Delta_0 \coloneqq W_{\G_1}^{\star}(P_0)-W_{\G_2}^{\star}(P_0)$ and  $\Delta(P) \coloneqq W_{\G_1}^{\star}(P)-W_{\G_2}^{\star}(P)$.

For each $i=1,2$,
\tightdisplay{
    \bigl|W_{\G_i}^{\star}(P_0)-W_{\G_i}^{\star}(P)\bigr|
    \leq
    \sup_{G\in\G_i}|W(P_0;G)-W(P;G)|
    =
    \rho_{\G_i}(P_0,P).
}
Therefore,
\[
    |\Delta(P)-\Delta_0|
    \leq
    \rho_{\G_1}(P_0,P)+\rho_{\G_2}(P_0,P).
    \tag{B.11}\label{appendix:eq11}
\]
By the posterior concentration bound established in the proof of
Theorem~\ref{Theorem1}, for each $i=1,2$ and every fixed $\eta>0$,
\[
    \Pi(P: \rho_{\G_i}(P_0,P)>\eta \mid \mathcal{D}_n )
    \stackrel{P_0}{\to} 0.
    \tag{B.12}\label{appendix:eq12}
\]

\paragraph{Part (i).}
First suppose $\Delta_0\neq 0$. By~\eqref{appendix:eq11},
\tightdisplay{
    \{\sign(\Delta(P))\neq \sign(\Delta_0)\}
    \subseteq
    \left\{
    \rho_{\G_1}(P_0,P)+\rho_{\G_2}(P_0,P)\geq |\Delta_0|
    \right\}
    \subseteq
    \bigcup_{i=1}^2
    \left\{
    \rho_{\G_i}(P_0,P)\geq \frac{|\Delta_0|}{2}
    \right\}.
}
Thus, the union bound and~\eqref{appendix:eq12} imply $\Pi( \sign(\Delta(P))\neq \sign(\Delta_0) \mid \D_n) \stackrel{P_0}{\to} 0$.

\paragraph{Part (ii).}
Now suppose $\Delta_0=0$. Then~\eqref{appendix:eq11} gives
$|\Delta(P)|\leq \rho_{\G_1}(P_0,P)+\rho_{\G_2}(P_0,P)$. Hence, for any $\varepsilon>0$,
\tightdisplay{
    \{P: |\Delta(P)|\geq \varepsilon\}
    \subseteq
    \bigcup_{i=1}^2
    \left\{P:
    \rho_{\G_i}(P_0,P)\geq \frac{\varepsilon}{2}
    \right\}.
}
Again, the union bound and~\eqref{appendix:eq12} imply $\Pi (P: |\Delta(P)|\geq \varepsilon \mid \D_n) \stackrel{P_0}{\to} 0$,
or equivalently,
\tightdisplay{
    \Pi\!\left(P:
    \bigl|W_{\G_1}^{\star}(P)-W_{\G_2}^{\star}(P)\bigr|<\varepsilon
    \,\middle|\,\D_n
    \right)
    \stackrel{P_0}{\to}1.
}
This proves part~(ii).
\end{proof}


\subsection{Proof of Theorem~\ref{Theorem3}}
\label{Appendix:proof Theorem3}

\begin{proof}[Proof of Theorem~\ref{Theorem3}]

When $|\alpha|=0$, PEWM reduces to EWM, and the result follows directly from \citet[Theorem~2.1]{kitagawa2018should}. Hence, suppose $|\alpha|>0$.

Write $\widehat G(\alpha)\coloneqq \widehat G_{\mathrm{PEWM}}(\alpha)$ and define
\tightdisplay{
    W_{n,\alpha}(G)
    \coloneqq
    \frac{|\alpha|}{|\alpha|+n}W(\widetilde\alpha;G)
    +
    \frac{n}{|\alpha|+n}W(\Prob_n;G).
}
By Proposition~\ref{Proposition1}, $\widehat G(\alpha)$ maximizes
$W_{n,\alpha}$ over $\G$. Therefore,
\tightdisplay{
    W_{\G}^{\star}(P_0)-W(P_0;\widehat G(\alpha))
    \leq
    2\sup_{G\in\G}|W_{n,\alpha}(G)-W(P_0;G)|.
}
Using the definition of $W_{n,\alpha}$ gives
\tightdisplay{
    W_{\G}^{\star}(P_0)-W(P_0;\widehat G(\alpha))
    \leq
    \frac{2|\alpha|}{|\alpha|+n}
    \sup_{G\in\G}|W(\widetilde\alpha;G)-W(P_0;G)|
    +
    \frac{2n}{|\alpha|+n}
    \sup_{G\in\G}|W(\Prob_n;G)-W(P_0;G)|.
}
Taking expectations under $P_0^{\otimes n}$ yields
\tightdisplay{
    \mathcal R_n(\alpha)
    \leq
    \frac{2|\alpha|}{|\alpha|+n}
    \sup_{G\in\G}|W(\widetilde\alpha;G)-W(P_0;G)|
    +
    \frac{2n}{|\alpha|+n}
    \Exp_{P_0}
    \left[
    \sup_{G\in\G}|W(\Prob_n;G)-W(P_0;G)|
    \right].
}

By \citet[Theorem~2.1]{kitagawa2018should}, there exists a constant $C<\infty$
such that
\[
    \sup_{\mathbf P_0^\star\in\mathcal P(M)}
    \Exp_{P_0}
    \left[
    \sup_{G\in\G}|W(\Prob_n;G)-W(P_0;G)|
    \right]
    \leq
    C\sqrt{\frac v n}.
\]
To control the prior term, observe that $|W(Q;G)|\leq \Exp_Q(|Z_1|+|Z_0|)$ for every probability measure $Q$ and every $G \in \mathcal{G}$. 
Assumption~\ref{Assumption:base measure} and the definition of $\mathcal P(M)$ then imply that, for some constant $C_\alpha<\infty$,
\tightdisplay{
    \sup_{\mathbf P_0^\star\in\mathcal P(M)}
    \sup_{G\in\G}|W(\widetilde\alpha;G)-W(P_0;G)|
    \leq
    C_\alpha .
}
Therefore,
\[
    \sup_{\mathbf P_0^\star\in\mathcal P(M)}
    \mathcal R_n(\alpha)
    \leq
    \frac{2|\alpha|}{|\alpha|+n}C_\alpha
    +
    \frac{2n}{|\alpha|+n}C\sqrt{\frac v n}.
\]
Since $|\alpha|<\infty$, the first term is $O(n^{-1})$ and the second is $O(n^{-1/2})$. Thus, there exists a constant $C_1<\infty$ such that
\tightdisplay{
    \sup_{\mathbf P_0^\star\in\mathcal P(M)}
    \mathcal R_n(\alpha)
    \leq
    C_1\sqrt{\frac v n}.
}
\end{proof}


\subsection{Proof of Proposition~\ref{Proposition1}}
\label{Appendix:proof Proposition1}

\begin{proof}[Proof of Proposition~\ref{Proposition1}]

Recall that $W(P;G)=\int f(\cdot; G) \,dP$. 
By \citet[Theorem~3]{ferguson1973bayesian}, if $P\sim \mathrm{DP}(\nu)$ and $f$ is integrable with respect to $\nu$, then $\Exp\!\left[\int f\,dP\right] = \int f\,d\bar\nu$, where $\bar\nu \coloneqq \nu / \nu(\mathcal D)$.

Suppose first that $|\alpha|>0$, and let
$\widetilde\alpha\coloneqq\alpha/|\alpha|$. Applying the preceding property with $\nu=\alpha+n\Prob_n$ gives
\tightdisplay{
    \int W(P;G)\,d\Pi(P\mid\mathcal D_n)
    =
    \int f(\cdot; G) \,d\left(\frac{\alpha+n\mathbb P_n}{|\alpha|+n}\right)
    =
    \frac{|\alpha|}{|\alpha|+n}W(\widetilde\alpha;G)
    +
    \frac{n}{|\alpha|+n}W(\mathbb P_n;G).
}
It follows that
\tightdisplay{
    \widehat G_{\mathrm{PEWM}}(\alpha)
    \in
    \arg\max_{G\in\mathcal G}
    \left\{
    \frac{|\alpha|}{|\alpha|+n} W(\widetilde\alpha;G)
    +
    \frac{n}{|\alpha|+n} W(\mathbb P_n;G)
    \right\}.
}

When $|\alpha|=0$, the posterior is $\mathrm{DP}(n\Prob_n)$. Applying the same property with $\nu=n\Prob_n$ yields
\tightdisplay{
    \int W(P;G)\,d\Pi(P\mid\mathcal D_n)
    =
    \int f(\cdot; G) \,d\mathbb P_n
    =
    W(\mathbb P_n;G).
}
Thus, posterior expected welfare coincides with empirical welfare for every $G\in\mathcal G$, and $\widehat G_{\mathrm{PEWM}}(0) = \widehat G_{\mathrm{EWM}}$.
\end{proof}


\subsection{Proof of Proposition~\ref{Proposition2}}
\label{Appendix:proof Proposition2}

\begin{proof}[Proof of Proposition~\ref{Proposition2}]

Fix $\tau\in\mathcal X$. Write $\mu_t(\tau)\coloneqq \Exp_{\mathbf P_0^\star}[Y(t)\mid X=\tau]$ ($t=0,1$) and $p_\tau\coloneqq \mathbf P_0^\star(X=\tau)$.
By Assumption~\ref{Assumption:DGP}, $p_\tau>0$ and
$e(\tau)\in[\kappa,1-\kappa]$. Moreover,
\tightdisplay{
    \Exp_{P_0}\!\left[(Z_1-Z_0)\mathbf 1\{X=\tau\}\right]
    =
    p_\tau\{\mu_1(\tau)-\mu_0(\tau)\}
    =
    p_\tau\mathrm{CATE}(\tau).
}

Because $\mathcal G=2^{\mathcal X}$, the PEWM rule treats cell $\tau$ if and only if the following cell-specific welfare contrast is positive:
\tightdisplay{
    a_{1n}(\tau)
    \coloneqq
    \frac{|\alpha|}{|\alpha|+n}
    \Exp_{\widetilde\alpha}\!\left[(Z_1-Z_0)\mathbf 1\{X=\tau\}\right]
    +
    \frac{1}{|\alpha|+n}
    \sum_{i=1}^n
    \left\{
    \frac{Y_iT_i}{e(X_i)}
    -
    \frac{Y_i(1-T_i)}{1-e(X_i)}
    \right\}
    \mathbf 1\{X_i=\tau\}.
}
The first term is $o(1)$ by Assumption~\ref{Assumption:base measure}, and the second term converges almost surely to $p_\tau\mathrm{CATE}(\tau)$ by the strong law of large numbers. Hence, $a_{1n}(\tau) \stackrel{\mathrm{a.s.}}{\longrightarrow} p_\tau\mathrm{CATE}(\tau)$.

The EWM rule is obtained analogously, with cell-specific welfare contrast
\tightdisplay{
    a_{2n}(\tau)
    \coloneqq
    \frac1n
    \sum_{i=1}^n
    \left\{
    \frac{Y_iT_i}{e(X_i)}
    -
    \frac{Y_i(1-T_i)}{1-e(X_i)}
    \right\}
    \mathbf 1\{X_i=\tau\},
}
and the same argument gives $a_{2n}(\tau) \stackrel{\mathrm{a.s.}}{\longrightarrow} 
p_\tau\mathrm{CATE}(\tau)$.

For the LIB rule, the cell-specific welfare contrast is
\tightdisplay{
    a_{3n}(\tau)
    \coloneqq
    \frac{\int y\,d\beta_{1,\tau}(y)
    +\sum_{i=1}^nY_iT_i\mathbf 1\{X_i=\tau\}}
    {|\beta_{1,\tau}|+n(1,\tau)}
    -
    \frac{\int y\,d\beta_{0,\tau}(y)
    +\sum_{i=1}^nY_i(1-T_i)\mathbf 1\{X_i=\tau\}}
    {|\beta_{0,\tau}|+n(0,\tau)}.
}
For the first term of $a_{3n}(\tau)$, the strong law of large numbers gives $n^{-1} \sum_{i=1}^nY_i T_i\mathbf 1\{X_i=\tau\} \stackrel{\mathrm{a.s.}}{\longrightarrow} p_\tau e(\tau)\mu_1(\tau)$ and $n(1,\tau)/n \stackrel{\mathrm{a.s.}}{\longrightarrow} p_\tau e(\tau)$.
Since $p_\tau e(\tau)>0$ and the prior mass and first moment of
$\beta_{1,\tau}$ are finite, this term converges almost surely to
$\mu_1(\tau)$. Analogously, the second term of $a_{3n}(\tau)$ converges
almost surely to $\mu_0(\tau)$. Therefore, $a_{3n}(\tau) \stackrel{\mathrm{a.s.}}{\longrightarrow} \mu_1(\tau)-\mu_0(\tau) \eqqcolon \mathrm{CATE}(\tau)$.

For the CES rule, the cell-specific welfare contrast is
\tightdisplay{
    a_{4n}(\tau)
    \coloneqq
    \frac{\sum_{i=1}^n Y_iT_i\mathbf 1\{X_i=\tau\}}{n(1,\tau)}
    -
    \frac{\sum_{i=1}^nY_i(1-T_i)\mathbf 1\{X_i=\tau\}}{n(0,\tau)}.
}
The same argument gives $a_{4n}(\tau) \stackrel{\mathrm{a.s.}}{\longrightarrow} \mathrm{CATE}(\tau)$.

Let $\widehat G_{j,n}$ denote the rule selected by procedure $j$, where $j=1,2,3,4$ correspond to PEWM, EWM, LIB, and CES, respectively. Each procedure $j$
treats cell $\tau$ when $a_{jn}(\tau)$ is positive. Since $p_\tau>0$ and $\mathrm{CATE}(\tau)\neq 0$, the continuous mapping theorem implies $\mathbf 1\{\tau\in\widehat G_{j,n}\} \stackrel{\mathrm{a.s.}}{\longrightarrow} \mathbf 1\{\mathrm{CATE}(\tau)>0\}$, $j=1,\ldots,4$.
\end{proof}

\appendix
\setcounter{section}{2} 

\numberwithin{figure}{section}
\numberwithin{table}{section}

\renewcommand{\thefigure}{\thesection.\arabic{figure}}
\renewcommand{\thetable}{\thesection.\arabic{table}}

\numberwithin{assumption}{section}
\numberwithin{theorem}{section}
\numberwithin{lemma}{section}
\numberwithin{definition}{section}
\numberwithin{proposition}{section}
\numberwithin{remark}{section}
\numberwithin{corollary}{section}
\numberwithin{example}{section}

\renewcommand{\theassumption}{\thesection.\arabic{assumption}}
\renewcommand{\thetheorem}{\thesection.\arabic{theorem}}
\renewcommand{\thelemma}{\thesection.\arabic{lemma}}
\renewcommand{\thedefinition}{\thesection.\arabic{definition}}
\renewcommand{\theproposition}{\thesection.\arabic{proposition}}
\renewcommand{\theremark}{\thesection.\arabic{remark}}
\renewcommand{\thecorollary}{\thesection.\arabic{corollary}}
\renewcommand{\theexample}{\thesection.\arabic{example}}

\section{Extensions}
\label{supplement:extensions}

This section provides four extensions of \textit{Nonparametric Bayesian Policy Learning} (NBPL): alternative welfare criteria, multi-valued discrete treatments, demeaned outcomes, and capacity constraints. The first three extensions are accommodated by augmenting the reduced-form vector used in the main text. 
I derive posterior regret contraction rates for these extensions except for alternative welfare criteria, which involve criterion-specific empirical processes and possibly require different proof techniques.

\subsection{Alternative Welfare Criteria}
\label{supplement:alternative criteria}

\paragraph{Equality-minded welfare.}

For a treatment rule $G$, let
\tightdisplay{
    F_{P_0^{\star},G}(y)
    \coloneqq
    \Exp_{P_0^{\star}}\!\left[
    \mathbf{1}\{Y(1)\le y,\, X\in G\}
    +
    \mathbf{1}\{Y(0)\le y,\, X\notin G\}
    \right]
}
denote the realized-outcome distribution induced by assigning treatment according to $G$.

\citet{kitagawa2021equality} propose a rank-dependent welfare criterion,
\tightdisplay{
    W_{\Lambda}(P_0^{\star};G)
    \coloneqq 
    \int_{0}^{\infty} \Lambda\!\left(F_{P_0^{\star},G}(y)\right)\, d y,
}
where $\Lambda: [0,1] \to [0,1]$ is a nonincreasing, convex function with $\Lambda(0)=1$, $\Lambda(1)=0$. 
Relative to utilitarian welfare, this criterion places greater weight on lower-ranked outcomes.

\paragraph{Quantile-optimal welfare.}

\citet{wang2018quantile} consider maximizing a specified quantile of the realized outcome distribution under rule $G$:
\tightdisplay{
    \max_{G\in\mathcal G}\; \mathcal Q_{P_0^\star,G}(\tau),
}
where $\mathcal Q_{P_0^\star,G}(\tau) \coloneqq \inf \{y:F_{P_0^\star,G}(y)\ge \tau\}$, $\tau \in (0,1)$,
denotes the $\tau$-quantile of the realized-outcome distribution under $G$. This criterion targets the lower tail of the realized-outcome distribution directly, rather than aggregating welfare over the full distribution.

\paragraph{Fairness-oriented welfare.}

\citet{fang2023fairness} study welfare maximization subject to a lower bound on a specified quantile of the realized-outcome distribution:
\tightdisplay{
    \max_{G\in\mathcal G}\, W(P_0^\star;G)
    \quad \text{subject to} \quad
    \mathcal Q_{P_0^\star,G}(\tau)\ge q,
}
where $W(P_0^\star;G)=\int y\,dF_{P_0^\star,G}(y)$, $\tau\in(0,1)$, and $q\in\mathbb R$.
The constraint requires the $\tau$-quantile of the outcome distribution induced by the selected rule to exceed the prescribed threshold $q$; for small $\tau$, this restricts the lower tail of the induced distribution.
Since both $W(P_0^\star;G)$ and $\mathcal Q_{P_0^\star,G}(\tau)$ are functions of $F_{P_0^\star,G}$ alone, this criterion depends on $P_0^\star$ only through the realized-outcome distribution $F_{P_0^\star,G}$.

\paragraph{Incorporating these criteria into NBPL.}

The preceding three welfare criteria are accommodated by NBPL because they depend on $P_0^\star$ only through $F_{P_0^\star,G}$, the realized-outcome distribution induced by assigning treatment according to $G$. Under Assumption~\ref{Assumption:DGP}, this distribution is point identified. 
In particular,
\tightdisplay{
    F_{P_0^{\star},G}(y)
    =
    \Exp_{P_0^{\star}} \left[
    \left(
    \frac{T}{e(X)} \mathbf{1}\{X \in G\}
    +
    \frac{1-T}{1-e(X)} \mathbf{1}\{X \notin G\}
    \right)
    \mathbf{1}\{Y \le y\}
    \right].
}
Thus, $F_{P_0^\star,G}$ is a deterministic function of the distribution of the \emph{augmented reduced-form vector}
\[
    \check D \coloneqq \left(Y,\frac{T}{e(X)},\frac{1-T}{1-e(X)},X\right).
\]
Accordingly, the NBPL framework and Algorithm~\ref{Algorithm:NBPL} extend directly to these criteria by replacing the reduced-form vector in the main text with $\check D$. Regret convergence rates for these alternative criteria require case-specific analysis, which I leave for future work.


\subsection{Multiple Treatments}
\label{supplement:multivalued treatment}

The binary-treatment framework extends naturally to multi-valued discrete treatments. Suppose there are $K>2$ treatments, denoted by $T \in\{1,\ldots,K\} $. Let $Y(1), \ldots, Y(K)$ denote the potential outcomes and let $e_k(x) \coloneqq \mathbf P_0^\star(T=k\mid X=x)$ ($k=1,\ldots,K$)
denote the propensity scores, where $\mathbf{P}_0^{\star}$ denotes the true distribution of $(Y(1),\ldots,Y(K),T,X)$ in the experimental population.
A treatment assignment rule is a $K$-partition of the covariate space, $G=(G_1,\ldots,G_K)$, where $G_1,\ldots,G_K\subseteq\mathcal X$ are disjoint and
$\bigcup_{k=1}^K G_k=\mathcal X$. Treatment $k$ is assigned to units with $X\in G_k$.

Let $\mathcal G$ be a VC class of subsets of $\mathcal X$ with VC-dimension $v$, that is, Assumption~\ref{Assumption:VC} holds.
Define the class of multiple-treatment assignment rules\footnote{For example, when $\mathcal X=\mathbb R$, the class of intervals $\{(x,x']:-\infty\leq x < x'\leq\infty\}\cup\{\emptyset\}$ is a VC class and can be used to form $K$-partitions of $\mathbb R$.}
\tightdisplay{
    \mathscr{G}
    \coloneqq
    \left\{
    G=(G_1,\ldots,G_K):
    G_k\in\mathcal G,\ k=1,\ldots,K,
    \text{ and } (G_1,\ldots,G_K) \text{ partitions } \mathcal X
    \right\}.
    \tag{C.1} \label{supplement:eq1}
}
The following Assumption~\ref{Assumption:multivalued treatment} replaces Assumption~\ref{Assumption:DGP}.

\begin{assumption}
\label{Assumption:multivalued treatment}
\begin{enumerate}[label=(\alph*)]
    \item \textup{(External Validity)}
    The target and experimental populations share the same joint distribution
    of potential outcomes and covariates:
    $Q_0^\star = P_0^\star$.

    \item \textup{(Unconfoundedness)}
    $(Y(1),\ldots,Y(K))\indep T \mid X$.

    \item \textup{(Strict Overlap)}
    There exists $\kappa\in(0,1/K)$ such that $e_k(x)\geq \kappa$ $(k=1,\ldots,K)$ for $\mathbf P_0^\star$-almost every $x\in\mathcal X$. The propensity scores are known.

    \item \textup{(Outcome Moments)}
    $\Exp_{\mathbf P_0^\star}|Y|^{2+\delta}<\infty$ for some $\delta>0$.
\end{enumerate}
\end{assumption}

Define the \textit{reduced-form data} $\breve D \coloneqq (Z_1,\ldots,Z_K,X)$, where
\tightdisplay{
    Z_k
    \coloneqq
    \frac{Y \cdot \mathbf{1}\{T = k\}}{e_k(X)}, \qquad k=1,\ldots,K,
}
and $\breve{D} \in \breve{\mathcal{D}} \coloneqq \mathcal{Z}_1 \times \cdots \times \mathcal{Z}_K \times \mathcal{X}$.
Let $P_0$ denote the true reduced-form distribution of $\breve D$. Under Assumption~\ref{Assumption:multivalued treatment}, welfare is point-identified under any treatment rule $G=(G_1,\ldots,G_K) \in \mathscr G$,
\tightdisplay{
    W(P_0;G)
    \coloneqq
    \Exp_{P_0^\star}\!\left[
    \sum_{k=1}^K Y(k) \mathbf 1\{X\in G_k\}
    \right]
    =
    \sum_{k=1}^K
    \Exp_{P_0}\!\left[
    Z_k\mathbf 1\{X\in G_k\}
    \right].
}
NBPL places a Dirichlet process prior on the reduced-form distribution of $\breve D$: $P \sim \mathrm{DP}(\alpha)$.
For a generic reduced-form distribution $P$ and any $G=(G_1,\ldots,G_K)\in\mathscr G$, define
\tightdisplay{
    W(P;G)
    \coloneqq
    \sum_{k=1}^K
    W^k(P;G_k),
    \qquad
    W^k(P;G_k) \coloneqq 
    \Exp_P\!\left[
    Z_k\mathbf 1\{X\in G_k\}
    \right].
    \tag{C.2} \label{supplement:eq2}
}
Let
\tightdisplay{
    G^\star(P) \in \arg\max_{G\in\mathscr G}W(P;G), 
    \qquad
    W_{\mathscr G}^\star(P) \coloneqq \sup_{G\in\mathscr G}W(P;G).
}
The welfare regret of the optimal treatment rule selected under $P$ is
\tightdisplay{
    R_{\mathscr G}(P_0;P) \coloneqq W_{\mathscr G}^\star(P_0)-W(P_0;G^\star(P)).
}

The following Assumption~\ref{Assumption:multivalued base measure} is the multiple-treatment analogue of
Assumption~\ref{Assumption:base measure}.

\begin{assumption}
\label{Assumption:multivalued base measure}

The base measure $\alpha$ on $\breve{\mathcal D} \coloneqq \mathcal Z_1 \times \cdots \times\mathcal Z_K\times\mathcal X$
has finite total mass $|\alpha|\coloneqq \alpha(\breve{\mathcal D})<\infty$ and satisfies $\int_{\breve{\mathcal D}} |z_k|\, d\alpha(z_1,\ldots,z_K,x)<\infty$ for $k=1,\ldots,K$.

\end{assumption}

Theorem~\ref{Theorem:C1} shows that multiple treatments  inflate the posterior contraction rate in Theorem~\ref{Theorem1} by the number of treatment values $K$. Hence, the rate remains unchanged for fixed $K$.

\begin{theorem}
\label{Theorem:C1}
Suppose Assumptions \ref{Assumption:multivalued treatment}, \ref{Assumption:multivalued base measure}, and 3 hold. Let $\mathscr G$ be the policy class of $K$-partitions defined in~\eqref{supplement:eq1}. 
Let $\mathcal{\breve D}_n \coloneqq \{\breve D_i\}_{i=1}^n \stackrel{\mathrm{i.i.d.}}{\sim} P_0$ denote the observed sample, and let $\Pi(P \in \cdot\mid \mathcal{\breve D}_n)$ denote the Dirichlet process posterior $P \mid \mathcal{\breve D}_n \sim \mathrm{DP}(\alpha + \sum_{i=1}^n \delta_{\breve D_i})$. Then, for every sequence $C_n\to\infty$ (arbitrarily slowly), as $n \to \infty$,
\tightdisplay{
    \Pi\!\left(P:
    R_{\mathscr G}(P_0;P)
    \leq 
    C_n K\sqrt{\frac vn}
    \,\middle|\, \mathcal{\breve D}_n
    \right)
    \stackrel{P_0}{\to} 1.
    \tag{C.3} \label{supplement:eq3}
}
\end{theorem}

\begin{proof}
For any class of measurable functions $\mathcal F$, let $Pf \coloneqq \Exp_{P}[f]$, $\|P\|_{\mathcal{F}} \coloneqq \sup_{f \in \mathcal{F}}|Pf|$, and $\|f\|_{P,r} \coloneqq (\Exp_{P}|f|^r)^{1/r}$ for $r>0$.

\textbf{Step 1.}
For $k=1,\ldots,K$, define the treatment-specific function class
\tightdisplay{
    \mathcal F_k
    \coloneqq
    \left\{
    f_{k,G}(\breve D)
    =
    Z_k\mathbf 1\{X\in G\}:
    G\in\mathcal G
    \right\}.
}
By Assumption~\ref{Assumption:VC} and \citet[Lemma A.1]{kitagawa2018supp}, each $\mathcal F_k$ is a VC-subgraph class with VC-dimension at most $v$. 
Moreover, $F_k(\breve D) = |Z_k|$ is an envelope for $\F_k$, and Assumptions~\ref{Assumption:multivalued treatment} and \ref{Assumption:multivalued base measure} ensure that $\| F_k \|_{P_0,2} < \infty$ and $\int_{\breve{\mathcal{D}}} F_k \, d\alpha < \infty$.

Let $P$ and $Q$ be any two reduced-form distributions of $\breve D$. Define
\tightdisplay{
    \rho_{\mathscr G}(P,Q) \coloneqq \sup_{G\in\mathscr G} |W(P;G)-W(Q;G)|.
}
By the triangle inequality,
\tightdisplay{
    \rho_{\mathscr G}(P,Q)
    =
    \sup_{G\in\mathscr G}
    \left|
    \sum_{k=1}^K
    (P-Q)\bigl[
    Z_k\mathbf 1\{X\in G_k\}
    \bigr]
    \right| 
    \leq
    \sum_{k=1}^K
    \sup_{G_k\in\mathcal G}
    \left|
    (P-Q)\bigl[
    Z_k \mathbf 1\{X\in G_k\}
    \bigr]
    \right|.
}
For each $k=1,\ldots,K$, let $\rho_{\mathcal{G}}^k(P,Q) \coloneqq \sup_{G_k \in \mathcal{G}} |W^k(P;G_k)-W^k(Q;G_k)|$.
Then,
\tightdisplay{
    \rho_{\mathscr G}(P,Q)
    \leq
    \sum_{k=1}^K
    \rho_{\mathcal{G}}^k(P,Q).
    \tag{C.4}\label{supplement:eq4}
}
Inequality~\eqref{supplement:eq4} produces $K$ treatment-specific terms, each indexed by the original class $\mathcal G$. In this sense, the multiple-treatment problem reduces to $K$ binary-treatment problems. Each term can therefore be controlled by the same argument used in the proof of Theorem~\ref{Theorem1}.

As in the proof of Lemma~\ref{Lemma1},
\tightdisplay{
    R_{\mathscr G}(P_0;P)
    \leq
    2\rho_{\mathscr G}(P_0,P)
    \leq
    2 \{ \rho_{\mathscr G}(P_0,\Prob_n)
    +
    \rho_{\mathscr G}(\Prob_n,P) \}.
    \tag{C.5} \label{supplement:eq5}
}
\eqref{supplement:eq4} implies $\rho_{\mathscr G}(P_0,\Prob_n) \leq \sum_{k=1}^K \rho_{\mathcal{G}}^k(P_0,\Prob_n)$ and $\rho_{\mathscr G}(\Prob_n, P) \leq \sum_{k=1}^K \rho_{\mathcal{G}}^k(\Prob_n, P)$. Thus,
\tightdisplay{
    R_{\mathscr G}(P_0;P) \leq 2 \sum_{k=1}^K \rho_{\mathcal{G}}^k(P_0,\Prob_n) + 
    2 \sum_{k=1}^K \rho_{\mathcal{G}}^k(\Prob_n, P)
    \tag{C.6} \label{supplement:eq6}.
}

\textbf{Step 2.}
For each $k$, $\rho_{\mathcal G}^k(P_0,\Prob_n)$ is controlled by the same VC maximal inequality used in the proof of Lemma~\ref{Lemma2}, with $\mathcal F$ replaced by $\mathcal F_k$:
\tightdisplay{
    \rho_{\mathcal{G}}^k(P_0,\Prob_n) = \|P_0 - \Prob_n\|_{\mathcal F_k} = O_{P_0} 
    \left(\sqrt{\frac{v}{n}} \right).
}

\textbf{Step 3.}
Define
\tightdisplay{
    \widetilde{\mathcal{E}}_n \coloneqq 
    \left\{ P: R_{\mathscr G}(P_0;P) > C_n K \sqrt{\frac{v}{n}} \right\}, 
    \qquad
    \widetilde{\mathcal A}_n
    \coloneqq
    \left\{ \{\breve D_i\}_{i=1}^n:
    \rho_{\mathscr G}(P_0,\Prob_n)
    \leq
    \frac{C_nK}{4}\sqrt{\frac vn}
    \right\}.
}
Also define, for each $k = 1, \ldots, K$,
\tightdisplay{
    \widetilde{\mathcal{A}}_{n,k} \coloneqq 
    \left\{ \{\breve D_i\}_{i=1}^n: \rho_{\mathcal{G}}^k(P_0,\Prob_n) \leq \frac{C_n}{4} \sqrt{\frac{v}{n}} \right\}.
}
Since $\bigcap_{k=1}^K \widetilde{\mathcal{A}}_{n,k} \subseteq \widetilde{\mathcal{A}}_n$ and $P_0^\infty(\widetilde{\mathcal A}_{n,k})\to1$ for each fixed $k$, we have $P_0^\infty(\widetilde{\mathcal A}_n)\to 1$ as $n \to \infty$. 
On $\widetilde{\mathcal A}_n$, the preceding regret bound~\eqref{supplement:eq5} implies that it suffices to show
\tightdisplay{
    \Pi
    \left(P:
    \rho_{\mathscr G}(\Prob_n,P)
    >
    \frac{C_nK}{4}\sqrt{\frac vn}
    \,\middle|\,
    \breve{\mathcal D}_n
    \right)
    =
    o_{P_0}(1).
}

\textbf{Step 4.}
By~\eqref{supplement:eq6} and the union bound, this follows if, for each $k=1,\ldots,K$,
\tightdisplay{
    \Pi
    \left(P:
    \rho_{\mathcal G}^k(\Prob_n,P)
    >
    \frac{C_n}{4}\sqrt{\frac vn}
    \,\middle|\,
    \breve{\mathcal D}_n
    \right)
    =
    o_{P_0}(1).
    \tag{C.7} \label{supplement:eq7}
}
The posterior bound~\eqref{supplement:eq7} for each $k$ is controlled by Steps~2--4 in the proof of Theorem~\ref{Theorem1}, with $\mathcal F$ replaced by $\mathcal F_k$. 
This completes the proof.
\end{proof}


\subsection{NBPL with Demeaned Outcomes}
\label{supplement:demeaned outcomes}

Following the demeaned EWM rule in \citet[Remark~2.3]{kitagawa2018should}, I introduce a demeaned version of NBPL. The motivation is that, although the NBPL rule $G^\star(P)\in\arg\max_{G\in\mathcal G} W(P;G)$ is invariant to positive rescaling of the outcome, it is generally not invariant to additive shifts. For instance, recoding a binary outcome from $Y \in \{0,1\}$ to $Y \in \{1,2\}$ may change the selected rule, which is undesirable in applied work. Demeaned NBPL eliminates this dependence on outcome coding: the selected rule is invariant to positive affine transformations of the outcome, $Y\mapsto aY+b$ with $a>0$.

To define the demeaned criterion, I use the augmented reduced-form data introduced in Section~\ref{supplement:alternative criteria}.\footnote{This augmentation is only notational: it makes the demeaning correction a functional of the posterior draw $P$, and the contraction-rate proof relies on the same empirical-process arguments as in Theorem~\ref{Theorem1}.} Let
\tightdisplay{
    \check D\coloneqq (Y,S_1,S_0,X),
    \qquad
    S_1\coloneqq \frac{T}{e(X)},
    \qquad
    S_0\coloneqq \frac{1-T}{1-e(X)},
}
where $\check{D} \in \check{\mathcal{D}} \coloneqq \mathcal Y \times \mathcal S_1\times\mathcal S_0\times\mathcal X$.
NBPL then places a Dirichlet process prior on the reduced-form distribution of $\check D$: $P\sim \mathrm{DP}(\alpha)$.

For a generic distribution $P$ of $\check D$, define
\tightalign{
    W(P;G)
    & \coloneqq
    \Exp_P
    \left[
    Y(S_1-S_0)\mathbf 1\{X\in G\}
    \right], \\
    A(P;G)
    & \coloneqq
    \Exp_P
    \left[
    (S_1-S_0)\mathbf 1\{X\in G\}
    \right].
}
The demeaned welfare criterion replaces $Y$ in $W(P;G)$ by its centered version $Y-\Exp_P[Y]$:
\tightdisplay{
    W^{\mathrm{dm}}(P;G)
    \coloneqq
    W(P;G)-\Exp_P[Y]\,A(P;G).
    \tag{C.8} \label{supplement:eq8}
}

Let
\tightdisplay{
    G_{\mathrm{dm}}^\star(P)
    \in
    \arg\max_{G\in\mathcal G} W^{\mathrm{dm}}(P;G),
    \qquad
    W_{\mathcal G, \mathrm{dm}}^\star(P)
    \coloneqq
    \sup_{G\in\mathcal G} W^{\mathrm{dm}}(P;G).
}
The corresponding posterior welfare regret under $P_0$ is
\tightdisplay{
    R_{\mathrm{dm}}(P_0;P)
    \coloneqq
    W_{\mathcal G}^\star(P_0)
    -
    W(P_0;G_{\mathrm{dm}}^\star(P)).
}
where $W_{\mathcal G}^\star(P_0) \coloneqq \sup_{G\in\mathcal G}W(P_0;G)$.

The following Assumption~\ref{Assumption:demeaned base measure} is the analogue of Assumption~\ref{Assumption:base measure} for demeaned outcomes.

\begin{assumption}
\label{Assumption:demeaned base measure}

The base measure $\alpha$ on $\check{\mathcal D} \coloneqq \mathcal Y \times \mathcal S_1\times\mathcal S_0\times\mathcal X$ has finite total mass $|\alpha| \coloneqq \alpha(\check{\mathcal D})<\infty$ and satisfies $\int_{\check{\mathcal D}} \{ |y| + |s_1| + |s_0| + |ys_1| + |ys_0| \} \,d\alpha(y,s_1,s_0,x)<\infty$. 

\end{assumption}

Theorem~\ref{Theorem:C2} shows that demeaning does not change the contraction rate in Theorem~\ref{Theorem1}.

\begin{theorem}
\label{Theorem:C2}
Suppose Assumptions~1, \ref{Assumption:demeaned base measure}, and 3 hold.
Let $\mathcal{\check{D}}_n \coloneqq \{\check{D}_i\}_{i=1}^n \stackrel{\mathrm{i.i.d.}}{\sim} P_0$
denote the observed sample, and let $\Pi(P \in \cdot\mid \mathcal{\check{D}}_n)$ denote the Dirichlet process posterior $P \mid \mathcal{\check{D}}_n \sim \mathrm{DP}(\alpha + \sum_{i=1}^n \delta_{\check{D}_i})$. Then, for every sequence $C_n\to\infty$ (arbitrarily slowly), as $n \to \infty$,
\tightdisplay{
    \Pi\!\left(P:
    R_{\mathrm{dm}}(P_0;P)
    \leq
    C_n\sqrt{\frac vn}
    \,\middle|\, \mathcal{\check{D}}_n
    \right)
    \stackrel{P_0}{\to}1.
    \tag{C.9} \label{supplement:eq9}
}
\end{theorem}

\begin{proof}
For any class of measurable functions $\mathcal F$, let $Pf \coloneqq \Exp_{P}[f]$, $\|P\|_{\mathcal{F}} \coloneqq \sup_{f \in \mathcal{F}}|Pf|$, and $\|f\|_{P,r} \coloneqq (\Exp_{P}|f|^r)^{1/r}$ for $r>0$.

\textbf{Step 1.}
Let $P$ and $Q$ be any two reduced-form distributions of $\check{D}$. Define
\tightdisplay{
    \rho_{\mathcal G,\mathrm{dm}}(P,Q)
    \coloneqq
    \sup_{G\in\mathcal G}
    |W^{\mathrm{dm}}(P;G)-W^{\mathrm{dm}}(Q;G)|.
}
Define two function classes $\mathcal F \coloneqq \{f_{G}(\check D): G\in\mathcal G\}$ and $\mathcal M \coloneqq \{m_G(\check D): G\in\mathcal G\}$, where
\tightdisplay{
    f_G(\check D) \coloneqq Y(S_1 - S_0)\mathbf 1\{X\in G\},
    \qquad
    m_G(\check D) \coloneqq (S_1 - S_0) \mathbf{1}\{X \in G\}.
}
The demeaned criterion can then be written as $W^{\mathrm{dm}}(P;G) = P f_G - (PY)P m_G$. 
By Assumption~\ref{Assumption:VC} and \citet[Lemma~A.1]{kitagawa2018supp}, both $\mathcal{F}$ and $\mathcal{M}$ are VC-subgraph classes with VC-dimension at most $v$. An envelope for $\mathcal{F}$ is $F(\check D) = |Y|(|S_1| + |S_0|)$, which satisfies $\|F\|_{P_0,2}<\infty$ and $\int_{\check{D}} F \, d\alpha<\infty$ under Assumptions~1 and~\ref{Assumption:demeaned base measure}.
Moreover, $\mathcal M$ has the constant envelope $\kappa^{-1}$ under Assumption~\ref{Assumption:DGP}.

For any two reduced-form distributions $P$ and $Q$,
\tightalign{
    & |W^{\mathrm{dm}}(P;G)-W^{\mathrm{dm}}(Q;G)| \\
    = \,\, &
    \left|
    (P-Q)f_G
    -
    \{(PY)P m_G - (QY)Q m_G \}
    \right| \\
    \leq \,\, &
    |(P-Q) f_G | +
    |(P-Q) Y|\,|P m_G |
    +
    |QY|\,|(P-Q)m_G|.
}
Since $|m_G(\check D)|\leq \kappa^{-1}$ for every $G\in\mathcal G$, taking supremum over $G$ on both sides gives
\tightdisplay{
    \rho_{\mathcal G,\mathrm{dm}}(P,Q)
    \leq
    \|P-Q\|_{\mathcal F}
    +
    \kappa^{-1}|(P-Q)Y|
    +
    |QY|\,\|P-Q\|_{\mathcal M}.
    \tag{C.10}\label{supplement:eq10}
}

As in the proof of Lemma~\ref{Lemma1},
\tightdisplay{
    R_{\mathrm{dm}}(P_0;P)
    \leq
    2\rho_{\mathcal G,\mathrm{dm}}(P_0,P)
    \leq
    2\{\rho_{\mathcal G,\mathrm{dm}}(P_0,\Prob_n)
    +
    \rho_{\mathcal G,\mathrm{dm}}(\Prob_n,P)\}.
    \tag{C.11} \label{supplement:eq11}
}

\textbf{Step 2.}
To control the term $\rho_{\mathcal G,\mathrm{dm}}(P_0,\Prob_n)$, apply \eqref{supplement:eq10} with $P=P_0$ and $Q=\Prob_n$:
\tightdisplay{
    \rho_{\mathcal G,\mathrm{dm}}(P_0,\mathbb{P}_n)
    \leq
    \|P_0 - \mathbb{P}_n\|_{\mathcal F}
    +
    \kappa^{-1}|(P_0 - \mathbb{P}_n)Y|
    +
    |\mathbb{P}_n Y| \cdot \|P_0 - \mathbb{P}_n\|_{\mathcal M}.
}
The first term on the right-hand side is controlled by the maximal inequality used in the proof of Lemma~\ref{Lemma2}, applied to the class $\mathcal F$. The second term is controlled by the same argument applied to the singleton class $\{Y\}$. The third term is controlled by applying the same maximal inequality to the class $\mathcal M$, together with the law of large numbers: $|\Prob_nY| \leq \Prob_n|Y| = O_{P_0}(1)$. Therefore,
\tightdisplay{
    \rho_{\mathcal G,\mathrm{dm}}(P_0,\Prob_n)
    =
    O_{P_0}\!\left(\sqrt{\frac vn}\right).
    \tag{C.12}\label{supplement:eq12}
}

\textbf{Step 3.}
Define $\mathcal B_n \coloneqq \left\{ \{\check{D}_i\}_{i=1}^n:\Prob_n|Y|\leq L_Y \right\}$, where $L_Y<\infty$ is a constant such that $L_Y>P_0|Y|$.
By the law of large numbers, $P_0^\infty(\mathcal B_n)\to 1$.

Define $\widetilde{\mathcal E}_n \coloneqq \left\{ P: R_{\mathrm{dm}}(P_0;P)>C_n\sqrt{v/n} \right\}$ and
\tightdisplay{
    \widetilde{\mathcal A}_n
    \coloneqq
    \left\{
    \{\check D_i\}_{i=1}^n:
    \rho_{\mathcal G,\mathrm{dm}}(P_0,\Prob_n)
    \leq
    \frac{C_n}{4}\sqrt{\frac vn}
    \right\}
    \cap
    \mathcal B_n.
}
The first event in the intersection has probability tending to one by \eqref{supplement:eq12}, and $P_0^\infty(\mathcal B_n)\to 1$. Therefore, $P_0^\infty(\widetilde{\mathcal A}_n)\to 1$.
By the regret decomposition~\eqref{supplement:eq11}, it remains to show
\tightdisplay{
    \Pi\left(
    P:
    \rho_{\mathcal G,\mathrm{dm}}(\Prob_n,P)
    >
    \frac{C_n}{4}\sqrt{\frac vn}
    \,\middle|\, \check\D_n
    \right)
    =
    o_{P_0}(1).
    \tag{C.13} \label{supplement:eq13}
}

\textbf{Step 4.}
On $\widetilde{\mathcal A}_n$, the event $\mathcal B_n$ holds, so applying \eqref{supplement:eq10} with $P=P$ and $Q=\Prob_n$ gives
\tightdisplay{
    \rho_{\mathcal G,\mathrm{dm}}(\Prob_n,P)
    \leq
    \|P-\Prob_n\|_{\mathcal F}
    +
    \kappa^{-1}|(P-\Prob_n)Y|
    +
    L_Y\|P-\Prob_n\|_{\mathcal M}.
}
Therefore, on $\widetilde{\mathcal A}_n$,
\tightalign{
    \Pi\left(
    P:
    \rho_{\mathcal G,\mathrm{dm}}(\Prob_n,P)
    >
    \frac{C_n}{4}\sqrt{\frac vn}
    \,\middle|\, \check\D_n
    \right) 
    \leq
    & \,\, 
    \Pi\left(
    P:
    \|P-\Prob_n\|_{\mathcal F}
    >
    \frac{C_n}{12}\sqrt{\frac vn}
    \,\middle|\, \check\D_n
    \right) \\
    \qquad & +
    \Pi\left(
    P:
    |(P-\Prob_n)Y|
    >
    \frac{\kappa C_n}{12}\sqrt{\frac vn}
    \,\middle|\, \check\D_n
    \right) \\
    \qquad & +
    \Pi\left(
    P:
    \|P-\Prob_n\|_{\mathcal M}
    >
    \frac{C_n}{12L_Y}\sqrt{\frac vn}
    \,\middle|\, \check\D_n
    \right).
}
The first posterior probability on the right-hand side is $o_{P_0}(1)$ by Steps~2--4 in the proof of Theorem~\ref{Theorem1}, applied to the class $\mathcal F$. 
The second is $o_{P_0}(1)$ by the same argument applied to the singleton class $\{Y\}$, and the third is $o_{P_0}(1)$ by the same argument applied to the class $\mathcal M$. 
Since $\kappa>0$, $L_Y<\infty$, and $C_n\to\infty$, the threshold sequences in all three displays diverge. This
proves~\eqref{supplement:eq13}, thus completing the proof.
\end{proof}


\subsection{NBPL with a Capacity Constraint}
\label{supplement:capacity constraint}

Policy learning is often subject to budgetary or capacity constraints \citep[e.g.,][]{bhattacharya2012inferring}. This section extends NBPL to settings in which at most a fraction $c\in(0,1)$ of the target population can be treated. 
I maintain the binary-treatment setup of Section~\ref{NBPL:setup} in the main text and impose the constraint through a capacity-adjusted welfare criterion.
For implementation details, see \citet[Appendix~C.3]{kitagawa2018supp}.

Let
\tightdisplay{
    \Gamma(D)
    \coloneqq
    \frac{YT}{e(X)}
    -
    \frac{Y(1-T)}{1-e(X)}.
}
For a generic reduced-form distribution $P$, let $P_X(G)\coloneqq P(X\in G)$. Following \citet{kitagawa2018should}, define the capacity-constrained welfare under rule $G$ as
\tightdisplay{
    W_c(P;G)
    =
    \frac{c}{\max\{c,P_X(G)\}}
    W(P;G), \qquad 
    W(P;G) \coloneqq \Exp_P\!\left[
    \Gamma(D)\mathbf 1\{X\in G\}
    \right].
    \tag{C.14} \label{supplement:eq14}
}
Thus, if $P_X(G)\leq c$, all units selected by $G$ can be treated. If $P_X(G)>c$, only a fraction $c/P_X(G)$ of the units selected by $G$ is randomly chosen for treatment.

NBPL places a Dirichlet process prior on the reduced-form distribution $P$: $\Pi \sim \mathrm{DP}(\alpha)$. For a
posterior draw $P$, let
\tightdisplay{
    G_c^\star(P)
    \in
    \arg\max_{G\in\mathcal G} W_c(P;G).
}
The capacity-constrained posterior welfare regret under $P_0$ is
\tightdisplay{
    R_c(P_0;P)
    \coloneqq
    W_{c,\mathcal G}^\star(P_0)
    -
    W_c(P_0;G_c^\star(P)),
}
where $W_{c,\mathcal G}^\star(P_0) \coloneqq \sup_{G\in\mathcal G} W_c(P_0;G)$.

Theorem~\ref{Theorem:C3} shows that incorporating capacity constraints does not change the posterior contraction rate in Theorem~\ref{Theorem1}.

\begin{theorem}
\label{Theorem:C3}
Suppose Assumptions 1--3 hold. Let $c\in(0,1)$ be fixed.
Let $\mathcal{D}_n \coloneqq \{D_i\}_{i=1}^n \stackrel{\mathrm{i.i.d.}}{\sim} P_0$ denote the observed sample, and let $\Pi(P \in \cdot \mid \D_n)$ denote the Dirichlet process posterior $P \mid \mathcal{D}_n \sim \mathrm{DP}(\alpha + \sum_{i=1}^n \delta_{D_i})$.
Then, for every sequence $C_n\to\infty$ (arbitrarily slowly), as $n \to \infty$,
\tightdisplay{
    \Pi\!\left(P:
    R_c(P_0;P)
    \leq
    C_n\sqrt{\frac vn}
    \,\middle|\,\D_n
    \right)
    \stackrel{P_0}{\to} 1.
    \tag{C.15} \label{supplement:eq15}
}
\end{theorem}

\begin{proof}
For any class of measurable functions $\mathcal F$, let $Pf \coloneqq \Exp_{P}[f]$, $\|P\|_{\mathcal{F}} \coloneqq \sup_{f \in \mathcal{F}}|Pf|$, and $\|f\|_{P,r} \coloneqq (\Exp_{P}|f|^r)^{1/r}$ for $r>0$.

\textbf{Step 1.}
Let
\[
    a_c(u)\coloneqq \frac{c}{\max\{c,u\}},
    \qquad u\in[0,1],
\]
so that $W_c(P;G)=a_c(P_X(G))W(P;G)$. The function $a_c$ is bounded above by one and is $1/c$-Lipschitz continuous \citep[Theorem D.1]{kitagawa2018online}. Hence, for any two reduced-form distributions $P$ and
$Q$,
\tightalign{
    & |W_c(P;G)-W_c(Q;G)| \\
    =  \,\, &
    |a_c(P_X(G))W(P;G)-a_c(Q_X(G))W(Q;G)| \\
    \leq \,\, & a_c(P_X(G)) \cdot |W(P;G)-W(Q;G)| + W(Q;G) \cdot |a_c(P_X(G)) - a_c(Q_X(G))|
    \\
    \leq \,\, &
    |W(P;G)-W(Q;G)|
    +
    \frac{|W(Q;G)|}{c}|P_X(G)-Q_X(G)|.
}
Since $|W(Q;G)|\leq Q|\Gamma|$, it follows that
\tightdisplay{
    |W_c(P;G)-W_c(Q;G)|
    \leq
    |(P-Q)[\Gamma\mathbf 1\{X\in G\}]|
    +
    \frac{Q|\Gamma|}{c}
    |(P-Q)[\mathbf 1\{X\in G\}]|.
    \tag{C.16}\label{supplement:eq16}
}

\textbf{Step 2.}
Define two functions classes $\mathcal{F} \coloneqq \{f_G(D): G \in \mathcal{G}\}$ and $\mathcal{H} \coloneqq \{h_G(D): G \in \mathcal{G}\}$, where
\tightdisplay{
    f_G(D) \coloneqq \Gamma(D) \mathbf{1}\{X \in G\}, \qquad 
    h_G(D) \coloneqq \mathbf{1}\{X \in G\} - \frac{1}{2}.
}
By Assumption~\ref{Assumption:VC} and \citet[Lemma~A.1]{kitagawa2018supp}, both $\mathcal F$ and $\mathcal H$ are VC-subgraph classes with VC-dimension at most $v$. Under Assumptions~1--2, $\mathcal F$ has an envelope $F$ satisfying $\|F\|_{P_0,2} < \infty$ and $\int_{\mathcal{D}} F \, d\alpha < \infty$. $\mathcal H$ has the constant envelope $1/2$.

For two reduced-form distributions $P$ and $Q$, define
\tightalign{
    \rho_{\mathcal G,c}(P,Q)
    & \coloneqq
    \sup_{G\in\mathcal G}|W_c(P;G)-W_c(Q;G)|, \\
    \rho_{\mathcal G}^{W}(P,Q)
    & \coloneqq
    \sup_{G\in\mathcal G}
    |(P-Q)[\Gamma\mathbf 1\{X\in G\}]| = \|P-Q\|_{\mathcal{F}}, \\
    \rho_{\mathcal G}^{X}(P,Q)
    & \coloneqq
    \sup_{G\in\mathcal G}
    |(P-Q)[\mathbf 1\{X\in G\}]| = \|P-Q\|_{\mathcal{H}}.
}
Then~\eqref{supplement:eq16} implies
\tightdisplay{
    \rho_{\mathcal G,c}(P,Q)
    \leq
    \rho_{\mathcal G}^{W}(P,Q)
    +
    \frac{Q|\Gamma|}{c}\rho_{\mathcal G}^{X}(P,Q).
    \tag{C.17}\label{supplement:eq17}
}

As in the proof of Lemma~\ref{Lemma1},
\tightdisplay{
    R_c(P_0;P)
    \leq
    2\rho_{\mathcal G,c}(P_0,P)
    \leq
    2\{\rho_{\mathcal G,c}(P_0,\Prob_n)
    +
    \rho_{\mathcal G,c}(\Prob_n,P)\}.
    \tag{C.18}\label{supplement:eq18}
}
Therefore, it remains to control the two terms on the right-hand side of \eqref{supplement:eq18}.

\textbf{Step 3.}
To control the term $\rho_{\mathcal G,c}(P_0,\mathbb{P}_n)$, apply~\eqref{supplement:eq17} with $P=\Prob_n$ and $Q=P_0$:
\tightdisplay{
    \rho_{\mathcal G,c}(P_0,\Prob_n)
    \leq
    \|\Prob_n - P_0\|_{\mathcal{F}}
    +
    \frac{P_0|\Gamma|}{c}
    \|\Prob_n - P_0 \|_{\mathcal{H}}.
}
Since $P_0|\Gamma| < \infty$ by Assumption~\ref{Assumption:DGP}, applying the VC maximal inequality used in the proof of Lemma~\ref{Lemma2} separately to $\mathcal F$ and $\mathcal H$ yields
\tightdisplay{
    \rho_{\mathcal G,c}(P_0,\Prob_n)
    =
    O_{P_0}\!\left(\sqrt{\frac{v}{n}}\right).
    \tag{C.19}\label{supplement:eq19}
}

\textbf{Step 4.}
Define $\mathcal B_n \coloneqq \left\{ \{D_i\}_{i=1}^n: \Prob_n|\Gamma|\leq L_\Gamma \right\}$, where $L_\Gamma<\infty$ is a constant such that $L_\Gamma>P_0|\Gamma|$. 
By the law of large numbers, $P_0^\infty(\mathcal B_n)\to 1$.

Define $\widetilde{\mathcal{E}}_n \coloneqq \{P: R_c(P_0;P) > C_n \sqrt{v/n} \}$ and
\tightdisplay{
    \widetilde{\mathcal{A}}_n
    \coloneqq
    \left\{ \{D_i\}_{i=1}^n:
    \rho_{\mathcal G,c}(P_0,\Prob_n)
    \leq
    \frac{C_n}{4}\sqrt{\frac vn}
    \right\}
    \cap
    \mathcal B_n.
}
The first event of the intersection has probability tending to one by \eqref{supplement:eq19}. Hence, $P_0^\infty(\widetilde{\mathcal{A}}_n)\to 1$.
By the regret bound~\eqref{supplement:eq18}, it remains to show
\tightdisplay{
    \Pi\left(P:
    \rho_{\mathcal G,c}(\Prob_n,P)
    >
    \frac{C_n}{4} \sqrt{\frac vn}
    \,\middle|\,\D_n
    \right)
    =
    o_{P_0}(1).
    \tag{C.20}\label{supplement:eq20}
}

\textbf{Step 5.}
To control the term $\rho_{\mathcal G,c}(\mathbb{P}_n, P)$, apply~\eqref{supplement:eq17} with $P=P$ and $Q=\Prob_n$:
\tightdisplay{
    \rho_{\mathcal G,c}(\Prob_n,P)
    \leq
    \|\Prob_n - P \|_{\mathcal{F}}
    +
    \frac{\Prob_n|\Gamma|}{c}
    \|\Prob_n - P \|_{\mathcal{H}}.
}
On $\widetilde{\mathcal{A}}_n$, one has $\rho_{\mathcal G,c}(\Prob_n,P) \leq \|\Prob_n - P \|_{\mathcal{F}} + (L_{\Gamma} / c) \|\Prob_n - P \|_{\mathcal{H}}$.
Therefore, on $\widetilde{\mathcal{A}}_n$,
\tightalign{
    \Pi\left(P:
    \rho_{\mathcal G,c}(\Prob_n,P)
    >
    \frac{C_n}{4} \sqrt{\frac vn}
    \,\middle|\,\D_n
    \right) 
    \leq
    \,\, & \Pi\left(P:
    \|\Prob_n - P \|_{\mathcal{F}}
    >
    \frac{C_n}{8}\sqrt{\frac vn}
    \,\middle|\,\D_n
    \right) \\
    & +
    \Pi\left(P:
    \|\Prob_n - P \|_{\mathcal{H}}
    >
    \frac{cC_n}{8L_\Gamma}\sqrt{\frac vn}
    \,\middle|\,\D_n
    \right).
}
The first posterior probability on the right-hand side is $o_{P_0}(1)$ by Steps~2--4 in the proof of Theorem~\ref{Theorem1}, applied to the class $\mathcal F$. The second posterior probability is $o_{P_0}(1)$ by the same argument applied to the class $\mathcal H$. Since $L_\Gamma<\infty$ and $c>0$ are fixed, the sequence $cC_n/(8L_\Gamma) \to\infty$. Therefore, \eqref{supplement:eq20} follows, completing the proof.
\end{proof}

\section{Proofs of Lemmas}
\label{supplement:proofs of lemmas}

\begin{proof}[Proof of Lemma~1]

Let $G^\star(P_0)\in\arg\max_{G\in\G}W(P_0;G)$ and
$G^\star(P)\in\arg\max_{G\in\G}W(P;G)$. By optimality of $G^\star(P)$ under $P$,
\tightalign{
    R(P_0;P)
    &=
    W(P_0;G^\star(P_0))-W(P_0;G^\star(P)) \\
    &=
    \{W(P_0;G^\star(P_0))-W(P;G^\star(P_0))\}
    +
    \{W(P;G^\star(P_0))-W(P;G^\star(P))\} \\
    &\qquad
    +
    \{W(P;G^\star(P))-W(P_0;G^\star(P))\} \\
    &\leq
    2\sup_{G\in\G}|W(P_0;G)-W(P;G)| \eqqcolon 2 \rho_{\mathcal{G}}(P_0,P).
}
By the triangle inequality, for any probability measures $Q_1,Q_2,Q_3$,
\tightdisplay{
    \rho_{\G}(Q_1,Q_3)
    \leq
    \rho_{\G}(Q_1,Q_2)+\rho_{\G}(Q_2,Q_3).
}
Applying this with $(Q_1,Q_2,Q_3)=(P_0,\Prob_n,P)$ gives
\tightdisplay{
    R(P_0;P)
    \leq
    2\{\rho_{\G}(P_0,\Prob_n)+\rho_{\G}(\Prob_n,P)\}.
}
\end{proof}

\bigskip
\begin{proof}[Proof of Lemma~2]

Let
\tightdisplay{
    \mathcal F \coloneqq \{f_G(D):G\in\G\},
    \qquad
    f_G(D) =
    \left\{
    \frac{YT}{e(X)} -
    \frac{Y(1-T)}{1-e(X)}
    \right\}
    \mathbf 1\{X\in G\}.
}
Then $\rho_{\G}(P_0,\Prob_n) = \|\Prob_n-P_0\|_{\mathcal F}$.
By \citet[Lemma~A.1]{kitagawa2018should}, $\mathcal F$ is a VC-subgraph class with VC-dimension at most $v \coloneqq \mathrm{VC}(\G)$. Under Assumption~\ref{Assumption:DGP}, $\mathcal F$ has an envelope $F$ satisfying $\|F\|_{P_0,2} < \infty$. 
By the VC covering bound and the uniform-entropy
maximal inequality \citep[Theorems~2.6.7 and~2.14.1]{vaart2023weak}, there exists a constant $K<\infty$ depending on $P_0$, such that 
\tightdisplay{
    \Exp_{P_0}
    \big[
    \rho_{\G}(P_0,\Prob_n)
    \big]
    =
    \Exp_{P_0}
    \|\Prob_n-P_0\|_{\mathcal F}
    \leq
    K\sqrt{\frac v n}.
}
By Markov's inequality,
\tightdisplay{
    P_0^\infty(\mathcal A_n^c)
    =
    P_0^\infty\!\left(
    \rho_{\G}(P_0,\Prob_n)>\frac{C_n}{4}\sqrt{\frac v n}
    \right)
    \leq
    \frac{
    4\Exp_{P_0}[\rho_{\G}(P_0,\Prob_n)]
    }{
    C_n\sqrt{v/n}
    }
    \leq
    \frac{4K}{C_n}.
}
Since $C_n\to\infty$, $P_0^\infty(\mathcal A_n^c)\to 0$. Hence
$P_0^\infty(\mathcal A_n)\to 1$.
\end{proof}

\bigskip
\begin{proof}[Proof of Lemma~3]

The result follows from the normalized Gamma process representation of the
Dirichlet process; see, for example,
\citet[Theorem~14.37 and Proposition~G.10]{ghosal2017fundamentals}.
\end{proof}

\bigskip
\begin{proof}[Proof of Lemma~4]

Let $F$ be an envelope of $\mathcal F$. By the triangle inequality, for every $f\in\mathcal F$,
\[
    \sqrt n V_n
    \left|
    Qf-\frac{\sum_{i=1}^n W_i f(D_i)}{\sum_{i=1}^n W_i}
    \right|
    \leq
    \sqrt n V_n \cdot QF
    +
    \sqrt n V_n \max_{1\leq i\leq n}F(D_i).
\]
Taking the supremum over $f\in\mathcal F$ gives
\tightdisplay{
    \sqrt n V_n
    \sup_{f\in\mathcal F}
    \left|
    Qf-\frac{\sum_{i=1}^n W_i f(D_i)}{\sum_{i=1}^n W_i}
    \right|
    \leq
    \sqrt n V_n \cdot QF
    +
    \sqrt n V_n \max_{1\leq i\leq n}F(D_i).
}

It remains to show that, for $P_0^\infty$-almost every realization of $\{D_i\}_{i \geq 1}$, both terms on the right-hand side are $o_{\mathbb Q_n}(1)$. 
Since $V_n\sim\mathrm{Beta}(|\alpha|,n)$,
\tightdisplay{
    \Exp_{\mathbb Q_n}[V_n]
    =
    \frac{|\alpha|}{|\alpha|+n},
    \qquad
    \Var_{\mathbb Q_n}(V_n)
    =
    \frac{|\alpha|n}{(|\alpha|+n)^2(|\alpha|+n+1)}.
}
Hence $V_n=O_{\mathbb Q_n}(n^{-1})$ and
$\sqrt n V_n=o_{\mathbb Q_n}(1)$. Moreover, $\Exp_{\mathbb Q_n}[QF] = |\alpha|^{-1} \int F\,d \alpha < \infty$ \citep[Theorem 3;][]{ferguson1973bayesian},
so $QF=O_{\mathbb Q_n}(1)$. Therefore, $\sqrt n V_n \cdot QF=o_{\mathbb Q_n}(1)$.

For the second term, since $P_0 F^2 <\infty$, \citet[Lemma 5]{ray2021bernstein} implies
\tightdisplay{
    \frac{1}{\sqrt n}\max_{1\leq i\leq n}F(D_i)\to 0,
    \qquad
    P_0^\infty\text{-almost surely.}
}
Thus, for $P_0^\infty$-almost every realization of $\{D_i\}_{i \geq 1}$,
\tightdisplay{
    \sqrt n V_n \max_{1\leq i\leq n}F(D_i)
    =
    (nV_n)
    \left\{
    \frac{1}{\sqrt n}\max_{1\leq i\leq n}F(D_i)
    \right\}
    =
    O_{\mathbb Q_n}(1) \cdot o(1)
    =
    o_{\mathbb Q_n}(1).
}
Combining the two bounds proves the claim.
\end{proof}

\bigskip
\begin{proof}[Proof of Lemma~5]

By \citet[Lemma 3.7.7]{vaart2023weak}, for any $1 \leq n_0 \leq n$
\tightalign{
    \Exp_{\xi} \left\| \frac{1}{\sqrt{n}} \sum_{i=1}^n \xi_i Z_i \right\|_{\mathcal{F}} 
    \leq & \,\, 2(n_0-1) \frac{1}{n} \sum_{i=1}^n \left\| Z_i\right\|_{\mathcal{F}} \Exp_{\xi} \left( \max_{1 \leq i \leq n} \frac{|\xi_i|}{\sqrt{n}} \right) \\
    & + 
    2 \lVert \xi_1 \rVert_{2,1} \cdot \max_{n_0 \leq k \leq n} \Exp_{R} \left\lVert \frac{1}{\sqrt{k}} \sum_{i=n_0}^k Z_{R_i} \right\rVert_{\F}.
}
Setting $n_0=1$ completes the proof.
\end{proof}

\bigskip
\begin{proof}[Proof of Lemma~6]

By Hoeffding's inequality \citep[Theorem 4]{hoeffding1963probability}, the left-hand side of~\eqref{appendix:eq5} increases if the sample without replacement $(D_{R_1}, \ldots, D_{R_n})$ is replaced by a sample with replacement $(\widehat{D}_1, \ldots, \widehat{D}_n)$.


\end{proof}

\bigskip
\begin{proof}[Proof of Lemma~7]

The argument follows the Poissonization and multiplier step in the proof of \citet[Theorem~3.7.3]{vaart2023weak}. 
Fix $n \in \mathbb{N}$ and $1\leq k\leq n$, and condition on $(D_1,\ldots,D_n)$. Let $\{\widetilde N_i\}_{i \geq 1}$ be an i.i.d. sequence of symmetrized Poisson variables with parameter $k/(2n)$. Let $\{\varepsilon_i\}_{i \geq 1}$ be an i.i.d. sequence of Rademacher variables independent of $\{\widetilde N_i\}_{i\geq 1}$ and $\{D_i\}_{i\geq 1}$.

By the Poissonization inequality \citep[Lemma~3.7.6]{vaart2023weak},
\tightdisplay{
    \Exp_{\widehat D}
    \left\|
    \widehat{\mathbb G}_{n,k}
    \right\|_{\mathcal F}
    \leq
    4\Exp_{\widetilde N}
    \left\|
    \frac{1}{\sqrt k}
    \sum_{i=1}^n \widetilde N_i\delta_{D_i}
    \right\|_{\mathcal F}.
}
Since each $\widetilde N_i$ is symmetric and
$\varepsilon_i$ is a Rademacher variable independent of $\widetilde N_i$, 
$\widetilde N_i\stackrel{d}{=}\varepsilon_i|\widetilde N_i|$. Hence, on the right-hand side, $\widetilde N_i$ may be replaced by $\varepsilon_i|\widetilde N_i|$.
Applying the multiplier inequality 
(Lemma 5) with $\xi_i=|\widetilde N_i|$, $Z_i=\varepsilon_i\delta_{D_i}$, and $(D_1, \ldots, D_n)$ fixed gives
\tightdisplay{
    \Exp_{\widehat D}
    \left\|
    \widehat{\mathbb G}_{n,k}
    \right\|_{\mathcal F}
    \leq
    8 \sqrt{\frac n k}\,
    \|\widetilde N_1\|_{2,1}
    \max_{1\leq j\leq n}
    \Exp_{\varepsilon,R}
    \left\|
    \frac{1}{\sqrt j}
    \sum_{i=1}^j \varepsilon_i\delta_{D_{R_i}}
    \right\|_{\mathcal F},
}
where $R=(R_1,\ldots,R_n)$ is uniformly distributed over permutations of
$\{1,\ldots,n\}$ and is independent of the data and the Rademacher variables.
By \citet[Problem~3.7.3]{vaart2023weak},
$\sqrt{n/k}\,\|\widetilde N_1\|_{2,1}\leq 2\sqrt2$ for all $n,k$. Hence
\tightdisplay{
    \Exp_{\widehat D}
    \left\|
    \widehat{\mathbb G}_{n,k}
    \right\|_{\mathcal F}
    \leq
    16 \sqrt2
    \max_{1\leq j\leq n}
    \Exp_{\varepsilon,R}
    \left\|
    \frac{1}{\sqrt j}
    \sum_{i=1}^j \varepsilon_i\delta_{D_{R_i}}
    \right\|_{\mathcal F}.
}

Define $U_j \coloneqq \Exp_\varepsilon \left\| j^{-1/2} \sum_{i=1}^j \varepsilon_i\delta_{D_i} \right\|_{\mathcal F}$. Let $\mathcal S_n$ denote the $\sigma$-field generated by all measurable functions of $\{D_i\}_{i \geq 1}$ that are symmetric in their first $n$ coordinates.
By exchangeability,
\tightdisplay{
    \Exp_{\varepsilon,R}
    \left\|
    \frac{1}{\sqrt j}
    \sum_{i=1}^j \varepsilon_i\delta_{D_{R_i}}
    \right\|_{\mathcal F}
    =
    \Exp(U_j\mid \mathcal S_n).
}
Therefore,
\tightdisplay{
    \max_{1\leq k\leq n}
    \Exp_{\widehat D}
    \left\|
    \widehat{\mathbb G}_{n,k}
    \right\|_{\mathcal F}
    \leq
    16 \sqrt2\,
    \Exp\!\left(
    \sup_{j\geq 1}U_j
    \,\middle|\,\mathcal S_n
    \right).
}

By \citet[Corollary~2.9.9]{vaart2023weak}, the variable $\sup_{j\geq1}U_j$ is integrable.
The symmetric $\sigma$-fields $\mathcal S_n$ decrease to
$\mathcal S\coloneqq\bigcap_{n\geq1}\mathcal S_n$. Hence, by the reverse
martingale convergence theorem,
\tightdisplay{
    \Exp\!\left(
    \sup_{j\geq 1}U_j
    \,\middle|\,\mathcal S_n
    \right)
    \to
    \Exp\!\left(
    \sup_{j\geq 1}U_j
    \,\middle|\,\mathcal S
    \right),
    \qquad P_0^\infty\text{-a.s.}
}
By the Hewitt--Savage zero-one law, $\mathcal S$ is trivial, so $\Exp\!\left( \sup_{j\geq 1}U_j \,\middle|\,\mathcal S \right)  = \Exp\!\left(\sup_{j\geq 1}U_j\right)$.
Consequently,
\tightdisplay{
    \limsup_{n\to\infty}
    \max_{1\leq k\leq n}
    \Exp_{\widehat D}
    \left\|
    \widehat{\mathbb G}_{n,k}
    \right\|_{\mathcal F}
    \leq
    16 \sqrt2\,
    \Exp\!\left(\sup_{j\geq 1}U_j\right),
    \qquad P_0^\infty\text{-a.s.}
}
This proves the claim.
\end{proof}

\bigskip
\begin{proof}[Proof of Lemma~8]

Let $\F_0\coloneqq\F \cup \{0\}$. Then $\|\cdot\|_{\mathcal F_0}=\|\cdot\|_{\mathcal F}$ and $\mathrm{VC}(\mathcal F_0) \leq v + 1$. Since $\sqrt{v+1}\leq\sqrt{2v}$ for $v\geq1$, replacing $\F$ by $\F_0$ affects only the universal constants in the bounds below. Therefore, we may assume without loss of generality that $0\in\F$.

Fix $j\in\mathbb N$ and condition on $(D_1,\ldots,D_j)$. For $f\in\mathcal F$, define
\tightdisplay{
    X_j(f) \coloneqq j^{-1/2} \sum_{i=1}^j \varepsilon_i f(D_i).
}
By Hoeffding's inequality \citep[Lemma 2.2.8]{vaart2023weak}, $X_j$ has sub-Gaussian increments with respect to the
empirical $L_2(\Prob_j)$-semimetric
\tightdisplay{
    d_j(f,g)
    \coloneqq
    \|f-g\|_{\Prob_j,2}
    =
    \left\{
    \frac1j\sum_{i=1}^j(f(D_i)-g(D_i))^2
    \right\}^{1/2}.
}
That is, for every $x>0$,
\tightdisplay{
    \Prob_\varepsilon
    \left(
    |X_j(f)-X_j(g)| > x
    \right)
    \leq
    2\exp\!\left\{
    -\frac{x^2}{2d_j^2(f,g)}
    \right\}.
}

Since $0\in\mathcal F$ and $X_j(0)=0$, and since
$d_j(f,0)\leq \|F\|_{\Prob_j,2}$ for every $f\in\mathcal F$, the maximal inequality for sub-Gaussian processes
\citep[Corollary~2.2.9]{vaart2023weak} gives
\tightdisplay{
    U_j
    =
    \Exp_\varepsilon \left(\sup_{f\in\mathcal F}|X_j(f)-X_j(0)| \right)
    \leq
    K
    \int_0^{\|F\|_{\Prob_j,2}}
    \sqrt{\log N(\eta,\mathcal F,L_2(\Prob_j))}\,d\eta ,
}
where $K<\infty$ is a universal constant. 
If $\|F\|_{\Prob_j,2}=0$, then $U_j=0$ and the desired bound is immediate. Otherwise, changing variables $\eta=\zeta\|F\|_{\Prob_j,2}$ gives
\tightdisplay{
    U_j
    \leq
    K\|F\|_{\Prob_j,2}
    \int_0^1
    \sqrt{
    \log N(\zeta\|F\|_{\Prob_j,2},\mathcal F,L_2(\Prob_j))
    }\,d\zeta.
    \tag{D.1} \label{supplement:eq21}
}

By the VC covering-number bound for VC-subgraph classes
\citep[Theorem~2.6.7]{vaart2023weak}, since $\mathcal F$ has VC-dimension $v$,
\tightdisplay{
    N(\zeta\|F\|_{\Prob_j,2},\mathcal F,L_2(\Prob_j))
    \leq
    K_0 v(16e)^v
    \left(\frac1\zeta\right)^{2v},
    \qquad 0<\zeta<1,
}
for a universal constant $K_0<\infty$. Taking logarithms gives
\tightalign{
    \log N(\zeta\|F\|_{\Prob_j,2},\mathcal F,L_2(\Prob_j))
    &\leq
    \log K_0+\log v+v\log(16e)+2v\log(1/\zeta) \\
    &\leq
    C_0v\{1+\log(1/\zeta)\},
}
where $C_0<\infty$ is a universal constant; the last inequality uses $v\geq1$ and the uniform boundedness of $(\log K_0+\log v)/v$ over $v\geq1$. 
Moreover, $\log(1/\zeta)\geq0$ implies
\tightdisplay{
    \sqrt{C_0v\{1+\log(1/\zeta)\}}
    \leq
    C_1\sqrt v \left\{1+\sqrt{\log(1/\zeta)} \right\},
}
for the universal constant $C_1=\sqrt{C_0}$.
Substituting into the entropy integral \eqref{supplement:eq21} gives
\tightdisplay{
    U_j \leq
    KC_1\|F\|_{\Prob_j,2}\sqrt v
    \int_0^1
    \left\{1+\sqrt{\log(1/\zeta)} \right\} \,d\zeta.
}
Since $\int_0^1\sqrt{\log(1/\zeta)}\,d\zeta<\infty$, there exists a universal
constant $C_2<\infty$ such that
\tightdisplay{
    U_j \leq
    C_2\|F\|_{\Prob_j,2}\sqrt v.
}

Taking expectations over the data $(D_1, \ldots, D_j) \sim P_0^{\otimes j}$ and using Jensen's inequality,
\tightdisplay{
    \Exp U_j
    \leq
    C_2\sqrt v\,\Exp\|F\|_{\Prob_j,2}
    \leq
    C_2\sqrt v
    \left\{
    \Exp \left( \frac1j\sum_{i=1}^j F(D_i)^2 \right)
    \right\}^{1/2}
    =
    C_2\sqrt v\,\|F\|_{P_0,2}.
}
Therefore,
\tightdisplay{
    \sup_{j\geq1}\Exp U_j \leq C_2\|F\|_{P_0,2}\sqrt v.
    \tag{D.2} \label{supplement:eq22}
}

It remains to control $\Exp \big( \sup_{j\geq 1} U_j \big)$. 
Apply the maximal inequality for Rademacher averages over increasing sample sizes \citep[Corollary~2.9.9]{vaart2023weak} with $Z_i=\delta_{D_i}$ and $\xi_i = \varepsilon_i$. 
Since $\|Z_i\|_{\mathcal F} \leq F(D_i)$ and $P_0 F^2<\infty$, this gives
\tightdisplay{
    \Exp\!\left(\sup_{j\geq1}U_j\right)
    \leq
    C_3\sup_{j\geq1}\Exp U_j
    +
    C_3\|F\|_{P_0,2},
}
for a universal constant $C_3<\infty$. Combining this bound with \eqref{supplement:eq22} yields
\tightdisplay{
    \Exp\!\left(\sup_{j\geq1}U_j\right)
    \leq
    C\sqrt v,
}
where $C<\infty$ depends on $P_0$ only through $\|F\|_{P_0,2}$. This proves the claim.
\end{proof}

\end{document}